\documentclass[11pt]{article}
\usepackage{graphicx}
\usepackage{color}
\usepackage{amsmath, amssymb, amsthm, ascmac, caption, comment, enumerate, ascmac, bm, newtxtext, mdframed}
\usepackage{latexsym, mathabx, mathrsfs, mathtools, psfrag, setspace, subcaption, booktabs}
\usepackage{natbib}
\usepackage{clipboard}
 
\mathtoolsset{showonlyrefs=true}

\definecolor{mycolor}{cmyk}{0.9, 0.3, 0.2, 0}
\usepackage[colorlinks=true, linkcolor=mycolor, filecolor=mycolor, urlcolor=mycolor, citecolor=mycolor, driverfallback=dvipdfmx]{hyperref}
\usepackage[margin = 2.3cm]{geometry}

\usepackage{titlesec}
\titleformat*{\section}{\Large\bfseries\sffamily}
\titleformat*{\subsection}{\large\bfseries\sffamily}
\titleformat*{\subsubsection}{\large\bfseries\sffamily}

\DeclareMathOperator*{\argmin}{argmin}

\renewcommand{\hat}{\widehat}
\renewcommand{\tilde}{\widetilde}
\renewcommand{\check}{\widecheck}
\renewcommand{\bar}{\overline}
\newcommand{\bbE}{\mathbb{E}}

\newcommand{\bbR}{\mathbb{R}}

\newcommand{\calU}{\mathcal{U}}

\newcommand{\calF}{\mathcal{F}}

\numberwithin{equation}{section}
 
\allowdisplaybreaks[2]

\theoremstyle{definition}
\newtheorem{theorem}{Theorem}[section]
\newtheorem{assumption}{Assumption}[section]

\newtheorem{example}{Example}[section]
\newtheorem{lemma}{Lemma}[section]
\newtheorem{proposition}{Proposition}[section]
\newtheorem{remark}{Remark}[section]

\title{Network-Adjusted GMM Estimation under Network Uncertainty}

\author{
    Tadao Hoshino
    \thanks{School of Political Science and Economics, Waseda University, 1-6-1 Nishi-waseda, Shinjuku-ku, Tokyo 169-8050, Japan. Email: \href{mailto:thoshino@waseda.jp}{thoshino@waseda.jp}. \\
    I thank \'{A}ureo de Paula, Wenjie Wang, Zhenlin Yang, and participants at seminars at NTU and SMU for their helpful comments and discussions.
    This work was supported by JSPS KAKENHI Grant Number 26K04852.
    }
}

\begin{document}
\maketitle

\begin{abstract}
    This paper proposes a network-adjusted generalized method of moments (NA-GMM) estimator for social interaction models when the observed network may differ from the true interaction network.
    NA-GMM is a novel penalized GMM approach that allows the elements of the observed interaction matrix to be modified to improve the fit of the moment conditions.
    To avoid unrestricted network adjustments, the NA-GMM criterion introduces a penalty on the amount of adjustment.
    Since NA-GMM does not aim to estimate the true interaction network itself, the estimator generally converges to a pseudo-true parameter.
    For a linear spatial autoregressive model, we prove that the NA-GMM estimator is consistent for the pseudo-true parameter and is asymptotically normally distributed under general moment misspecification.
    We also prove that a fixed-weight version of the NA-GMM estimator has a desirable bias reduction property relative to conventional GMM without network adjustment.
    An empirical application to U.S. county-level COVID-19 infection data demonstrates the usefulness of the proposed method.
\end{abstract}

\clearpage

\section{Introduction}

Social interactions play a crucial role in human behavior.
Knowledge diffusion among workers, peer effects among schoolmates, and the transmission of infectious diseases across communities are only a few examples in which accounting for social interactions is essential.
However, in most empirical applications, researchers do not have precise information about the underlying interaction structure.
For example, friendship networks among students are often collected directly through surveys, but individuals may differ in how they report friendships, and some links may be missing.
Moreover, even when such a survey-based network is available, it is not clear whether it is the relevant network for the social interaction effect under study.
In spatial interaction models, researchers often construct a network from administrative adjacency or geographic distance.
However, different choices of distance measures, cutoffs, or normalizations can lead to different networks, and it is often difficult to justify that a particular choice is appropriate.
More generally, even if the true network were observed at a particular time point, there may exist unobserved heterogeneity in how units perceive the strength of each link across different situations.

Since incomplete network data and network errors are common and often unavoidable in practice, a growing literature has developed statistical inference methods that are robust to such uncertainty.
One approach is to estimate a specific network formation model using partially observed network data and simulate relevant moment conditions to estimate the parameters of interest (e.g., \citealp{chandrasekhar2011econometrics, reeves2024model, boucher2025estimating}).
Another approach, which is available even when no network information is observed, is to directly estimate the interaction matrix nonparametrically using large panel data (e.g., \citealp{de2025identifying}).
\citet{lewbel2023social} also do not require prior knowledge of the network structure, but their methodology relies on the availability of many small networks.
By contrast, \citet{lewbel2024ignoring} show that conventional two-stage least squares (2SLS) estimators can remain valid even when network errors are ignored, provided that the observed network contains only a sufficiently small amount of measurement error.
Although we do not discuss it here, there is also a growing literature in causal inference that studies treatment spillovers under incomplete network data.

As demonstrated by these studies, estimating social interaction models under network uncertainty typically requires substantial additional sources of identification and restrictions, such as specific network formation models, large panel data, many small networks, or sufficiently small network errors.
However, in many empirical situations, such additional information is not readily available.
Thus, the key question of this paper is: what can researchers do when such information is unavailable but network uncertainty is still a concern, other than naively applying conventional estimators that ignore network errors?
To address this question, we propose a novel penalized generalized method of moments (GMM) approach that allows the observed network to be modified through optimization.

If the model specification is correct and the moment conditions are valid under the true interaction network, network errors may result in violations of the moment conditions.
For example, in standard linear spatial (or network) autoregressive (SAR) models, researchers often employ a 2SLS estimator with spatially lagged covariates as instrumental variables (IVs) for the endogenous interaction effect term.
In this case, the validity of the moment conditions can be checked empirically, for example, by using classical Sargan over-identification test.
Then, if we reject the over-identification test, we may suspect the presence of network errors, provided that the model is correctly specified and the IVs used are valid.
This naturally leads to a new class of moment-based estimators, which we call network-adjusted GMM (NA-GMM), that adjust the network itself simultaneously with the estimation of model parameters.
Since each network is typically a high dimensional object, allowing it to be modified freely leads to an identification problem due to too many degrees of freedom.
Thus, the NA-GMM method introduces a penalty term for the size of the network correction.

Conceptually, the proposed NA-GMM method is closely related to the generalized empirical likelihood (GEL) method and the optimally-transported GMM (OT-GMM) developed by \citet{schennach2026optimally}.
Similar to this paper, these studies consider situations in which the model may fail to pass over-identification tests.
As for the source of this moment mismatch, GEL considers biased sampling from the target population, and OT-GMM considers general measurement errors in the data.
These approaches aim to improve the moment fit by allowing such distortions to be corrected in the estimation process.
NA-GMM builds on the same motivation as theirs.
However, we attribute the moment mismatch to an imprecisely observed network.

Another closely related literature is GMM estimation with misspecified moment conditions (e.g., \citealp{hall2003large, hansen2021inference, kleibergen2025double, kleibergen2025testing}).
Although the NA-GMM procedure can update the observed network to improve the moment fit, the resulting modified network is generally different from the true network.
Consequently, NA-GMM estimates of the model parameters are also inconsistent for the true parameter, and they generally converge to corresponding pseudo-true values.
The aforementioned studies on misspecified GMM investigate the convergence of estimators to pseudo-true values and asymptotic normality around those values.
We establish similar results for NA-GMM.

\begin{flushleft}
\textbf{Paper Outline.}
\end{flushleft}

This paper proceeds as follows.
Section \ref{sec:model} presents the basic framework of NA-GMM.
We first define network errors as differences between the true interaction matrix and the matrix used by the researcher, and focus on models in which these errors enter the moment conditions linearly.
This class of models includes standard linear SAR models and linear treatment spillover models.
We then introduce an uncertainty set, a set of entries of the interaction matrix that are suspected to be affected by network errors, and define the NA-GMM estimator as the solution to a joint optimization problem over the model parameters and corrections to the entries in the uncertainty set.
Because unlimited network modifications can create an identification problem when the uncertainty set is large, the NA-GMM criterion includes a quadratic penalty for the size of the network modification.
Thanks to the linearity of network errors and the quadratic cost function, we can show that the optimal network correction has a closed-form solution, and the NA-GMM criterion reduces to a simple continuous-updating (CU) GMM-type criterion in which the model parameters enter the GMM weight matrix.

In Section \ref{sec:SAR}, we focus on the bias reduction property of NA-GMM in linear SAR models.
We show that the NA-GMM weight matrix downweights particular directions of the moment conditions that are more severely affected by network errors.
In addition, as a simplified alternative to the original NA-GMM estimator, we consider a fixed-weight version of NA-GMM.
For this fixed-weight estimator, we prove that the estimator has a smaller worst-case bias than the naive GMM estimator without network adjustment.

We next turn to the asymptotic properties of NA-GMM in Section \ref{sec:asymp}.
We first establish consistency of the NA-GMM estimator for the pseudo-true parameter under high-level conditions and derive asymptotic normality under local misspecification, where the moment misspecification caused by network errors is only local.
Meanwhile, as pointed out by \cite{hall2003large}, under general moment misspecification, the convergence of the GMM weight matrix can affect the limiting distribution of the GMM estimator, and this significantly complicates the asymptotic analysis.
Therefore, to investigate the distribution of the NA-GMM estimator under non-local misspecification, we narrow our attention only to linear SAR models and prove its asymptotic normality under more primitive conditions.
This result also shows that performing statistical inference based directly on the asymptotic distribution is difficult in practice, because the asymptotic variance depends on the unknown true parameter and true interaction matrix.
We therefore suggest using NA-GMM as a diagnostic tool for assessing the robustness of GMM estimates to network uncertainty and provide several diagnostic analyses.

Sections \ref{sec:MC} and \ref{sec:empir} provide numerical illustrations of NA-GMM.
Section \ref{sec:MC} reports Monte Carlo simulations.
The simulation results show that, across various experimental setups, NA-GMM generally achieves a smaller average bias than the naive GMM estimator that ignores network errors, confirming its desirable bias reduction property.
The fixed-weight NA-GMM estimator also has a slightly smaller bias than the naive GMM estimator, as predicted by the theory, although it does not perform as good as the original NA-GMM estimator.
Importantly, when there is no network error, all GMM estimators are accurate.
This suggests that NA-GMM can be used as a robust alternative to conventional GMM.
The empirical application in Section \ref{sec:empir} analyzes U.S. county-level COVID-19 infection data.
The naive GMM estimate indicates a large positive spatial spillover effect of COVID-19 infection.
Applying NA-GMM to the same model substantially improves the moment fit, while the estimate of the spatial effect changes only slightly.
This finding suggests that the estimated spatial effect is relatively robust to the specification of the interaction network.

Section \ref{sec:conclude} concludes.

\begin{flushleft}
\textbf{Notation.}
\end{flushleft}

For an integer $a$, we write $[a] \coloneqq \{1,2,\ldots,a\}$.
For a vector $x$, $\|x\|$ denotes the Euclidean norm.
For a symmetric matrix $A$, $\lambda_{\min}(A)$ and $\lambda_{\max}(A)$ denote its minimum and maximum eigenvalues, respectively.
For a $p \times q$ matrix $M = (m_{ij})$, define $\|M\| \coloneqq \left(\sum_{i=1}^p \sum_{j=1}^q m_{ij}^2\right)^{1/2}$, $\|M\|_{\mathrm{op}} \coloneqq \left\{\lambda_{\max}(M^\top M)\right\}^{1/2}$, $\|M\|_\infty \coloneqq \max_{1 \le i \le p} \sum_{j=1}^q |m_{ij}|$, and $\|M\|_1 \coloneqq \max_{1 \le j \le q} \sum_{i=1}^p |m_{ij}|$.
Finally, $c$, possibly with subscripts, denotes a generic positive constant whose value may vary from line to line.

\section{Network-Adjusted GMM Estimator}\label{sec:model}

\subsection{Setup}

Suppose that we observe data of size $n$: $\{(W_i, A_{i1}^{\mathrm{obs}}, \ldots, A_{in}^{\mathrm{obs}}) : i \in [n]\}$.
Here, $W_i$ is a vector of observed variables, including response variables and covariates, and $A_{ij}^{\mathrm{obs}}$ is the $(i,j)$-th element of the adjacency matrix $\bm A_n^{\mathrm{obs}}$ that represents the observed social network.
The true social network is denoted by $\bm A_n = (A_{ij})_{i,j \in [n]}$, and it may or may not be fully observed.

Typically, network effects are determined not directly through $\bm A_n$, but through its transformed version $\bm G_n = (G_{ij})_{i,j \in [n]}$; for example, $G_{ij} = \left(\sum_{k \neq i} A_{ik}\right)^{-1} A_{ij}$.
The observed counterpart of $\bm G_n$ is denoted by $\bm G_n^{\mathrm{obs}} = (G_{ij}^{\mathrm{obs}})_{i,j \in [n]}$.
The transformation from $\bm A_n$ to $\bm G_n$ may be unknown and heterogeneous across individuals, whereas that from $\bm A_n^{\mathrm{obs}}$ to $\bm G_n^{\mathrm{obs}}$ is predetermined by the researcher.
For both $\bm G_n$ and $\bm G_n^{\mathrm{obs}}$, we assume that the diagonal elements are zero.
In what follows, both of these matrices are treated as non-random.\footnote{
    This treatment does not necessarily rule out stochastic network formation or random network errors.
    Rather, all subsequent arguments should be interpreted as being conditional on the realized true and observed interaction matrices.
    }

Under this setup, the target parameter of interest is $\theta_0 \in \bbR^{d_\theta}$, which satisfies
\begin{align}\label{eq:momeq}
        &\bbE[\bar \mu_n(\bm G_n; \theta_0)] = \bm 0 \\
        &\text{where} \quad \bar \mu_n(\bm G_n; \theta) \coloneqq \frac{1}{n}\sum_{i \in [n]}  \mu_i(\bm W_n, \bm G_n; \theta),
\end{align}
$\bm W_n = (W_1, \ldots, W_n)^\top$, and $\mu_i(\bm W_n, \bm G_n; \theta)$ is an $\bbR^{d_\mu}$-valued moment function.
In what follows, we focus on over-identified cases, where $d_\mu > d_\theta$.

If $\bm G_n$ is observable, we can construct a GMM estimator based on \eqref{eq:momeq} to estimate $\theta_0$.
However, in our setting, $\bm G_n$ is not available; instead, only $\bm G_n^{\mathrm{obs}}$ is available, so that, in general, 
\begin{align*}
    \bbE[\bar \mu_n(\bm G_n^{\mathrm{obs}}; \theta_0)] \neq \bm 0.
\end{align*}
We define the network error matrix as 
\begin{align}
    \bm D_n \coloneqq \bm G_n - \bm G_n^{\mathrm{obs}},
\end{align}
whose $(i,j)$-th element is given by $D_{ij} = G_{ij} - G_{ij}^{\mathrm{obs}}$ for $i \neq j$ and $D_{ii} = 0$ for all $i \in [n]$.
For the sources of network errors, we can consider two types: measurement error in the observed network $\bm A_n^{\mathrm{obs}}$ and misspecification of the transformation used to construct $\bm G_n^{\mathrm{obs}}$.
We cover both types of errors and refer to them simply as network errors.

We often have prior knowledge that $G_{ij}^{\mathrm{obs}} = G_{ij}$, or equivalently $D_{ij} = 0$, for certain pairs $(i,j)$.
For example, if students in different grades do not interact, or if residents in distant neighborhoods have no opportunity to communicate, then it may be reasonable to set $G_{ij}^{\mathrm{obs}} = G_{ij} = 0$ when $i$ and $j$ belong to different communities.
We define the set of such pairs as $\calF_n$.
Since $\mathrm{diag}(\bm D_n) = 0$ is always satisfied, $(i,i) \in \calF_n$ for all $i \in [n]$.
In addition, we define the \textit{uncertainty set} $\calU_n \coloneqq \{(i,j): i,j \in [n]\} \setminus \calF_n$ as the set of pairs $(i,j)$ for which the researcher is uncertain about the value of $G_{ij}$.
Note that $(i,j) \in \calF_n$ implies $D_{ij} = 0$, but the converse is not true in general.
Accordingly, the uncertainty set $\calU_n$ must contain the support of the true network errors.
In practice, researchers may want to specify $\calU_n$ conservatively; however, setting $\calU_n$ too large may reduce estimation efficiency and violate some of the theoretical conditions imposed below.

Hereafter, when there is no confusion, we drop the subscript $n$ from $\calF_n$ and $\calU_n$.
Let $\calU(i) \coloneqq \{ j \in [n] : \: (i,j) \in \calU \}$ be the $i$-th uncertainty set and let $m_i \coloneqq |\calU(i)|$ be its size.
To facilitate the analysis, we assume throughout that the moment function is linear in the network error matrix in the following sense.

\begin{assumption}[Linearity in the network error]\label{as:linearmom}
For all $i \in [n]$, the moment function $\mu_i$ is decomposed as
\begin{align}
    \mu_i(\bm W_n, \bm G_n; \theta) = q_i(\bm W_n, \bm G_n^\text{obs}; \theta) + Z_i \bm V_i(\bm W_n, \calU; \theta)^\top \bm D_{\calU(i)},
\end{align}
where $Z_i$ is a $d_\mu \times 1$ vector of IVs, $\bm D_{\calU(i)}$ is the $m_i \times 1$ subvector of the $i$-th row of $\bm D_n$ corresponding to the elements of $\calU(i)$, $q_i(\bm W_n, \bm G_n^\text{obs}; \theta)$ is a $d_\mu \times 1$ vector-valued function of $\theta$, and $\bm V_i(\bm W_n,\calU;\theta) = (V_j(\bm W_n; \theta))_{j \in \calU(i)}$ is an $m_i \times 1$ vector-valued function of $\theta$.
\end{assumption}

Under Assumption \ref{as:linearmom}, we can write $\bar\mu_n(\bm G_n; \theta) = \bar q_n(\theta) + \bar u_n(\vec{\bm D}_\calU; \theta)$, where 
\begin{align}
    \bar q_n(\theta) \coloneqq \frac{1}{n} \sum_{i \in [n]} q_i(\bm W_n, \bm G_n^\text{obs}; \theta), \quad
    \bar u_n(\vec{\bm d}_m;\theta) \coloneqq \frac{1}{n} \sum_{i \in [n]} Z_i \bm V_i(\bm W_n, \calU; \theta)^\top \bm d_{m,i},
\end{align}
and $\vec{\bm D}_\calU \coloneqq (\bm D_{\calU(1)}^\top, \ldots, \bm D_{\calU(n)}^\top)^\top$.
Here, $m \coloneqq \sum_{i \in [n]} m_i$, $\vec{\bm d}_m \coloneqq (\bm d_{m,1}^\top, \ldots, \bm d_{m,n}^\top)^\top \in \bbR^m$, and $\bm d_{m,i} = (d_{ij})_{j \in \calU(i)}$ is a generic $m_i \times 1$ vector.  
Throughout, we assume that $m > 0$.

\medskip

Two specific examples that fit into our framework are given below.

\begin{example}[Linear SAR models]\label{ex:sar}
    Suppose that the outcome $Y_i$ is determined by the following linear SAR model:
    \begin{align}
        Y_i = \alpha_0 \sum_{j \neq i} G_{ij} Y_j + X_i^\top \beta_0 + \varepsilon_i,
    \end{align}
    where $X_i$ is a vector of exogenous regressors.
    To consistently estimate the parameter $\theta_0 = (\alpha_0, \beta_0^\top)^\top$, we address the endogeneity of the autoregressive term $\sum_{j \neq i} G_{ij} Y_j$ using an IV method.
    Assume that $\sum_{j \neq i} G_{ij}^{\mathrm{obs}} X_j$ provides a valid set of IVs for $\sum_{j \neq i} G_{ij} Y_j$.
    Note that this assumption does not require the observed network to coincide with or approximate the true network well; it rules out the case in which the observed network is entirely uninformative about the true peer exposure.
    Then, letting $Z_i = (\sum_{j \neq i} G_{ij}^{\mathrm{obs}} X_j^\top, X_i^\top)^\top$ and $W_i = (Y_i, Z_i^\top)^\top$, we can obtain \eqref{eq:momeq} and satisfy Assumption \ref{as:linearmom} by setting
    \begin{align}
    \mu_i(\bm W_n, \bm G_n; \theta)
    & = Z_i \left( Y_i - \alpha \sum_{j \neq i} G_{ij} Y_j - X_i^\top \beta \right) \\
    q_i(\bm W_n, \bm G_n^\text{obs}; \theta)
    & = Z_i \left( Y_i - \alpha \sum_{j \neq i} G_{ij}^\text{obs} Y_j - X_i^\top \beta \right), 
    \quad V_j(\bm W_n; \theta) = - \alpha Y_j.
\end{align}
\end{example}

\begin{example}[Treatment spillover models]\label{ex:treat}
    Consider a linear treatment spillover model of the following form:
    \begin{align}
        Y_i = \beta_{01} T_i + \beta_{02} \sum_{j \neq i} G_{ij} T_j + X_i^\top \beta_{03} + \varepsilon_i.
    \end{align}
    Here, both the individual treatment $T_i$ and the treatment spillover $\sum_{j \neq i} G_{ij} T_j$ may be endogenous.
    Let $Z_{1i}$ be a set of IVs for $T_i$, and construct IVs for the spillover term as $\sum_{j \neq i} G_{ij}^{\mathrm{obs}} Z_{1j}$.
    Then, for $\theta = (\beta_1, \beta_2, \beta_3^\top)^\top$, letting $Z_i = (Z_{1i}^\top, \sum_{j \neq i} G_{ij}^{\mathrm{obs}} Z_{1j}^\top, X_i^\top)^\top$ and $W_i = (Y_i, T_i, Z_i^\top)^\top$, we can obtain \eqref{eq:momeq} and satisfy Assumption \ref{as:linearmom} by setting
    \begin{align}
    \mu_i(\bm W_n, \bm G_n; \theta)
    & = Z_i \left( Y_i - \beta_1 T_i - \beta_2 \sum_{j \neq i} G_{ij} T_j - X_i^\top \beta_3 \right) \\
    q_i(\bm W_n, \bm G_n^\text{obs}; \theta)
    & = Z_i \left( Y_i - \beta_1 T_i - \beta_2 \sum_{j \neq i} G_{ij}^\text{obs} T_j - X_i^\top \beta_3 \right), \quad
    V_j(\bm W_n; \theta) = - \beta_2 T_j.
    \end{align}
\end{example}

\subsection{NA-GMM estimator}

We now introduce our network-adjusted GMM (NA-GMM) estimator.
Define the NA-GMM criterion function as follows:
\begin{align}
    Q_{n,\rho}(\theta)
    \coloneqq
    \inf_{\vec{\bm d}_m} \left\{ \left\| \bar q_n(\theta) + \bar u_n(\vec{\bm d}_m; \theta) \right\|_{\Omega_n}^2 + \frac{\rho}{m} \| \vec{\bm d}_m \|^2 \right\},
\end{align}
where $\|x \|_{\Omega_n}^2 \coloneqq x^\top \Omega_n x$ and $\Omega_n \in \bbR^{d_\mu \times d_\mu}$ is a non-stochastic, symmetric, positive definite GMM weight matrix.
Then, the NA-GMM estimator is defined by
\begin{align}\label{eq:nagmm}
        \hat \theta_{n,\rho} = \argmin_{\theta \in \Theta} Q_{n,\rho}(\theta),
\end{align}
where $\Theta \subset \bbR^{d_\theta}$ is a compact parameter set.
The criterion function is designed to account for moment mismatch caused by network errors.
The sample moments based on the observed network, $\bar q_n(\theta)$, may not be close to zero because $\bm G_n^{\mathrm{obs}}$ is mismeasured, and $\bar u_n(\vec{\bm d}_m; \theta)$ represents an adjustment term to improve the fit of moment conditions.

Introducing a penalty term $(\rho/m) \| \vec{\bm d}_m \|^2$ is crucial.
If the network adjustments were unrestricted (i.e., $\rho = 0$), the moment conditions could be made nearly satisfied for arbitrary $\theta$, which leads to a weak identification issue for the NA-GMM estimator.
In our criterion, the penalty term measures the quadratic cost of adjusting the network.
Hence, $Q_{n,\rho}(\theta)$ becomes small only when the observed-network moment conditions can be corrected by a relatively small amount of network adjustments.

Note that the adjusted network obtained through the inner optimization should not be interpreted as an estimate of the true interaction network $\bm G_n$.
It is chosen only to improve the moment fit under the penalty.
Therefore, under general network errors, the NA-GMM estimator should be interpreted as estimating a pseudo-true parameter.
That said, in many empirically relevant situations, the NA-GMM method is expected to produce an estimate with smaller bias than the naive GMM estimator without network adjustment:
\begin{align}\label{eq:naive}
    \hat \theta_n^{\text{naive}} = \argmin_{\theta \in \Theta} \left\| \bar q_n(\theta) \right\|_{\Omega_n}^2.
\end{align}
For more specific discussions on this topic in the context of SAR models, see Section \ref{sec:SAR}.

Recall that we have assumed that $d_\mu > d_\theta$.
Under just-identification in which the sample moment equation $\bar q_n(\theta) = \bm 0$ has a solution, the inner minimization can be achieved by $\vec{\bm d}_m = 0$.
In this case, the estimator reduces to the just-identified naive GMM estimator based on $\bm G_n^\text{obs}$, and the network adjustment plays no role.
The over-identification assumption rules out this trivial situation.

\begin{remark}[Relation to \cite{schennach2026optimally}]
    This construction of the GMM criterion is closely related to the OT-GMM developed by \cite{schennach2026optimally}.
    In their framework, moment mismatch is interpreted as arising from errors in the observed variables, and their method is to optimally perturb the data distribution so that the moment conditions can be restored.
    By contrast, in our framework, moment mismatch arises from errors in the observed interaction network.
    Accordingly, while NA-GMM keeps the observed variables fixed, it adjusts the network components directly.
\end{remark}

To compute the NA-GMM estimator, we need to solve a profiled minimization problem in which the inner minimization is with respect to a potentially high-dimensional vector $\vec{\bm d}_m$.
Fortunately, this inner minimization problem has a closed-form solution.
To see this, for each $i \in [n]$, define
\begin{align}
    \bm H_{m,i}(\theta)
    \coloneqq
    \frac{1}{n}  Z_i \bm V_i(\bm W_n,\calU;\theta)^\top  \in \bbR^{d_\mu \times m_i},
\end{align}
and let $\bm H_m(\theta) = \left( \bm H_{m,1}(\theta), \ldots, \bm H_{m,n}(\theta) \right) \in \bbR^{d_\mu \times m}$.
Then, we have $\bar u_n(\vec{\bm d}_m;\theta) = \sum_{i \in [n]} \bm H_{m,i}(\theta)\bm d_{m,i} = \bm H_m(\theta)\vec{\bm d}_m$.
Hence, the NA-GMM criterion function can be written as 
\begin{align}\label{eq:quad}
    Q_{n,\rho}(\theta) = \inf_{\vec{\bm d}_m} \left\{\left\| \bar q_n(\theta) + \bm H_m(\theta)\vec{\bm d}_m \right\|_{\Omega_n}^2 + \frac{\rho}{m}  \| \vec{\bm d}_m \|^2 \right\}.
\end{align}

\begin{remark}[Quasi-Bayesian interpretation]\label{rem:quasiBayes}
From a quasi-Bayesian perspective, for a given $\rho > 0$, suppose that $\theta$ has a uniform prior on $\Theta$ and that the network adjustment vector has the Gaussian prior
\begin{align}
    \vec{\bm d}_m \mid \rho \sim N\left( \bm 0, \frac{m}{n \rho} I_m \right).
\end{align}
Following \citet{chernozhukov2003mcmc}, define the quasi-posterior
\begin{align}
    p_{n,\rho}( \theta, \vec{\bm d}_m )
    &\propto \exp\left\{ -\frac{n}{2} \left\| \bar q_n(\theta) + \bar u_n( \vec{\bm d}_m; \theta ) \right\|_{\Omega_n}^2 \right\} \pi_\Theta(\theta) \pi_{n,\rho}( \vec{\bm d}_m ) \\
    &\propto \bm 1\{ \theta \in \Theta \} \exp\left[ -\frac{n}{2} \left\{ \left\| \bar q_n(\theta) + \bm H_m(\theta)\vec{\bm d}_m \right\|_{\Omega_n}^2 + \frac{\rho}{m} \| \vec{\bm d}_m \|^2 \right\} \right],
\end{align}
where $\pi_\Theta$ and $\pi_{n,\rho}$ denote the corresponding prior densities.
Then, in view of \eqref{eq:quad}, we can see that the NA-GMM method is to seek a mode of this quasi-posterior.
In particular, we can interpret that $\rho$ governs the strength of the prior belief that the network adjustment vector is close to zero.
The relationship between the zero-mean Gaussian prior and the quadratic penalty is analogous to the Bayesian interpretation of ridge regression (see, e.g., Chapter 11 of \citealp{murphy2022probabilistic}).
Alternative priors on the network adjustment vector would lead to different regularization structures for NA-GMM.
Exploring such alternatives is an interesting topic for future research.
\end{remark}

The expression \eqref{eq:quad} shows that the network adjustment term $\vec{\bm d}_m$ enters the criterion function quadratically.
Due to this quadratic structure, the minimization problem with respect to $\vec{\bm d}_m$ can be solved exactly.
We formally state this result in the following lemma.
\begin{lemma}\label{lem:NA-GMM-obj}
    For $\rho > 0$, the NA-GMM criterion function can be expressed as
\begin{align}\label{eq:NA-GMM-obj2}
    Q_{n,\rho}(\theta) = \left\| \bar q_n(\theta) \right\|^2_{\Psi_{n,\rho}(\theta)},
\end{align}
where
\begin{align}
    \Psi_{n,\rho}(\theta) \coloneqq \left( \Omega_n^{-1} + \frac{m}{\rho}\bm H_m(\theta) \bm H_m(\theta)^\top \right)^{-1}.
\end{align}
\end{lemma}
The proof is by direct calculation; see Appendix \ref{app:NA-GMM-obj}.

\bigskip

Using the Woodbury matrix identity, we obtain
\begin{align}
    \Omega_n - \Psi_{n,\rho}(\theta)
    = \Omega_n \bm H_m(\theta) \left( \bm H_m(\theta)^\top \Omega_n \bm H_m(\theta) + \frac{\rho}{m} I_m \right)^{-1} \bm H_m(\theta)^\top \Omega_n,
\end{align}
which is positive semidefinite.
Hence, compared with the naive GMM in \eqref{eq:naive}, the NA-GMM criterion modifies the weight matrix by downweighting components associated with $\bm H_m(\theta)$.
In view of \eqref{eq:NA-GMM-obj2}, the amount of downweighting is governed by the reciprocal of $\rho$ and the eigenvalues of $\bm H_m(\theta)\bm H_m(\theta)^\top$.
That is, directions associated with larger eigenvalues of this matrix are downweighted more heavily, and a smaller value of $\rho$ strengthens this downweighting.
Note that, under our setup,
\begin{align}
    \bar q_n(\theta_0) \approx - \bar u_n(\vec{\bm D}_{\calU};\theta_0) = - \bm H_m(\theta_0)\vec{\bm D}_{\calU},
\end{align}
where the approximation becomes equality in expectation.
Thus, when the interaction matrix is misspecified, the resulting violation of the moment conditions approximately lies in the column space of $\bm H_m(\theta_0)$.
The NA-GMM downweights such violations, and therefore may be less sensitive than the naive GMM estimator to network uncertainty.

As $\rho$ becomes smaller, the downweighting becomes stronger.
At the same time, one should not choose $\rho$ to be too small.
In an extreme case, if one sets $\rho = 0$, then the criterion function is reduced to $Q_{n,0}(\theta) = \inf_{\vec{\bm d}_m} \| \bar q_n(\theta) + \bm H_m(\theta)\vec{\bm d}_m \|_{\Omega_n}^2$.
Because $\vec{\bm d}_m$ typically has a large number of degrees of freedom, this may lead to a weak identification issue for $\theta$.
Thus, there is a trade-off in the choice of $\rho$: a smaller $\rho$ makes the estimator less sensitive to violations of the moment conditions, whereas a larger $\rho$ leads to stronger identification.

\begin{remark}[Alternative NA-GMM estimators]
    The NA-GMM criterion function in \eqref{eq:NA-GMM-obj2} is similar to that of the continuous-updating (CU) GMM estimator (e.g., \citealp{hansen1996finite,kleibergen2025double}) in that the GMM weight matrix depends on $\theta$ and changes during optimization.
    However, this feature can generate multiple local minima and make the optimization computationally difficult.
    A computationally simpler alternative is the fixed-weight NA-GMM estimator:
    \begin{align}\label{eq:FW-NA-GMM}
        \hat \theta_{n, \rho}^{\mathrm{fw}}
        =
        \argmin_{\theta \in \Theta}
        \left\| \bar q_n(\theta) \right\|_{\Psi_{n,\rho}(\theta^\dagger)}^2,
    \end{align}
    where $\theta^\dagger$ is a fixed vector specified by any method.
    Then, once $\rho$ is given, its computation is as easy as that of conventional GMM.

    Another alternative is the iterated GMM-type estimator:
    \begin{align}
        \hat \theta_{n, \rho}^{(t+1)}
        =
        \argmin_{\theta \in \Theta}
        \left\| \bar q_n(\theta) \right\|_{\Psi_{n,\rho}(\hat \theta_{n,\rho}^{(t)})}^2,
    \end{align}
    where $\hat \theta_{n,\rho}^{(0)}$ is an initial estimator obtained by any method.
    The iteration stops when a certain convergence criterion is met.
    Convergence of such iterations generally requires additional contractivity conditions (\citealp{hansen2021inference}).

    Both estimators are computationally more attractive than the original NA-GMM.
    In linear models such as the SAR model in Section \ref{sec:SAR}, they also admit closed-form expressions.
    Moreover, in that case, it is relatively straightforward to verify a desirable bias reduction property relative to the naive GMM; see Proposition \ref{prop:op-norm}.

    Under correct model specification, CU-GMM and iterated GMM are first-order asymptotically equivalent (\citealp{newey2004higher}), and fixed-weight one-step GMM is consistent for the same true parameter.
    Under moment misspecification, however, they generally converge to different limits because different objective functions are associated with different pseudo-true values.
    Consequently, the choice of initial values, such as $\theta^\dagger$ or $\hat \theta_{n,\rho}^{(0)}$, may affect the final estimate in a nontrivial way.
\begin{comment}
    Finally, in general, a small discrepancy among these GMM estimates should not be interpreted as evidence of small network errors.
    Such a discrepancy only indicates that the effect of network errors on the moment mismatch is limited.
    Since network errors affect the estimator through $\bm H_m(\theta)\vec{\bm D}_{\calU}$, different network errors may produce similar, or even negligible, effects on the resulting estimates.
\end{comment}
\end{remark}

\section{Bias Reduction Property in SAR Models}\label{sec:SAR}

\subsection{NA-GMM estimator for SAR models}

In this section, with a particular focus on linear SAR models, we investigate the bias reduction property of the NA-GMM estimators.
Specifically, we consider the same model as in Example \ref{ex:sar}:
\begin{align}
        Y_i = \alpha_0 \sum_{j \neq i} G_{ij} Y_j + X_i^\top \beta_0 + \varepsilon_i,
\end{align}
with $Z_i = (\sum_{j \neq i} G_{ij}^{\mathrm{obs}} X_j^\top, X_i^\top)^\top$.
As we have seen above, the moment function can be written as
\begin{align}
    \mu_i(\bm W_n, \bm G_n; \theta) = q_i(\bm W_n,\bm G_n^\text{obs}; \theta) + Z_i  \bm V_i(\bm W_n, \calU; \theta)^\top \bm D_{\calU(i)}
\end{align}
where $\theta = (\alpha, \beta^\top)^\top$, $q_i(\bm W_n, \bm G_n^\text{obs}; \theta) = Z_i ( Y_i - \alpha \sum_{j \neq i} G_{ij}^\text{obs} Y_j  - X_i^\top \beta )$, $\bm V_i(\bm W_n, \calU; \theta) = -\alpha \bm Y_{\calU(i)}$, and $\bm Y_{\calU(i)} = (Y_j)_{j \in \calU(i)}$.
Hence, $\bm H_{m,i}(\theta) = n^{-1} Z_i \bm V_i(\bm W_n,\calU;\theta)^\top = -(\alpha/n) Z_i \bm Y_{\calU(i)}^\top$.
Moreover, let
\begin{align}
    \Gamma_n \coloneqq \frac{1}{n}\left(Z_1 \bm Y_{\calU(1)}^\top, \ldots , Z_n \bm Y_{\calU(n)}^\top\right) \in \bbR^{d_\mu \times m}
\end{align}
so that we can write $\bm H_m(\theta) =  - \alpha \Gamma_n$ and $\bar u_n(\vec{\bm d}_m; \theta) = - \alpha \Gamma_n \vec{\bm d}_m$.
Then, the NA-GMM criterion function can be written as
\begin{align}
    Q_{n,\rho}(\theta) = \| \bar q_n(\theta) \|^2_{\Psi_{n,\rho}(\alpha)}
\end{align}
where
\begin{align}
    \Psi_{n,\rho}(\alpha) \coloneqq \left(\Omega_n^{-1} + \frac{m \alpha^2}{\rho} \Phi_n \right)^{-1},
\end{align}
and $\Phi_n \coloneqq \Gamma_n \Gamma_n^\top$.

To simplify the discussion, let us assume $\Omega_n = I_{d_\mu}$ hereinafter in this section.
As remarked in the previous section, this criterion function is a downweighted GMM criterion in certain directions, compared with the naive GMM.
To see this more clearly, let the eigenvalue decomposition of $\Phi_n$ be
\begin{align}
    \Phi_n = P_n \Lambda_n P_n^\top,
\end{align}
where $\Lambda_n = \mathrm{diag}(\lambda_{n,1},\ldots,\lambda_{n,d_\mu})$, and $\lambda_{n,k} \ge 0$ for all $k \in [d_\mu]$.
Then
\begin{align}
    I_{d_\mu} + \frac{m \alpha^2}{\rho} \Phi_n
    = P_n \left( I_{d_\mu} + \frac{m \alpha^2}{\rho} \Lambda_n \right) P_n^\top,
\end{align}
and hence
\begin{align}
    \Psi_{n,\rho}(\alpha) = \left(I_{d_\mu} + \frac{m \alpha^2}{\rho} \Phi_n\right)^{-1} 
    & = P_n \left( I_{d_\mu} + \frac{m \alpha^2}{\rho} \Lambda_n \right)^{-1} P_n^\top \\
    & = P_n \mathrm{diag} \left( \frac{\rho}{\rho + m\alpha^2\lambda_{n,1}}, \ldots, \frac{\rho}{\rho + m\alpha^2\lambda_{n,d_\mu}} \right) P_n^\top.
\end{align}
Thus, $\Psi_{n,\rho}(\alpha)$ and $\Phi_n$ share the same eigenvectors, and each eigen-direction is shrunk by the factor $\rho/(\rho + m\alpha^2\lambda_{n,k})$.
Directions associated with larger eigenvalues of $\Phi_n$ are therefore downweighted more strongly.

Recall that, in the SAR model, moment misspecification arises through the term $-\bm H_m(\theta_0) \vec{\bm D}_{\calU} = \alpha_0 \Gamma_n \vec{\bm D}_{\calU}$.
Thus, directions associated with larger eigenvalues of $\Phi_n = \Gamma_n\Gamma_n^\top$ are precisely those in which network errors can generate larger moment violations.
The NA-GMM weight matrix $\Psi_{n,\rho}(\alpha)$ downweights these directions, whereas the naive GMM with $\Omega_n=I_{d_\mu}$ assigns equal weight to all directions.

\subsection{Fixed-weight NA-GMM estimator}

It is difficult to analyze the bias reduction property of the original CU-type estimator in depth because it does not have a closed-form expression.
Therefore, here we instead investigate the fixed-weight version of the NA-GMM estimator.

Write
\begin{align}
    \bar q_n(\theta)
    & = \frac{1}{n}\bm Z_n^\top \left(\bm Y_n - \alpha \bm G_n^\text{obs} \bm Y_n - \bm X_n \beta \right)
    = \frac{1}{n}\bm Z_n^\top \left( \bm Y_n - \bm R_n^\text{obs} \theta \right),
\end{align}
where $\bm Y_n = (Y_1, \ldots, Y_n)^\top$, $\bm X_n = (X_1, \ldots, X_n)^\top$, $\bm Z_n = (Z_1, \ldots, Z_n)^\top$, and $\bm R_n^\text{obs} = (\bm G_n^\text{obs} \bm Y_n,\bm X_n)$.
Let $\bm R_n = (\bm G_n \bm Y_n, \bm X_n)$ and $\bm \varepsilon_n = (\varepsilon_1, \ldots, \varepsilon_n)^\top$ so that $\bm Y_n = \bm R_n \theta_0 + \bm \varepsilon_n$.
Then,
\begin{align}
    \bar q_n(\theta)
    & = \frac{1}{n}  \left[ \bm Z_n^\top \left\{ \bm R_n^\text{obs} (\theta_0 - \theta) + (\bm R_n - \bm R_n^\text{obs})\theta_0 + \bm \varepsilon_n \right\} \right] \\
    & = \Pi_n (\theta_0 - \theta) + \bm b_n,
\end{align}
where $\Pi_n \coloneqq n^{-1} \bm Z_n^\top \bm R_n^\text{obs}$, and $\bm b_n \coloneqq n^{-1} \bm Z_n^\top (\alpha_0 \bm D_n \bm Y_n + \bm \varepsilon_n) = \alpha_0 \Gamma_n \vec{\bm D}_\calU + n^{-1} \bm Z_n^\top \bm \varepsilon_n$.

Define, for some $\kappa > 0$,
\begin{align}\label{eq:fw_weight}
    \Psi_{n, \kappa}
    \coloneqq
    \left(I_{d_\mu} + \frac{m}{\kappa}\Phi_n\right)^{-1}
\end{align}
and 
\begin{align}
    \hat \theta_{n, \kappa}^{\text{fw}}
    & \coloneqq \argmin_{\theta \in \Theta} \left(\Pi_n (\theta_0 - \theta) + \bm b_n \right)^\top \Psi_{n, \kappa} \left(\Pi_n (\theta_0 - \theta) + \bm b_n \right) \\
    & = \theta_0 + \left( \Pi_n^\top \Psi_{n, \kappa} \Pi_n \right)^{-1}
    \Pi_n^\top \Psi_{n, \kappa} \bm b_n,
\end{align}
assuming that $\Pi_n^\top \Psi_{n, \kappa} \Pi_n$ is non-singular.
As $\kappa \to \infty$, the weight matrix $\Psi_{n, \kappa}$ collapses to an identity matrix, leading to the naive GMM estimator; that is, $\hat \theta_{n}^{\text{naive}} = \hat \theta_{n,\infty}^{\text{fw}}$.
Write
\begin{align}\label{eq:biasFW}
    \hat \theta_{n, \kappa}^\text{fw} - \theta_0
    & = L_n(\kappa) \vec{\bm D}_\calU +  \left( \Pi_n^\top \Psi_{n, \kappa} \Pi_n \right)^{-1}
    \Pi_n^\top \Psi_{n, \kappa} \left( \frac{1}{n}\bm Z_n^\top \bm \varepsilon_n \right),
\end{align}
where
\begin{align}
    L_n(\kappa)
    \coloneqq
    \alpha_0 \left( \Pi_n^\top \Psi_{n, \kappa} \Pi_n \right)^{-1} \Pi_n^\top \Psi_{n, \kappa} \Gamma_n.
\end{align}
From \eqref{eq:biasFW}, we can clearly see that the bias component of the fixed-weight NA-GMM comes from the network errors transformed through the matrix $L_n(\kappa)$.
Since the second term of \eqref{eq:biasFW} is typically of order $O_P(n^{-1/2})$ for any $\kappa > 0$, the relative performance of the fixed-weight NA-GMM compared with the naive GMM estimator is essentially governed by the first term.
The next proposition shows that $L_n(\kappa)$ generally becomes larger in operator norm as $\kappa$ increases.

\begin{proposition}\label{prop:op-norm}
    Suppose that $\Pi_n^\top \Psi_{n, \kappa} \Pi_n$ is non-singular for any $\kappa > 0$.
    Then, $\|L_n(\kappa)\|_{\mathrm{op}}$ is weakly increasing in $\kappa$.
\end{proposition}

\begin{remark}[An implication of Proposition \ref{prop:op-norm}]
The monotonicity result in Proposition \ref{prop:op-norm} does not imply that the fixed-weight NA-GMM estimator uniformly reduces the bias for every realization of network errors as $\kappa \to 0$.
For a given $\vec{\bm D}_\calU$, the bias can exceed that of the naive GMM in some directions.
Nevertheless, the result implies a sharp comparison in a worst-case sense.
Indeed, noting that $\sup_{\|\vec{\bm d}_m\|\le c} \|L_n(\kappa)\vec{\bm d}_m \| = c \|L_n(\kappa)\|_{\mathrm{op}}$, Proposition \ref{prop:op-norm} gives
\begin{align}
    \sup_{\|\vec{\bm d}_m\|\le c}
    \left\| L_n(\kappa)\vec{\bm d}_m \right\| \le \sup_{\|\vec{\bm d}_m\|\le c} \left\| L_n(\infty)\vec{\bm d}_m \right\|.
\end{align}
Thus, although the fixed-weight NA-GMM estimator may not uniformly reduce the bias for all realizations of network errors, it achieves at least asymptotically a smaller worst-case bias than the naive GMM.

For the original CU-type NA-GMM, such a direct comparison is not available due to the dependence of the weight matrix on $\alpha$.
Nevertheless, the form of the weight matrix suggests a similar bias reduction mechanism.
Indeed, as shown in Section \ref{sec:MC}, numerical experiments indicate that the CU-type NA-GMM generally achieves a smaller bias than the fixed-weight estimator.
\end{remark}

\section{Asymptotic Properties of the NA-GMM Estimator}\label{sec:asymp}

In this section, we investigate the asymptotic properties of the NA-GMM estimator.
We first consider a general NA-GMM estimator in a local misspecification setting under some high-level conditions.
We then study the asymptotic properties of the estimator in SAR models under general moment misspecification and provide more primitive conditions.

\subsection{Asymptotic results in a general setting}

\renewcommand{\theassumption}{\thesection.1.\arabic{assumption}}
\setcounter{assumption}{0}

Since the NA-GMM estimator is defined based on misspecified moment conditions, its probability limit generally does not coincide with the true $\theta_0$ (\citealp{hall2003large,hansen2021inference}).
To characterize the limiting behavior of the NA-GMM estimator, we define the \textit{pseudo-true} parameter value as follows:
\begin{align}
    \theta_{n,\rho}^*
    & \coloneqq \argmin_{\theta \in \Theta} Q_{n,\rho}^*(\theta), 
\end{align}
where $Q_{n,\rho}^*(\theta) \coloneqq \| \varphi_n(\theta) \|^2_{\Psi_{n, \rho}^*(\theta)}$, $\varphi_n(\theta) \coloneqq \bbE[\bar q_n(\theta)]$, and
\begin{align}
    \Psi_{n, \rho}^*(\theta) \coloneqq \left(\Omega_n^{-1} + \frac{m}{\rho} \bbE [\bm H_m(\theta) \bm H_m(\theta)^\top] \right)^{-1}.
\end{align}
We first show that, under the following set of conditions, the NA-GMM estimator $\hat \theta_{n, \rho}$ converges in probability to $\theta_{n,\rho}^*$.
 
\begin{assumption}[Identification]\label{as:ident}
    $Q_{n,\rho}^*(\theta)$ is continuous on $\Theta$, and for each $\epsilon>0$,
    \begin{align}
    \inf_{\theta \in \Theta: \: \| \theta - \theta_{n,\rho}^*\| \ge \epsilon} \{ Q_{n,\rho}^*(\theta) - Q_{n,\rho}^*(\theta_{n,\rho}^*)\} \ge c_\epsilon > 0
    \end{align}
    for all sufficiently large $n$.
\end{assumption}

\begin{assumption}\label{as:IV}
    $\{Z_i\}$ are non-stochastic and uniformly bounded in $i \in [n]$ for $n \ge 1$.
\end{assumption}

\begin{assumption}\label{as:omega}
    $\Omega_n \in \bbR^{d_\mu \times d_\mu}$ is a symmetric, non-stochastic matrix that satisfies $0 < \underline{c}_\Omega \le \lambda_{\min}(\Omega_n) \le \lambda_{\max}(\Omega_n) \le \bar c_\Omega < \infty$ for all sufficiently large $n$.
\end{assumption}

\begin{assumption}\label{as:unifmom}
    (i) $\sup_{\theta \in \Theta} \| \varphi_n(\theta) \| = O(1)$.
    (ii) $\sup_{\theta \in \Theta} \| \bar q_n(\theta) - \varphi_n(\theta) \| = o_P(1)$.
    (iii) Let $\{a_i\}_{i=1}^n$ be any non-stochastic uniformly bounded scalar sequence, and define $\bar a_{n,j} \coloneqq n^{-1} \sum_{i: \: j \in \calU(i)} a_i$ for each $j \in [n]$. 
    \begin{align}
    \sup_{\theta \in \Theta} \left| \frac{1}{n}\sum_{j \in [n]} \bar a_{n,j} \left\{ V_j(\bm W_n,\theta)^2 - \bbE \left[V_j(\bm W_n, \theta)^2\right] \right\} \right| = o_P(m^{-1}).
    \end{align}
\end{assumption}

Assumptions \ref{as:ident} and \ref{as:unifmom} are high-level conditions that need to be verified case by case.
In the next subsection, we provide more primitive alternative conditions to these in the case of SAR models.

In Assumption \ref{as:IV}, we assume that the IV vector $Z_i$ is fixed and bounded.
$Z_i$ typically comprises the covariate vector $X_i$ and its network-lagged version constructed using the observed interaction matrix $\bm G_n^\text{obs}$, and, as mentioned above, $\bm G_n^\text{obs}$ is treated as non-random.
Thus, the assumption essentially requires the covariates to be non-stochastic, which is a common setup in the literature on spatial and network econometrics.
Assumption \ref{as:omega} assumes that the GMM weight matrix is fixed, which precludes two-step optimal GMM-type estimators.
Combined with Assumption \ref{as:IV}, typical choices of $\Omega_n$ include $\Omega_n = I_{d_\mu}$ and $\Omega_n = (n^{-1} \bm Z_n^\top \bm Z_n)^{-1}$.
These assumptions are introduced mainly to simplify the theoretical presentation and to highlight the main properties of NA-GMM.
Relaxing them could be possible, but is beyond the scope of this paper.

\begin{theorem}[Consistency for the pseudo-true value]\label{thm:consistency}
    Suppose that Assumptions \ref{as:ident} -- \ref{as:unifmom} hold.
    Then, for each given $\rho > 0$, $\|\hat \theta_{n, \rho} - \theta_{n,\rho}^*\| = o_P(1)$.
\end{theorem}

Next, we derive the asymptotic distribution of the NA-GMM estimator.
As pointed out in \cite{hall2003large}, under general moment misspecification, the form of the asymptotic distribution critically depends on the convergence rate and limiting distribution of the weight matrix $\Psi_{n,\rho}(\theta)$.
Thus, a fully general discussion would complicate the analysis and would not necessarily be informative for practical use.
To avoid this complexity, we first consider a \textit{local misspecification} case in this subsection.
The case of general moment misspecification will be discussed in the next subsection.

Let $\mathcal N_0 \subset \Theta$ be a given fixed neighborhood of $\theta_0$.

\begin{assumption}[Locally misspecified moment condition]\label{as:popmom}
    (i) $\varphi_n(\theta_0) = \delta_n/\sqrt{n}$ and $\delta \coloneqq \lim_{n \to \infty} \delta_n$ exists.
    (ii) $\bar q_n(\theta)$ is twice continuously differentiable on $\mathcal N_0$, and, for each $\ell$-th element $\varphi_{n,\ell}(\theta)$ of $\varphi_n(\theta)$, $\max_{1\le \ell\le d_\mu} \sup_{\theta \in \mathcal N_0} \| \partial_{\theta\theta^\top}\varphi_{n,\ell}(\theta) \|_{\mathrm{op}} = O(1)$.
\end{assumption}

Assumption \ref{as:popmom}(i) is a local misspecification condition, which can be equivalently written as $\bbE[\bar u_n(\vec{\bm D}_{\calU}; \theta_0)] = -\delta_n/\sqrt{n}$.
In the SAR model in Section \ref{sec:SAR}, we have
\begin{align}
    \bbE[\bar u_n(\vec{\bm D}_{\calU}; \theta_0)] 
    & = - \alpha_0 \bbE[\Gamma_n] \vec{\bm D}_{\calU} \\
    & = - \frac{1}{n} \alpha_0 \left(Z_1 \bbE[\bm Y_{\calU(1)}]^\top, \ldots , Z_n \bbE[\bm Y_{\calU(n)}]^\top \right) \vec{\bm D}_{\calU} = O\left( \frac{m}{n} \right),
\end{align}
provided that $\bbE[Y_j]$'s and $\vec{\bm D}_{\calU}$ are uniformly bounded.
Thus, Assumption \ref{as:popmom}(i) is satisfied if $m = O(\sqrt{n})$.
This condition covers empirical situations in which only a small part of the network is uncertain.
For instance, in survey-based network data, respondents may be allowed to nominate only a fixed number of peers.
Most respondents may not reach this nomination limit, but a few respondents may have more peers than the limit allows.
The condition also covers cases in which the survey is incomplete only for a small number of schools, villages, or firms.

\begin{assumption}\label{as:Psi}
    Each element of $\Psi_{n,\rho}(\theta)$ is twice continuously differentiable on $\mathcal N_0$, such that
    \begin{align}
        \max_{1 \le j \le d_\theta} \sup_{\theta \in \mathcal N_0} \left\| \partial_{\theta_j}\Psi_{n,\rho}^*(\theta) \right\|_{\mathrm{op}} = O(1), \quad 
        \max_{1\le j,k\le d_\theta} \sup_{\theta\in\mathcal N_0} \left\| \partial_{\theta_j\theta_k}\Psi_{n,\rho}^*(\theta) \right\|_{\mathrm{op}} = O(1).
    \end{align}
\end{assumption}

\begin{assumption}\label{as:K}
    Let $\dot\varphi_n(\theta) \coloneqq \partial_{\theta^\top} \varphi_n(\theta)$, $\Psi_\rho^* \coloneqq \lim_{n \to \infty} \Psi_{n,\rho}^*(\theta_0)$, $\dot\varphi \coloneqq \lim_{n \to \infty} \dot\varphi_n(\theta_0)$, and $K_\rho \coloneqq \dot\varphi^\top \Psi_\rho^* \dot\varphi$.
    $\lambda_{\min}(K_\rho) > 0$.
\end{assumption}

\begin{assumption}[LLN]\label{as:LLN}
    (i) $\sup_{\theta\in\mathcal N_0} \| \partial_{\theta^\top} \bar q_n(\theta)-\dot\varphi_n(\theta) \| = o_P(1)$, and $\max_{1\le \ell\le d_\mu} \sup_{\theta\in\mathcal N_0} \| \partial_{\theta\theta^\top}\bar q_{n,\ell}(\theta) - \partial_{\theta\theta^\top}\varphi_{n,\ell}(\theta) \| = o_P(1)$.
    (ii) $\sup_{\theta\in\mathcal N_0} \| \Psi_{n,\rho}(\theta) - \Psi_{n,\rho}^*(\theta) \|_{\mathrm{op}} = o_P(1)$, $\max_{1\le j\le d_\theta} \sup_{\theta\in\mathcal N_0} \| \partial_{\theta_j}\Psi_{n,\rho}(\theta) - \partial_{\theta_j}\Psi_{n,\rho}^*(\theta) \|_{\mathrm{op}} = o_P(1)$, and $\max_{1\le j,k\le d_\theta} \sup_{\theta\in\mathcal N_0} \| \partial_{\theta_j\theta_k}\Psi_{n,\rho}(\theta) - \partial_{\theta_j\theta_k}\Psi_{n,\rho}^*(\theta) \|_{\mathrm{op}} = o_P(1)$.
\end{assumption}

\begin{assumption}[CLT]\label{as:CLT} $\sqrt{n}\left(\bar q_n(\theta_0)-\varphi_n(\theta_0)\right) \overset{d}{\to} N(0,\Sigma_q)$.
\end{assumption}

Assumptions \ref{as:Psi} and \ref{as:K} are mild regularity conditions.
By contrast, Assumptions \ref{as:LLN} and \ref{as:CLT} directly impose asymptotic properties on the sample moments and the weight matrix.
In the SAR model considered below, these conditions can be verified under relatively primitive conditions.

\begin{theorem}[Asymptotic normality under local misspecification]\label{thm:local}
    Suppose that Assumptions \ref{as:ident} -- \ref{as:CLT} hold.
    Then,
\begin{align}
    \sqrt{n}(\hat\theta_{n,\rho} - \theta_{n,\rho}^*)
    \overset{d}{\to} N\left(\bm 0, K_\rho^{-1}\dot\varphi^\top \Psi_\rho^* \Sigma_q \Psi_\rho^* \dot\varphi K_\rho^{-1} \right),
\end{align}
and
\begin{align}
    \sqrt{n}(\hat\theta_{n,\rho} - \theta_0)
    \overset{d}{\to} N\left(- K_\rho^{-1} \dot\varphi^\top \Psi_\rho^*\delta, K_\rho^{-1}\dot\varphi^\top \Psi_\rho^* \Sigma_q \Psi_\rho^* \dot\varphi K_\rho^{-1} \right).
\end{align}
\end{theorem}

The second result of Theorem \ref{thm:local} implies that, as a special case of Assumption \ref{as:popmom}(i), if $\delta_n = o(1)$, then the limiting distribution of $\sqrt{n}(\hat\theta_{n,\rho} - \theta_0)$ is centered at zero.
Hence, even when the moment conditions are correctly specified, there is little risk in using the NA-GMM estimator as an alternative to the naive estimator.

\subsection{Asymptotic results for SAR models}

\renewcommand{\theassumption}{\thesection.2.\arabic{assumption}}
\setcounter{assumption}{0}

We next study the asymptotic properties of the NA-GMM estimator for SAR models in a general network error setting.
The model and the choice of IVs are the same as those in Example \ref{ex:sar} and Section \ref{sec:SAR}.
The GMM weight matrix $\Omega_n$ is allowed to be any matrix satisfying Assumption \ref{as:omega}.

We first show that the high-level conditions introduced in the previous subsection can be verified under the following set of primitive assumptions.

\begin{assumption}\label{as:error}
    (i) $\{\varepsilon_i\}$ are independent and identically distributed (IID).
    (ii) $\bbE[\varepsilon_i] = 0$, and $\bbE |\varepsilon_i|^{4 + c} < \infty$ for some $c > 0$, such that $\sigma_\varepsilon^2 \coloneqq \bbE[\varepsilon_i^2]$, $\mu_{\varepsilon,3} \coloneqq \bbE[\varepsilon_i^3]$, and $\mu_{\varepsilon,4} \coloneqq \bbE[\varepsilon_i^4]$ exist.
\end{assumption}

\begin{assumption}\label{as:X}
    $\{X_i\}$ are non-stochastic and uniformly bounded in $i \in [n]$ for $n \ge 1$.
\end{assumption}

\begin{assumption}\label{as:network}
    (i) $\bm M_n \coloneqq (I_n - \alpha_0 \bm G_n)^{-1}$ exists.
    (ii) $\|\bm G_n\|_1$, $\|\bm G_n\|_\infty$, $\|\bm G_n^{\mathrm{obs}}\|_1$, $\|\bm G_n^{\mathrm{obs}}\|_\infty$, $\|\bm M_n\|_1$, and $\|\bm M_n\|_\infty$ are uniformly bounded in $n \ge 1$.
\end{assumption}

\begin{assumption}[Boundedness of uncertainty set]\label{as:ubound}
    Let $\bm U_n \coloneqq (\bm{1}\{(i,j) \in \calU_n\})_{i,j \in [n]}$. 
    $\|\bm U_n\|_1$ and $\|\bm U_n\|_\infty$ are uniformly bounded in $n \ge 1$.
\end{assumption}

\begin{assumption}[Identification]\label{as:identsar}
    (i) $\lambda_{\min} \left(\Pi_{n,-1}^\top \Psi_{n,\rho}^*(\alpha) \Pi_{n,-1} \right)  \ge c_\beta$ uniformly in $\alpha \in \Theta_\alpha$ for all sufficiently large $n$, where
    \begin{align}
        \Psi_{n,\rho}^*(\alpha) \coloneqq \left(\Omega_n^{-1} + \alpha^2 \bbE[ \Xi_{n,m,\rho} ]\right)^{-1},
        \quad \Xi_{n,m,\rho} \coloneqq \frac{m}{\rho} \Phi_n,
    \end{align}
    and $\Pi_n = \left(\pi_{n,1}, \Pi_{n,-1} \right)$, with $\pi_{n,1}$ being the first column of $\Pi_n$ and $\Pi_{n,-1}$ collecting the remaining columns. 
    (ii) For each $\epsilon>0$, 
    \begin{align}
    \inf_{\alpha \in \Theta_\alpha: \: |\alpha - \alpha_{n,\rho}^*| \ge \epsilon} \left\{ \inf_{\beta} Q_{n,\rho}^*(\alpha,\beta) - \inf_{\beta} Q_{n,\rho}^*(\alpha_{n,\rho}^*,\beta) \right\} \ge c_\epsilon > 0
    \end{align}
    for all sufficiently large $n$.
\end{assumption}

\begin{assumption}\label{as:variance}
For a given $v=(v_1^\top,v_2^\top,v_3^\top)^\top \in \bbR^{2d_\mu+d_\mu^2}$, let $\mathcal A_n(v)$ be a $1 \times n$ vector and $\mathcal B^s_n(v)$ be an $n \times n$ symmetric matrix satisfying 
\begin{align}
    \sqrt{n}
    v^\top \begin{pmatrix}
        \bar q_n(\theta_0) - \varphi_n(\theta_0) \\
        \pi_{n,1} - \bbE[\pi_{n,1}] \\
        \operatorname{vec}(\Xi_{n,m,\rho} - \bbE[\Xi_{n,m,\rho}])
    \end{pmatrix}
    = \mathcal A_n(v) \bm\varepsilon_n + \bm \varepsilon_n^\top \mathcal B^s_n(v) \bm \varepsilon_n - \bbE[\bm\varepsilon_n^\top \mathcal B^s_n(v) \bm\varepsilon_n],
\end{align}
whose definitions are given in \eqref{eq:varcomp}.
Let $\omega_{n,\rho}^2(v) \coloneqq \operatorname{Var} (\mathcal A_n(v)\bm \varepsilon_n + \bm \varepsilon_n^\top \mathcal B_n^s(v) \bm \varepsilon_n)$.
For any $v$ such that $\|v\| = 1$, $\omega_\rho^2(v) \coloneqq \lim_{n\to\infty} \omega_{n,\rho}^2(v)$ exists, and $\omega_{n,\rho}^2(v) \ge c_\omega > 0$ for all sufficiently large $n$.
\end{assumption}

In Assumption \ref{as:error}, we assume that the error terms are IID.
This assumption is stronger than necessary, but is imposed for simplicity and to demonstrate the properties of the NA-GMM estimator more transparently.
Noting that $Z_i = (\sum_{j \neq i} G_{ij}^{\mathrm{obs}} X_j^\top, X_i^\top)^\top$, Assumption \ref{as:X} parallels Assumption \ref{as:IV}.
Assumption \ref{as:network} is standard in the literature and limits the magnitude of social interactions.

Assumption \ref{as:ubound} is essential.
It requires that the size of the uncertainty set for each unit, $m_i$, is uniformly bounded, and also that each unit does not appear in other units' uncertainty sets too many times.
The assumption implies that $m = O(n)$, and therefore accommodates a wider range of network uncertainty than the local misspecification setting discussed above.
Roughly speaking, the uncertainty set should have a locally clustered structure.
This condition should be natural in many empirical applications where the network consists of multiple clusters, such as schools, grades, villages, or firms, and units in different clusters typically do not have close interactions.

Assumption \ref{as:identsar} is the main identification condition in the context of SAR models, where condition (i) ensures the identification of the pseudo-true parameter $\beta_{n,\rho}^*$ for $\beta$, and condition (ii) assumes the identifiability of $\alpha^*_{n,\rho}$, the pseudo-true parameter for $\alpha$.
Assumption \ref{as:variance} is a technical condition that rules out variance degeneracy of the NA-GMM estimator.
In Lemma \ref{lem:SAR_CLT}, we prove that $\mathcal A_n(v) \bm\varepsilon_n + \bm \varepsilon_n^\top \mathcal B^s_n(v) \bm \varepsilon_n - \bbE[\bm\varepsilon_n^\top \mathcal B^s_n(v) \bm\varepsilon_n]$ is asymptotically distributed as $N(0, \omega_\rho^2(v))$ under the assumptions made here.
The explicit expression for $\omega_\rho^2(v)$ can be obtained by applying the definitions of $\mathcal A_n(v)$ and $\mathcal B_n^s(v)$ to (3.2) of \cite{kelejian2001asymptotic}.

\medskip

We first show in the next lemma that the assumptions given above are sufficient to verify the main high-level conditions imposed in Theorems \ref{thm:consistency} and \ref{thm:local}.

\begin{lemma}[Verification of high-level conditions]\label{lem:verify}
    Suppose that Assumptions \ref{as:omega}, \ref{as:error} -- \ref{as:variance} hold.
    Then, Assumptions \ref{as:ident}, \ref{as:unifmom}, \ref{as:LLN}, and \ref{as:CLT} are satisfied.
\end{lemma}

Lemma \ref{lem:verify} implies that, under Assumptions \ref{as:error}--\ref{as:variance}, the NA-GMM estimator for the SAR model is consistent for the pseudo-true parameter and asymptotically normal in the sense of Theorem \ref{thm:local}, if the misspecification of the moment conditions is mild.
In fact, under the same conditions, we can show that, even when the moment misspecification is not local, the estimator is asymptotically normally distributed around the pseudo-true parameter.

\begin{theorem}[Asymptotic normality]\label{thm:SAR}
    Suppose that Assumptions \ref{as:omega}, \ref{as:error} -- \ref{as:variance} hold.
    In addition, assume that $H_{\rho}^* \coloneqq \lim_{n \to \infty} H_{n,\rho}^*(\theta_{n,\rho}^*)$ exists and is non-singular, and that $\mathcal J_{\rho} \coloneqq \lim_{n \to \infty} \mathcal J_{n,\rho}$ exists.
    Here, $H_{n,\rho}^*(\theta) \coloneqq \partial_{\theta \theta^\top} Q_{n, \rho}^*(\theta)$, whose explicit expression is given in \eqref{eq:hessian}.
    The definition of $\mathcal J_{n,\rho}$ is provided in the proof.
    Then,
    \begin{align}
    \sqrt n(\hat\theta_{n,\rho}-\theta_{n,\rho}^*)
    \overset{d}{\to}
    N\left(\bm 0, (H_\rho^*)^{-1}\mathcal J_\rho\Sigma_\rho^{\mathrm{SAR}}\mathcal J_\rho^\top (H_\rho^*)^{-1}\right),
    \end{align}
    where the definition of $\Sigma_\rho^{\mathrm{SAR}}$ is provided in Lemma \ref{lem:SAR_CLT}.
\end{theorem}

As shown in Theorem \ref{thm:SAR}, the NA-GMM estimator for the SAR model is asymptotically normal around the pseudo-true parameter under fairly general network error settings.
However, statistical inference based on this result appears challenging, not only because the asymptotic variance takes an extremely complicated form, but also, more fundamentally, because it depends on the true parameter $\theta_0$ and the true interaction matrix $\bm G_n$, which are fully unknown in our context.
One might consider alternative simulation-based inference procedures; however, these generally require prior knowledge of the true (or, at least, approximately true) interaction structure (e.g., \citealp{kojevnikov2021bootstrap, conley2023bootstrap}).
Moreover, there is some debate about the empirical usefulness of statistical inference for pseudo-true parameters (e.g., \citealp{hansen2021inference, andrews2026true}), despite the desirable bias reduction property of NA-GMM estimators discussed in Proposition \ref{prop:op-norm}.

Considering these issues, we suggest using the NA-GMM method as a diagnostic tool for assessing the robustness of the naive GMM estimator to network uncertainty.
In the next subsection, we provide several diagnostic analyses that may be useful in empirical applications.

\subsection{Diagnostics for network uncertainty}\label{subsec:diag}

\subsubsection{Sensitivity of $\hat\theta_{n,\rho}$ to $\rho$}

In NA-GMM, $\rho$ should be interpreted more as a sensitivity parameter to network errors than as a regularization parameter governing the prediction performance.
Thus, the NA-GMM method can be used most informatively when we examine how the estimate $\hat\theta_{n,\rho}$ changes along the path of $\rho$.
If the estimate remains stable even for relatively small values of $\rho$, this suggests that the naive GMM estimate is not very sensitive to the type of network uncertainty captured by the network adjustment procedure.
By contrast, if the estimate changes substantially as $\rho$ decreases, this indicates that the empirical conclusion may be sensitive to potential network errors.
Thus, such a path serves as a direct diagnostic for the robustness of the naive GMM estimator to network uncertainty.

\subsubsection{Moment fit improvement}

As a diagnostic statistic for quantifying the improvement in moment fit due to the network adjustment, we consider
\begin{align}
    r_n(\rho) \coloneqq 1 - \frac{\left\| \bar q_n(\hat \theta_{n, \rho}) + \bar u_n(\vec{\bm d}_m^\circ(\hat \theta_{n, \rho}); \hat \theta_{n, \rho})\right\|_{\Omega_n}}{\left\| \bar q_n(\hat \theta_n^\text{naive}) \right\|_{\Omega_n}},
\end{align}
where $\vec{\bm d}_m^\circ(\theta)$ is the optimal network correction vector, whose definition is given in \eqref{eq:dcirc} in Appendix.
This statistic measures the relative improvement in moment fit achieved by the NA-GMM adjustment.
\begin{comment}
By direct calculation, we can find that the numerator of the second term can be computed as
\begin{align}
    \left\| \bar q_n(\hat \theta_{n, \rho}) + \bar u_n(\vec{\bm d}_m^\circ(\hat \theta_{n, \rho}); \hat \theta_{n, \rho})\right\|_{\Omega_n} 
    %& = \left\| \bar q_n(\hat \theta_{n, \rho}) - \bm H_m(\hat \theta_{n, \rho}) \bm B_m(\hat \theta_{n, \rho})^{-1} \bm H_m(\hat \theta_{n, \rho})^\top \Omega_n \bar q_n(\hat \theta_{n, \rho}) \right\|_{\Omega_n} \\
    %& = \left\| \Omega_n^{-1} \left\{ \Omega_n - \Omega_n \bm H_m(\hat \theta_{n, \rho})  \bm B_m(\hat \theta_{n, \rho})^{-1} \bm H_m(\hat \theta_{n, \rho})^\top \Omega_n \right\} \bar q_n(\hat \theta_{n, \rho}) \right\|_{\Omega_n} \\
    %& = \left\| \Omega_n^{-1} \Psi_{n,\rho}(\hat \theta_{n, \rho}) \bar q_n(\hat \theta_{n, \rho}) \right\|_{\Omega_n} \\
    & = \left( \bar q_n(\hat \theta_{n, \rho})^\top \Psi_{n,\rho}(\hat \theta_{n, \rho}) \Omega_n^{-1} \Psi_{n,\rho}(\hat \theta_{n, \rho}) \bar q_n(\hat \theta_{n, \rho}) \right)^{1/2} 
\end{align}
\end{comment}
Note that $0 \le r_n(\rho) \le 1$ holds as long as $\left\| \bar q_n(\hat \theta_n^\text{naive}) \right\|_{\Omega_n} > 0$.
Indeed, we have
\begin{align}
    \left\| \bar q_n(\hat \theta_{n, \rho}) + \bar u_n(\vec{\bm d}_m^\circ(\hat \theta_{n, \rho}); \hat \theta_{n, \rho})\right\|_{\Omega_n}^2
    & \le \left\| \bar q_n(\hat \theta_{n, \rho}) + \bar u_n(\vec{\bm d}_m^\circ(\hat \theta_{n, \rho}); \hat \theta_{n, \rho})\right\|_{\Omega_n}^2 + \frac{\rho}{m} \left\|\vec{\bm d}_m^\circ(\hat \theta_{n, \rho})\right\|^2 \\
    & \le \left\| \bar q_n(\hat \theta_n^\text{naive}) + \bar u_n(\vec{\bm d}_m^\circ(\hat \theta_n^\text{naive}); \hat \theta_n^\text{naive})\right\|_{\Omega_n}^2 + \frac{\rho}{m} \left\|\vec{\bm d}_m^\circ(\hat \theta_n^\text{naive})\right\|^2 \\
    & \le \left\| \bar q_n(\hat \theta_n^\text{naive}) \right\|_{\Omega_n}^2 .
\end{align}
Hence, a value of $r_n(\rho)$ close to one indicates a large improvement in moment fit, whereas a value close to zero indicates that the network adjustment yields little improvement for this $\rho$.

\subsubsection{$J$-type statistic}

As another diagnostic statistic for measuring the extent of moment fit improvement, we consider a $J$-type statistic.
In SAR models, define
\begin{align}
    J_n(\rho) \coloneqq \frac{n}{\hat\sigma^2_{n,\rho}}
    \left\| \bar q_n(\hat \theta_{n,\rho}) + \bar u_n(\vec{\bm d}_m^\circ(\hat \theta_{n,\rho}); \hat \theta_{n,\rho}) \right\|_{\Omega_n}^2,
\end{align}
where $\hat\sigma^2_{n,\rho} \coloneqq n^{-1} \left\| \bm Y_n - \bm R_n^\text{obs}\hat \theta_{n,\rho} \right\|^2$.
When $\rho=\infty$, the adjustment term disappears and $\hat \theta_{n,\infty}$ reduces to the naive GMM estimator.
In particular, if $\Omega_n = (n^{-1} \bm Z_n^\top \bm Z_n)^{-1}$, we have
\begin{align}
    J_n(\infty) = \frac{1}{\hat\sigma^2_{n,\infty}}(\bm Y_n-\bm R_n^\text{obs}\hat \theta_n^\text{naive})^\top \bm Z_n(\bm Z_n^\top\bm Z_n)^{-1}\bm Z_n^\top (\bm Y_n-\bm R_n^\text{obs}\hat \theta_n^\text{naive}),
\end{align}
which coincides with the classical Sargan $J$-statistic.
Similar to the preceding diagnostic statistic, $J_n(\rho)$ summarizes the improvement in the fit of the moment conditions.

\section{Monte Carlo Simulations}\label{sec:MC}

In this section, we examine the finite sample performance of NA-GMM in linear SAR models.
The purpose of the simulation is to evaluate the bias reduction property of the fixed-weight NA-GMM and the CU-type original NA-GMM estimators.

The simulation design is as follows.
For each sample size $n \in \{500,1000\}$, we first randomly allocate $n$ units on a square lattice of size $\sqrt{1.2n} \times \sqrt{1.2n}$.
The true adjacency matrix $\bm A_n$ is then constructed based on these locations.
For each unit $i$, we first define its peer candidates as units located within distance three from $i$, and then randomly select a number of candidates from $\{0,1,\ldots,5\}$ to form links with $i$.
The true interaction matrix $\bm G_n$ is obtained by row-normalizing this adjacency matrix.
For each $i$, the uncertainty set $\calU(i)$ consists of all units excluding $i$ located within distance five from $i$.
The outcome is then generated from the linear SAR model
\begin{align}
        Y_i = \alpha_0 \sum_{j \neq i} G_{ij} Y_j + \beta_{00} + X_{1i} \beta_{10} + X_{2i} \beta_{20} + \varepsilon_i,
\end{align}
where $\alpha_0 = 0.6$, $(\beta_{00}, \beta_{10}, \beta_{20}) = (1.2,-1,1.4)$, $X_{1i} \sim N(0,1)$, $X_{2i} \sim |N(0,1)|$, and $\varepsilon_i \sim N(0,1)$ independently across $i$.

After the data are generated, we create the imprecisely observed network $\bm A_n^{\mathrm{obs}}$ by deleting true links and adding false links to $\bm A_n$.
Specifically, each true link in the uncertainty set is deleted with probability $p_{\mathrm{drop}} \in \{0,0.2,0.3\}$, and each absent link in the uncertainty set is added as a false link with probability $p_{\mathrm{add}} \in \{0,0.05,0.1\}$.
The observed interaction matrix $\bm G_n^{\mathrm{obs}}$ is obtained by row-normalizing $\bm A_n^{\mathrm{obs}}$.
For each design, we compare the performance of the naive GMM estimator based on $\bm G_n^{\mathrm{obs}}$, the fixed-weight NA-GMM estimator, and the original NA-GMM estimator.
The IVs for the spatial autoregressive term are constructed from the first and second spatial lags of $(X_1, X_2)$ based on $\bm G_n^{\mathrm{obs}}$.
For NA-GMM and fixed-weight NA-GMM, we consider two GMM weight matrices: the identity weight $\Omega_n = I_{d_\mu}$ and the 2SLS weight $\Omega_n = (n^{-1}\bm Z_n^\top \bm Z_n)^{-1}$.
Note that the naive GMM estimator with the 2SLS weight corresponds to the usual 2SLS estimator.
For the penalty parameter in the NA-GMM estimators, we consider $\rho \in \{2^{-6},2^{-5},\ldots,2^{14}\}$, where we set $\kappa = \rho/\alpha_0^2$ for fixed-weight NA-GMM.
Each simulation setup is repeated 500 times, and the performance of each estimator is evaluated by the average bias and root mean squared error (RMSE) for estimating $\alpha_0$ and $\beta_{10}$.

To save space, the summary tables for the simulation results with $\rho = 2^{-3}$, $2^8$, and $2^{12}$ are provided in Appendix \ref{app:table}.
The main findings from these tables are as follows.
First, for the estimation of the coefficient $\beta_{10}$, the three GMM estimators perform very similarly, and all of them produce accurate estimates of $\beta_{10}$ in all scenarios.
This is not surprising because network errors are generated independently of the covariates.

We now focus on the estimation of $\alpha_0$.
When there are no network errors, not only the naive GMM estimator but also the two NA-GMM estimators estimate $\alpha_0$ very accurately.
This is an expected result because, when there are no errors in the observed interaction matrix, the population optimal network correction is $\vec{\bm d}_m = \bm 0$, which reduces the NA-GMM criterion to the standard naive GMM criterion.
In the presence of network errors, for both types of GMM weights, the original NA-GMM estimator with small $\rho$ generally outperforms the other estimators in terms of bias reduction.
This corroborates the desirable bias reduction property of the NA-GMM method.

Interestingly, although Proposition \ref{prop:op-norm} shows that the fixed-weight NA-GMM estimator also has a bias reduction property, the magnitude of bias reduction is much smaller than that of the original NA-GMM estimator.
In particular, almost no bias reduction is observed when we set $\Omega_n = (n^{-1}\bm Z_n^\top \bm Z_n)^{-1}$.
This finding may be explained by the structure of the GMM weight matrix.
As shown in \eqref{eq:fw_weight}, when $\Omega_n = (n^{-1}\bm Z_n^\top \bm Z_n)^{-1}$, the fixed-weight NA-GMM estimator uses the weight matrix $\Psi_{n,\kappa} = (n^{-1}\bm Z_n^\top \bm Z_n + (m/\kappa)\Phi_n)^{-1}$.
Noting that $\Phi_n = n^{-2} \sum_{i \in [n]} \sum_{j \in \calU(i)} Z_i Z_i^\top Y_j^2$, if $n^{-1}\sum_{j \in \calU(i)} Y_j^2$ does not vary much across $i$, then the network adjustment component $(m/\kappa)\Phi_n$ becomes approximately proportional to $n^{-1}\bm Z_n^\top \bm Z_n$ for any $\kappa$.
In such cases, fixed-weight NA-GMM would behave very similarly to 2SLS, as observed here.
These findings indicate that allowing the GMM weight matrix to depend on the parameter plays an important role in reducing the bias.
Another interesting finding is that, even for the original NA-GMM estimator, we do not observe any bias reduction when network errors arise only from dropping existing links.
This indicates that, while NA-GMM can effectively reduce some types of bias, it may be less effective for other types.
These points deserve further investigation in future work.

Figures \ref{fig:alpha_I} and \ref{fig:alpha_S} summarize the simulation results for estimating $\alpha_0$.
The gray area in each panel corresponds to the simulated pointwise 95 percent interval at each value of $\rho$.
The results discussed above can also be visually confirmed in these figures.

\begin{figure}[ht]
\begin{center}
\includegraphics[width = 12cm]{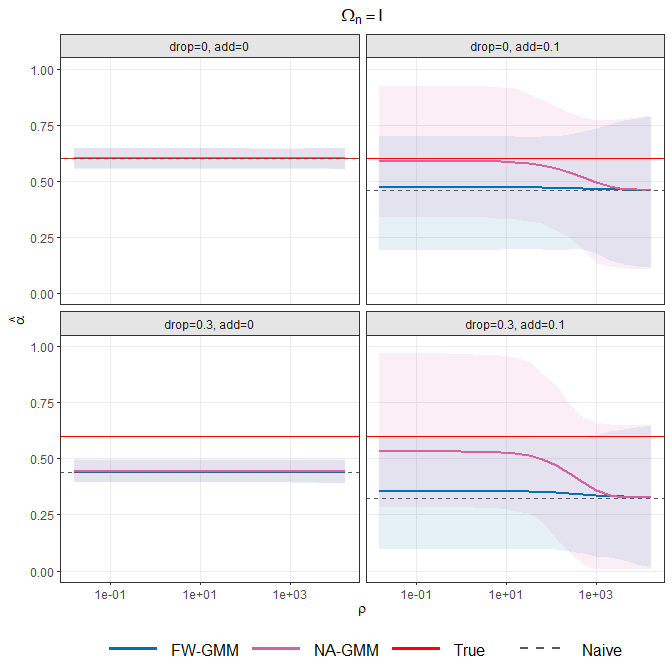}
\caption{Simulation results for estimating $\alpha_0$: $\Omega_n = I_{d_\mu}$}
\label{fig:alpha_I}
\end{center}
\end{figure}

\begin{figure}[ht]
\begin{center}
\includegraphics[width = 12cm]{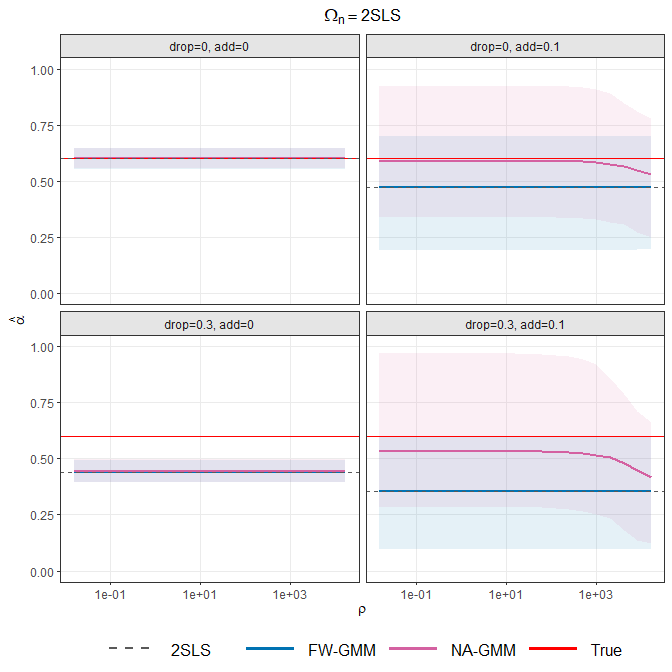}
\caption{Simulation results for estimating $\alpha_0$: $\Omega_n =$ 2SLS weight}
\label{fig:alpha_S}
\end{center}
\end{figure}

\section{An Empirical Application: Spatial Diffusion of COVID-19 Infections}\label{sec:empir}

In this section, we illustrate the NA-GMM method by applying it to U.S. county-level COVID-19 infection rate data in 2022.
County-level COVID-19 case data for 2022 were obtained from the data archive of the Centers for Disease Control and Prevention (CDC).\footnote{\url{https://data.cdc.gov/}}
For county characteristics, we use the median household income, the share of bachelor's degree (or higher) holders, and the unemployment rate.
These variables, together with the county population data, were obtained from the U.S. Department of Agriculture (USDA) Economic Research Service county-level data.\footnote{\url{https://www.ers.usda.gov/data-products/county-level-data-sets}}

After excluding non-contiguous regions, isolated counties, and one county with an extreme outlying infection rate, the final sample consists of 3,099 contiguous U.S. counties.
Using the annual infection rate, defined as total cases in 2022 divided by the population, as the outcome variable, we consider the following SAR model:
\begin{align}
    \texttt{infection rate}_i 
    & = \alpha_0 \sum_{j \neq i} G_{ij} \texttt{infection\ rate}_j \\
    & \quad + \left(1, \ln \texttt{median income}_i, \texttt{bachelor}_i, \texttt{unemployment}_i\right)^\top \beta_0 + \varepsilon_i.
\end{align}
Here, the observed counterpart of $G_{ij}$, $G^\text{obs}_{ij}$, is the row-normalized version of $A^\text{obs}_{ij}$, where $A^\text{obs}_{ij}$ is equal to one if counties $i$ and $j$ are adjacent and zero otherwise.
Figure \ref{fig:map} shows a map of $\texttt{infection\ rate}_i$ across counties.
This pattern indicates that COVID-19 infection rates in adjacent counties are strongly related, suggesting that $\alpha_0$ is relatively large and positive in the SAR model.

\begin{figure}[ht]
\begin{center}
\includegraphics[width = 16cm]{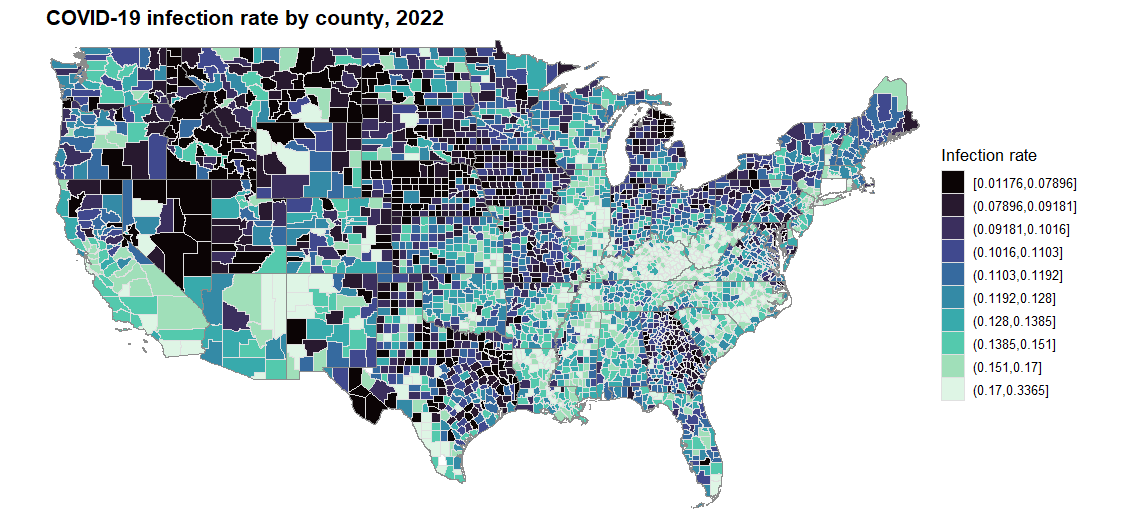}
\caption{County-level COVID-19 infection rates in the U.S.}
\label{fig:map}
\end{center}
\end{figure}

We first estimated the model by 2SLS, using the first and second spatial lags of the regressors as IVs.
The estimate of $\alpha_0$ was about 0.8 and significantly positive.
However, the computed Sargan $J$-statistic was 12.17, suggesting possible misspecification of the moment conditions.
We next estimated the same model using the NA-GMM method.
Figure \ref{fig:sb} reports the path of the NA-GMM estimate $\hat \alpha_{n,\rho}$ over different values of $\rho$.
As shown in the figure, as $\rho$ tends to infinity, the estimate converges to the 2SLS estimate, which is consistent with our theory.
Moreover, even for very small values of $\rho$, the estimate deviates only slightly from the 2SLS estimate.

Interestingly, the two diagnostic statistics for the improvement in moment fit introduced in Subsection \ref{subsec:diag} show that the moment fit is substantially improved for small values of $\rho$, as reported in Figures \ref{fig:diag1} and \ref{fig:diag2}.
Taken together, these findings suggest that, although the moment conditions appear to be misspecified to some extent, the moment correction induced by the network adjustment does not lead to a substantial change in the estimate of the spatial parameter $\alpha_0$.
Thus, the estimate of $\alpha_0$ appears relatively robust to the type of network uncertainty in this dataset.

\begin{figure}[ht]
\begin{center}
\includegraphics[width = 10cm]{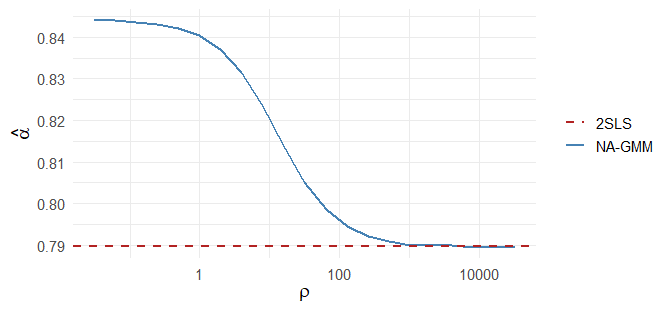}
\caption{Estimated $\alpha^*_{n,\rho}$}
\label{fig:sb}
\end{center}
\end{figure}

\begin{figure}[ht]
\centering
\begin{subfigure}[b]{0.48\textwidth}
\centering
\includegraphics[width = 9cm]{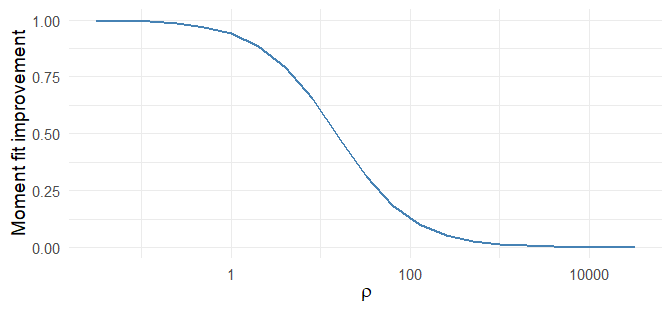}
\caption{Moment fit improvement}
\label{fig:diag1}
\end{subfigure}
\begin{subfigure}[b]{0.48\textwidth}
\centering
\includegraphics[width = 9cm]{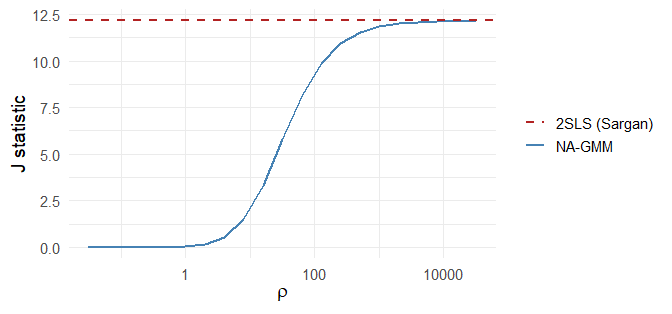}
\caption{$J$-type statistic}
\label{fig:diag2}
\end{subfigure}
\caption{Two diagnostic statistics for network adjustment}
\label{fig:diag}
\end{figure}

\section{Conclusion}\label{sec:conclude}

In this paper, we proposed a network-adjusted generalized method of moments (NA-GMM) framework for IV-based linear social interaction models under network uncertainty.
The proposed method modifies the standard GMM approach by allowing the observed interaction matrix to be adjusted while penalizing the size of the adjustment.
For linear SAR models, we showed that the NA-GMM criterion downweights moment directions that are sensitive to network errors and established a bias reduction property for a fixed-weight version of the estimator.
We also established consistency for the pseudo-true parameter and asymptotic normality under general moment misspecification.
Finally, we proposed diagnostic analyses for evaluating the robustness of naive GMM estimates to network errors.
An empirical application to U.S. county-level COVID-19 infection data demonstrated the usefulness of the proposed method.

The present analysis leaves several open questions and potential extensions.
First, it would be important to further investigate the bias reduction property of NA-GMM; this paper established a related result only for a fixed-weight version in linear SAR models.
Second, developing feasible statistical inference procedures for NA-GMM estimators would be important for improving the applicability of the proposed framework.
The present framework could also be extended to more general models, such as models in which network errors enter nonlinearly and models with endogenous network errors.
Finally, while we adopted a quadratic penalty for the network adjustment for convenience, it would be worth studying more general penalty structures.
These topics are left for future research.

\clearpage
\appendix
\begin{center}
    \Large \textbf{Appendix}
\end{center}

\section{Technical Appendix}

\subsection{Proof of Lemma \ref{lem:NA-GMM-obj}}\label{app:NA-GMM-obj}

Recall that the NA-GMM criterion function is
\begin{align}
    Q_{n,\rho}(\theta) = \inf_{\vec{\bm d}_m}\left\{\left\| \bar q_n(\theta) + \bm H_m(\theta)\vec{\bm d}_m \right\|_{\Omega_n}^2 + \frac{\rho}{m} \| \vec{\bm d}_m \|^2 \right\}.
\end{align}
The first-order condition for the minimization with respect to $\vec{\bm d}_m$ is
\begin{align}
    \bm H_m(\theta)^\top \Omega_n \left( \bar q_n(\theta) + \bm H_m(\theta)\vec{\bm d}_m \right) + \frac{\rho}{m} \vec{\bm d}_m = \bm 0 .
\end{align}
Hence,
\begin{align}
    \left( \bm H_m(\theta)^\top \Omega_n \bm H_m(\theta) + \frac{\rho}{m} I_m \right) \vec{\bm d}_m 
    = - \bm H_m(\theta)^\top \Omega_n \bar q_n(\theta).
\end{align}
Therefore, the minimizer is characterized as
\begin{align}\label{eq:dcirc}
    \vec{\bm d}_m^\circ(\theta)
    = - \bm B_m(\theta)^{-1} \bm H_m(\theta)^\top \Omega_n \bar q_n(\theta)
\end{align}
for $\rho > 0$, where $\bm B_m(\theta) \coloneqq \bm H_m(\theta)^\top \Omega_n \bm H_m(\theta) + (\rho/m) I_m$.
Expanding the criterion function, we have
\begin{align}
    Q_{n,\rho}(\theta)
    &= \bar q_n(\theta)^\top \Omega_n \bar q_n(\theta) + 2 \bar q_n(\theta)^\top \Omega_n \bm H_m(\theta)\vec{\bm d}_m^\circ(\theta) \\
    &\quad + \vec{\bm d}_m^\circ(\theta)^\top \bm H_m(\theta)^\top \Omega_n \bm H_m(\theta) \vec{\bm d}_m^\circ(\theta) + \frac{\rho}{m} \vec{\bm d}_m^\circ(\theta)^\top \vec{\bm d}_m^\circ(\theta) \\
    &= \bar q_n(\theta)^\top \Omega_n \bar q_n(\theta) + 2 \bar q_n(\theta)^\top \Omega_n \bm H_m(\theta)\vec{\bm d}_m^\circ(\theta) + \vec{\bm d}_m^\circ(\theta)^\top \bm B_m(\theta) \vec{\bm d}_m^\circ(\theta).
\end{align}
Substituting $\vec{\bm d}_m^\circ(\theta)$ into the second term gives
\begin{align}
    2\bar q_n(\theta)^\top \Omega_n \bm H_m(\theta)\vec{\bm d}_m^\circ(\theta) = - 2\bar q_n(\theta)^\top \Omega_n \bm H_m(\theta)\bm B_m(\theta)^{-1}\bm H_m(\theta)^\top \Omega_n \bar q_n(\theta).
\end{align}
For the third term, we have
\begin{align}
    \vec{\bm d}_m^\circ(\theta)^\top \bm B_m(\theta) \vec{\bm d}_m^\circ(\theta)
    &= \bar q_n(\theta)^\top \Omega_n \bm H_m(\theta)\bm B_m(\theta)^{-1}\bm H_m(\theta)^\top \Omega_n \bar q_n(\theta).
\end{align}
Therefore,
\begin{align}
    Q_{n,\rho}(\theta)
    &= \bar q_n(\theta)^\top \Omega_n \bar q_n(\theta) - 2 \bar q_n(\theta)^\top \Omega_n \bm H_m(\theta)\bm B_m(\theta)^{-1}\bm H_m(\theta)^\top \Omega_n \bar q_n(\theta) \\
    &\quad + \bar q_n(\theta)^\top \Omega_n \bm H_m(\theta)\bm B_m(\theta)^{-1}\bm H_m(\theta)^\top \Omega_n \bar q_n(\theta) \\
    &= \bar q_n(\theta)^\top \Omega_n \bar q_n(\theta) - \bar q_n(\theta)^\top \Omega_n \bm H_m(\theta)\bm B_m(\theta)^{-1}\bm H_m(\theta)^\top \Omega_n \bar q_n(\theta) \\
    &= \bar q_n(\theta)^\top \left[ \Omega_n - \Omega_n \bm H_m(\theta)\bm B_m(\theta)^{-1}\bm H_m(\theta)^\top \Omega_n \right] \bar q_n(\theta).
\end{align}
Recalling the definition of $\bm B_m(\theta)$, we obtain
\begin{align}
Q_{n,\rho}(\theta)
&= \bar q_n(\theta)^\top \left[ \Omega_n - \Omega_n \bm H_m(\theta) \left( \bm H_m(\theta)^\top \Omega_n \bm H_m(\theta) + \frac{\rho}{m} I_m \right)^{-1} \bm H_m(\theta)^\top \Omega_n \right] \bar q_n(\theta).
\end{align}
Finally, noting the identity $(A^{-1} + U C^{-1} U^\top)^{-1} = A - A U (U^\top A U + C)^{-1} U^\top A$, we obtain
\begin{align}
    Q_{n,\rho}(\theta) = \bar q_n(\theta)^\top \left( \Omega_n^{-1} + \frac{m}{\rho}\bm H_m(\theta) \bm H_m(\theta)^\top \right)^{-1} \bar q_n(\theta),
\end{align}
as desired.
\qed

\subsection{Proof of Proposition \ref{prop:op-norm}}

Let $t = \kappa^{-1}$ and define
\begin{align}
    \Psi_n(t)
    \coloneqq
    \left( I_{d_\mu} + t m \Phi_n \right)^{-1},
\end{align}
and
\begin{align}
    K_n(t)
    \coloneqq
    \left( \Pi_n^\top \Psi_n(t) \Pi_n \right)^{-1} \Pi_n^\top \Psi_n(t) \Phi_n \Psi_n(t)\Pi_n \left( \Pi_n^\top \Psi_n(t) \Pi_n \right)^{-1}.
\end{align}
Recalling $\Phi_n = \Gamma_n \Gamma_n^\top$ and $L_n(\kappa) = \alpha_0 \left( \Pi_n^\top \Psi_{n, \kappa} \Pi_n \right)^{-1} \Pi_n^\top \Psi_{n, \kappa} \Gamma_n$, we have
\begin{align}
    \|L_n(\kappa)\|_{\mathrm{op}}^2
    =
    \alpha_0^2
    \sup_{\|\bm c\|=1}
    \bm c^\top K_n(\kappa^{-1}) \bm c.
\end{align}
Thus, the result follows if $\bm c^\top K_n(t)\bm c$ is non-increasing in $t$.

Let
\begin{align}
    A_n(t)
    \coloneqq
    \Pi_n^\top \Psi_n(t) \Pi_n,
    \quad
    B_n(t)
    \coloneqq
    \Pi_n^\top \Psi_n(t) \Phi_n \Psi_n(t) \Pi_n,
    \quad
    C_n(t)
    \coloneqq
    \Pi_n^\top \Psi_n(t) \Phi_n \Psi_n(t) \Phi_n \Psi_n(t) \Pi_n,
\end{align}
so that
\begin{align}
    K_n(t) = A_n(t)^{-1} B_n(t) A_n(t)^{-1}.
\end{align}
The inverse matrix differentiation formula gives
\begin{align}
    \frac{\partial}{\partial t} \Psi_n(t) = - m \Psi_n(t) \Phi_n \Psi_n(t).
\end{align}
Therefore, $\partial A_n(t) / \partial t = - m B_n(t)$ and $\partial B_n(t) /  \partial t = - 2 m C_n(t)$.
Again by the inverse matrix differentiation formula,
\begin{align}
    \frac{\partial}{\partial t} A_n(t)^{-1} = m A_n(t)^{-1} B_n(t) A_n(t)^{-1}.
\end{align}
Hence,
\begin{align}
    \frac{\partial}{\partial t} K_n(t)
    & = \left( \frac{\partial}{\partial t} A_n(t)^{-1} \right) B_n(t) A_n(t)^{-1}
    + A_n(t)^{-1} \left( \frac{\partial}{\partial t} B_n(t) \right) A_n(t)^{-1}
    + A_n(t)^{-1} B_n(t) \left( \frac{\partial}{\partial t} A_n(t)^{-1} \right) \\
    & = 2 m A_n(t)^{-1} \left\{ B_n(t) A_n(t)^{-1} B_n(t) - C_n(t) \right\} A_n(t)^{-1}.
\end{align}

Now, define
\begin{align}
    X_n(t)
    \coloneqq \Psi_n(t)^{1/2} \Pi_n,
    \qquad
    Y_n(t)
    \coloneqq \Psi_n(t)^{1/2} \Phi_n \Psi_n(t)\Pi_n.
\end{align}
Then
\begin{align}
    A_n(t)
    =
    X_n(t)^\top X_n(t),
    \qquad
    B_n(t)
    = Y_n(t)^\top X_n(t) = X_n(t)^\top Y_n(t),
    \qquad
    C_n(t)
    =
    Y_n(t)^\top Y_n(t).
\end{align}
Therefore,
\begin{align}
    B_n(t) A_n(t)^{-1} B_n(t)
    =
    Y_n(t)^\top
    X_n(t)
    \left( X_n(t)^\top X_n(t) \right)^{-1}
    X_n(t)^\top
    Y_n(t).
\end{align}
Letting $P_{X_n(t)} \coloneqq X_n(t) \left( X_n(t)^\top X_n(t) \right)^{-1} X_n(t)^\top$, since $I_{d_\mu} - P_{X_n(t)}$ is positive semidefinite,
\begin{align}
    C_n(t) - B_n(t) A_n(t)^{-1} B_n(t)
    & =
    Y_n(t)^\top
    \left( I_{d_\mu} - P_{X_n(t)} \right)
    Y_n(t)
\end{align}
is positive semidefinite.
This means that, for any $\bm c$,
\begin{align}
    - \bm c^\top
    \left(
    \frac{\partial}{\partial t} K_n(t)
    \right)
    \bm c
    = 2 m
    \bm c^\top A_n(t)^{-1} 
    \left\{
    C_n(t) - B_n(t) A_n(t)^{-1} B_n(t)
    \right\}
     A_n(t)^{-1}\bm c 
    \ge
    0,
\end{align}
which implies the desired result.

\qed

\subsection{Proof of Theorem \ref{thm:consistency}}

\begin{lemma}\label{lem:unifconv}
Suppose that Assumptions \ref{as:IV} -- \ref{as:unifmom} hold.
Then, we have
\begin{align}
    \sup_{\theta \in \Theta} \left| Q_{n,\rho}(\theta) - Q_{n,\rho}^*(\theta) \right| = o_P(1).
\end{align}
\end{lemma}

\begin{proof}
By definition,
\begin{align}
    \bm H_m(\theta)\bm H_m(\theta)^\top
    & = \frac{1}{n^2}\sum_{i \in [n]} Z_i Z_i^\top [V_i(\bm W_n,\calU;\theta)^\top V_i(\bm W_n,\calU;\theta)] \\
    & = \frac{1}{n^2}\sum_{i \in [n]} \sum_{j\in \calU(i)} Z_i Z_i^\top (V_j(\bm W_n,\theta))^2 \\
    & = \frac{1}{n^2}\sum_{i \in [n]} \sum_{j \in [n]} Z_i Z_i^\top \bm{1}\{j \in \calU(i)\}  (V_j(\bm W_n,\theta))^2.
\end{align}
Hence, for each $(k,\ell)$, 
\begin{align}
    \left(\bm H_m(\theta)\bm H_m(\theta)^\top - \bbE[\bm H_m(\theta)\bm H_m(\theta)^\top]\right)_{k,\ell} = \frac{1}{n}\sum_{j \in [n]} \bar a_{n,j} \left\{ V_j(\bm W_n,\theta)^2 - \bbE[V_j(\bm W_n,\theta)^2] \right\}
\end{align}
with $\bar a_{n,j} = n^{-1} \sum_{i: \: j \in \calU(i)} Z_{ik}Z_{i\ell}$, which is uniformly bounded by Assumption \ref{as:IV}.
Therefore, by Assumption \ref{as:unifmom}(iii),
\begin{align}
    \sup_{\theta\in\Theta}
    \left\|
    \bm H_m(\theta)\bm H_m(\theta)^\top - \bbE[\bm H_m(\theta)\bm H_m(\theta)^\top]
    \right\|
    = o_P(m^{-1}).
\end{align}
Using $A^{-1}-B^{-1}=A^{-1}(B-A)B^{-1}$,
\begin{align}
    & \sup_{\theta\in\Theta}
    \| \Psi_{n,\rho}(\theta)-\Psi_{n,\rho}^*(\theta) \| \\
    & \le
    \frac{m}{\rho} \sup_{\theta\in\Theta}\| \Psi_{n,\rho}(\theta) \|_{\mathrm{op}} \sup_{\theta\in\Theta} \left\|  \bm H_m(\theta)\bm H_m(\theta)^\top - \bbE[\bm H_m(\theta)\bm H_m(\theta)^\top]  \right\| \sup_{\theta\in\Theta} \| \Psi_{n,\rho}^*(\theta) \|_{\mathrm{op}}.
\end{align}
By Assumption \ref{as:omega} and the fact that $\bm H_m(\theta)\bm H_m(\theta)^\top$ and $\bbE[\bm H_m(\theta)\bm H_m(\theta)^\top]$ are positive semidefinite, $\sup_{\theta\in\Theta}\|\Psi_{n,\rho}(\theta)\|_{\mathrm{op}} = O(1)$ and $\sup_{\theta\in\Theta}\|\Psi_{n,\rho}^*(\theta)\|_{\mathrm{op}} = O(1)$.
As a result, we obtain
\begin{align}\label{eq:Psi_unif}
    \sup_{\theta\in\Theta} \| \Psi_{n,\rho}(\theta)-\Psi_{n,\rho}^*(\theta) \| = o_P(1)
\end{align}
for each given $\rho > 0$.

Now, by the triangular inequality,
\begin{align}
    |Q_{n,\rho}(\theta)-Q_{n,\rho}^*(\theta)|
    & \le \left| \bar q_n(\theta)^\top \Psi_{n,\rho}(\theta)\bar q_n(\theta) - \varphi_n(\theta)^\top \Psi_{n,\rho}(\theta) \varphi_n(\theta) \right| \\
    & \quad + \left| \varphi_n(\theta)^\top \left( \Psi_{n,\rho}(\theta)-\Psi_{n,\rho}^*(\theta) \right) \varphi_n(\theta) \right|\\
    & \le \left| \left(\bar q_n(\theta) - \varphi_n(\theta) \right)^\top \Psi_{n,\rho}(\theta)\bar q_n(\theta) \right| + \left| \varphi_n(\theta)^\top \Psi_{n,\rho}(\theta)\left(\bar q_n(\theta) - \varphi_n(\theta) \right) \right| \\
    & \quad + \left| \varphi_n(\theta)^\top \left( \Psi_{n,\rho}(\theta)-\Psi_{n,\rho}^*(\theta) \right) \varphi_n(\theta) \right|.
\end{align}
In addition, by Assumptions \ref{as:unifmom}(i) and \ref{as:unifmom}(ii), we have $\|\bar q_n(\theta)\| \le \|\bar q_n(\theta) - \varphi_n(\theta) \| + \|\varphi_n(\theta) \| = O_P(1)$ uniformly in $\theta \in \Theta$.
Hence,
\begin{align}
    \sup_{\theta\in\Theta} \left| \left(\bar q_n(\theta) - \varphi_n(\theta) \right)^\top \Psi_{n,\rho}(\theta)\bar q_n(\theta) \right|
    & \le \sup_{\theta\in\Theta}\|\bar q_n(\theta) - \varphi_n(\theta) \| \sup_{\theta\in\Theta}\|\Psi_{n,\rho}(\theta)\|_\mathrm{op} \sup_{\theta\in\Theta} \|\bar q_n(\theta)\| \\
    & = o_P(1).
\end{align}
Analogously, the second term is $o_P(1)$ as well.
In addition, the third term is $o_P(1)$ by Assumption \ref{as:unifmom}(i) and \eqref{eq:Psi_unif}.
Therefore, $\sup_{\theta\in\Theta}|Q_{n,\rho}(\theta)-Q_{n,\rho}^*(\theta)| = o_P(1)$, as claimed.
\end{proof}

%%%%%%%%%%%%%%%%%%%%%%%%%%%%%%%%%%%%%%%%%%%%%%

\begin{flushleft}
\textbf{Proof of Theorem \ref{thm:consistency}}
\end{flushleft}

Under the uniform convergence result in Lemma \ref{lem:unifconv} and Assumption \ref{as:ident}, the result follows from the standard M-estimation theory.
\qed

%%%%%%%%%%%%%%%%%%%%%%%%%%%%%%%%%%%%%%%%%%%%%%

\subsection{Proof of Theorem \ref{thm:local}}

\begin{lemma}\label{lem:rootn_bias}
Under Assumptions \ref{as:popmom}, \ref{as:Psi}, and \ref{as:K}, we have 
\begin{align}
    \theta_{n,\rho}^* - \theta_0 = O(n^{-1/2}).
\end{align}
\end{lemma}

\begin{proof}
Fix $\bm c \in \bbR^{d_\theta}$ and set $\theta = \theta_0 + n^{-1/2}\bm c$.
By Taylor expansion and Assumption \ref{as:popmom},
\begin{align}
    \varphi_n(\theta)
    & = \varphi_n(\theta_0) + \dot\varphi_n(\theta_0) (\theta - \theta_0) + o(\|\theta-\theta_0\|) \\
    & = n^{-1/2}\left(\delta_n + \dot\varphi_n(\theta_0)\bm c + o(1)\right).
\end{align}
By Assumption \ref{as:Psi}, $\Psi_{n,\rho}^*(\theta) = \Psi_{n,\rho}^*(\theta_0) + o(1)$.
Hence,
\begin{align}
    n Q_{n,\rho}^*(\theta)
    & = n \varphi_n(\theta)^\top \Psi_{n,\rho}^*(\theta) \varphi_n(\theta) \\
    & = \left[ \delta_n + \dot\varphi_n(\theta_0) \bm c + o(1) \right]^\top \left\{ \Psi_{n,\rho}^*(\theta_0) + o(1) \right\} \left[ \delta_n + \dot\varphi_n(\theta_0) \bm c + o(1) \right] \\
    & = \left[ \delta_n + \dot\varphi_n(\theta_0) \bm c \right]^\top \Psi_{n,\rho}^*(\theta_0) \left[ \delta_n + \dot\varphi_n(\theta_0) \bm c \right] + o(1) \\
    & = \delta_n ^\top \Psi_{n,\rho}^*(\theta_0) \delta_n + 2 \delta_n^\top \Psi_{n,\rho}^*(\theta_0)  \dot\varphi_n(\theta_0) \bm c + \bm c^\top K_{n,\rho} \bm c + o(1),
\end{align}
where $K_{n,\rho} \coloneqq \dot\varphi_n(\theta_0)^\top \Psi_{n,\rho}^*(\theta_0) \dot\varphi_n(\theta_0)$.
Since $Q_{n,\rho}^*(\theta_0) = n^{-1} \delta_n^\top \Psi_{n,\rho}^*(\theta_0) \delta_n$, it follows that
\begin{align}
    n \left\{Q_{n,\rho}^*(\theta) - Q_{n,\rho}^*(\theta_0)\right\} = 2 \delta_n^\top \Psi_{n,\rho}^*(\theta_0) \dot\varphi_n(\theta_0) \bm c + \bm c^\top K_{n,\rho} \bm c + o(1).
\end{align}
The first term is of order $O(\|\bm c\|)$, whereas the second term is $\bm c^\top K_{n,\rho} \bm c \ge \lambda_{\min}(K_{n,\rho})\|\bm c\|^2$. 
Thus, Assumption \ref{as:K} implies that, for sufficiently large $C > 0$, we have
\begin{align}
    \inf_{\|\bm c\| = C} \left\{Q_{n,\rho}^*(\theta_0 + n^{-1/2} \bm c) - Q_{n,\rho}^*(\theta_0)\right\} > 0
\end{align}
for all sufficiently large $n$, which implies that the minimizer $\theta_{n,\rho}^*$ of $Q_{n,\rho}^*(\theta)$ must satisfy $\theta_{n,\rho}^* - \theta_0 = O(n^{-1/2})$.
\end{proof}

%%%%%%%%%%%%%%%%%%%%%%%%%%%%%%%%%%%%%%%%%%%%%%

\bigskip

\begin{flushleft}
\textbf{Proof of Theorem \ref{thm:local}}
\end{flushleft}

Recall $Q_{n,\rho}(\theta) = \bar q_n(\theta)^\top \Psi_{n,\rho}(\theta) \bar q_n(\theta)$, and $Q_{n,\rho}^*(\theta) = \varphi_n(\theta)^\top \Psi_{n,\rho}^*(\theta) \varphi_n(\theta)$, and let $S_{n,\rho}(\theta) \coloneqq \partial_\theta Q_{n,\rho}(\theta)$, and $S_{n,\rho}^*(\theta) \coloneqq \partial_\theta Q_{n,\rho}^*(\theta)$.
The $(j,k)$-th element of $\partial_{\theta^\top} S_{n,\rho}^*(\theta)$ is
\begin{align}
    (\partial_{\theta^\top} S_{n,\rho}^*(\theta))_{jk}
    & = 2 \partial_{\theta_j}\varphi_n(\theta)^\top \Psi_{n,\rho}^*(\theta) \partial_{\theta_k}\varphi_n(\theta) + 2 \partial_{\theta_j\theta_k}\varphi_n(\theta)^\top \Psi_{n,\rho}^*(\theta) \varphi_n(\theta) \\
    & \quad + 2 \partial_{\theta_j}\varphi_n(\theta)^\top \partial_{\theta_k}\Psi_{n,\rho}^*(\theta) \varphi_n(\theta) + 2 \partial_{\theta_k}\varphi_n(\theta)^\top \partial_{\theta_j}\Psi_{n,\rho}^*(\theta) \varphi_n(\theta) \\
    & \quad + \varphi_n(\theta)^\top \partial_{\theta_j\theta_k}\Psi_{n,\rho}^*(\theta) \varphi_n(\theta).
\end{align}
Hence, we can write
\begin{align}
    \partial_{\theta^\top} S_{n,\rho}^*(\theta)
    =
    2K_{n,\rho}
    +
    R_{n,\rho}^*(\theta),
\end{align}
where the $(j,k)$-th element of $R_{n,\rho}^*(\theta)$ is
\begin{align}
    (R_{n,\rho}^*(\theta))_{jk}
    & = 2\left( \partial_{\theta_j}\varphi_n(\theta)^\top \Psi_{n,\rho}^*(\theta) \partial_{\theta_k}\varphi_n(\theta) - \partial_{\theta_j}\varphi_n(\theta_0)^\top \Psi_{n,\rho}^*(\theta_0) \partial_{\theta_k}\varphi_n(\theta_0) \right) \\
    & \quad + 2 \partial_{\theta_j\theta_k}\varphi_n(\theta)^\top \Psi_{n,\rho}^*(\theta) \varphi_n(\theta) \\
    & \quad + 2 \partial_{\theta_j}\varphi_n(\theta)^\top \partial_{\theta_k}\Psi_{n,\rho}^*(\theta) \varphi_n(\theta) + 2 \partial_{\theta_k}\varphi_n(\theta)^\top \partial_{\theta_j}\Psi_{n,\rho}^*(\theta) \varphi_n(\theta) \\
    &\quad + \varphi_n(\theta)^\top \partial_{\theta_j\theta_k}\Psi_{n,\rho}^*(\theta) \varphi_n(\theta).
\end{align}
Now, consider $\theta$ such that $\|\theta-\theta_0\| = o(1)$.
By Assumption \ref{as:popmom}, $\varphi_n(\theta) = o(1)$ and $\partial_{\theta_j}\varphi_n(\theta) = \partial_{\theta_j}\varphi_n(\theta_0) + o(1)$ for $j \in [d_\theta]$.
Moreover, by Assumption \ref{as:Psi}, $\Psi_{n,\rho}^*(\theta) = \Psi_{n,\rho}^*(\theta_0) + o(1)$, $\partial_{\theta_j}\Psi_{n,\rho}^*(\theta) = O(1)$, $\partial_{\theta_j\theta_k}\Psi_{n,\rho}^*(\theta) = O(1)$ for $j,k \in [d_\theta]$.
Then, since $\partial_{\theta_j\theta_k}\varphi_n(\theta) = O(1)$ by Assumption \ref{as:popmom}(ii), each term in $(R_{n,\rho}^*(\theta))_{jk}$ is $o(1)$ uniformly over $\|\theta-\theta_0\| = o(1)$.

\bigskip

Next, since $S_{n,\rho}^*(\theta_{n,\rho}^*) = \bm 0$, a mean-value expansion at $\theta_0$ gives
\begin{align}
    \bm 0
    = S_{n,\rho}^*(\theta_0) + \partial_{\theta^\top} S_{n,\rho}^*(\bar \theta_n^*)(\theta_{n,\rho}^*-\theta_0),
\end{align}
where $\bar \theta_n^* \in [\theta_0, \theta_{n,\rho}^*]$.
By Lemma \ref{lem:rootn_bias}, $\bar\theta_n^* - \theta_0 = O(n^{-1/2})$, and thus $\partial_{\theta^\top} S_{n,\rho}^*(\bar \theta_n^*) = 2 K_{n,\rho} + o(1)$.
Also,
\begin{align}
    S_{n,\rho}^*(\theta_0)
    & = 2 \dot\varphi_n(\theta_0)^\top \Psi_{n,\rho}^*(\theta_0)\varphi_n(\theta_0) +  (\varphi_n(\theta_0)^\top \partial_{\theta_j}\Psi_{n,\rho}^*(\theta_0)\varphi_n(\theta_0))_{j \in [d_\theta]} \\
    & = 2 \dot\varphi_n(\theta_0)^\top \Psi_{n,\rho}^*(\theta_0)\varphi_n(\theta_0) + O(n^{-1})
\end{align}
by Assumptions \ref{as:popmom}(i) and \ref{as:Psi}.
Hence,
\begin{align}
    S_{n,\rho}^*(\theta_0)
    =
    \frac{2}{\sqrt{n}}\dot\varphi_n(\theta_0)^\top \Psi_{n,\rho}^*(\theta_0)\delta_n
    +
    o(n^{-1/2}),
\end{align}
and therefore
\begin{align}\label{eq:bias}
    \sqrt{n}(\theta_{n,\rho}^*-\theta_0)
    = -K_{n,\rho}^{-1}\dot\varphi_n(\theta_0)^\top \Psi_{n,\rho}^*(\theta_0) \delta_n + o(1).
\end{align}

\bigskip

Meanwhile, since $S_{n,\rho}(\hat\theta_{n,\rho}) = \bm 0$, a mean-value expansion at $\theta_{n,\rho}^*$ yields
\begin{align}
    \bm 0 = S_{n,\rho}(\theta_{n,\rho}^*) + \partial_{\theta^\top} S_{n,\rho}(\bar\theta_n)(\hat\theta_{n,\rho} -\theta_{n,\rho}^*),
\end{align}
where $\bar \theta_n \in [\hat\theta_{n,\rho}, \theta_{n,\rho}^*]$.
Noting that $\bar \theta_n - \theta_0 = o_P(1)$ by Theorem \ref{thm:consistency} and Lemma \ref{lem:rootn_bias}, we have
\begin{align}
    \partial_{\theta^\top} S_{n,\rho}(\bar\theta_n)
    = \partial_{\theta^\top} S_{n,\rho}^*(\bar\theta_n) + o_P(1) = 2 K_{n,\rho} + o_P(1)
\end{align}
by Assumption \ref{as:LLN}.
Similarly, expanding $S_{n,\rho}(\theta_{n,\rho}^*)$ around $\theta_0$, we have
\begin{align}
    \sqrt{n} S_{n,\rho}(\theta_{n,\rho}^*)
    & = \sqrt{n} S_{n,\rho}(\theta_0) + \partial_{\theta^\top} S_{n,\rho} (\check\theta_n) \sqrt{n}(\theta_{n,\rho}^*-\theta_0) \\
    & = \sqrt{n} S_{n,\rho}(\theta_0) + [2 K_{n,\rho} + o_P(1)] \sqrt{n} (\theta_{n,\rho}^* - \theta_0)  \\
    & = \sqrt{n} S_{n,\rho}(\theta_0) - 2 \dot\varphi_n(\theta_0)^\top \Psi_{n,\rho}^*(\theta_0)\delta_n + o_P(1),
\end{align}
where $\check\theta_n \in [\theta_0, \theta_{n,\rho}^*]$, and the last equality is from \eqref{eq:bias}.
Noting that $\bar q_n(\theta_0)=O_P(n^{-1/2})$ by Assumptions \ref{as:popmom} and \ref{as:CLT}, it is straightforward to see that 
\begin{align}
    \sqrt{n} S_{n,\rho}(\theta_0)
    & = 2\{\partial_{\theta^\top} \bar q_n(\theta_0)\}^\top \Psi_{n,\rho}(\theta_0)\sqrt{n}\bar q_n(\theta_0) + o_P(1) \\
    & = 2\dot\varphi_n(\theta_0)^\top \Psi_{n,\rho}^*(\theta_0) \sqrt{n} \bar q_n(\theta_0) + o_P(1) \\
    & = 2\dot\varphi_n(\theta_0)^\top \Psi_{n,\rho}^*(\theta_0) \sqrt{n} \left(\bar q_n(\theta_0)-\varphi_n(\theta_0)\right) + 2\dot\varphi_n(\theta_0)^\top \Psi_{n,\rho}^*(\theta_0)\delta_n + o_P(1)
\end{align}
by Assumption \ref{as:LLN}.
Therefore,
\begin{align}
    \sqrt{n} S_{n,\rho}(\theta_{n,\rho}^*) = 2\dot\varphi_n(\theta_0)^\top \Psi_{n,\rho}^*(\theta_0)
    \sqrt{n} \left(\bar q_n(\theta_0)-\varphi_n(\theta_0)\right) + o_P(1),
\end{align}
which leads to
\begin{align}
    \sqrt{n}(\hat\theta_{n,\rho} - \theta_{n,\rho}^*) = -K_{n,\rho}^{-1} \dot\varphi_n(\theta_0)^\top \Psi_{n,\rho}^*(\theta_0) \sqrt{n}\left(\bar q_n(\theta_0) - \varphi_n(\theta_0)\right) + o_P(1).
\end{align}
Consequently, by Assumption \ref{as:CLT},
\begin{align}
    \sqrt{n}(\hat\theta_{n,\rho} - \theta_{n,\rho}^*) \overset{d}{\to} N\left(\bm 0, K_\rho^{-1}\dot\varphi^\top \Psi_\rho^* \Sigma_q \Psi_\rho^* \dot\varphi K_\rho^{-1} \right).
\end{align}
The second result follows from \eqref{eq:bias}.

\qed

%%%%%%%%%%%%%%%%%%%%%%%%%%%%%%%%%%%%%%%%%%%%%%

\section{Proofs of Lemma \ref{lem:verify} and Theorem \ref{thm:SAR}}

Recall that
\begin{align}
    \hat \theta_{n,\rho} & = \argmin_{\theta \in \Theta} Q_{n,\rho}(\theta) \\
        &\quad Q_{n,\rho}(\theta) = \|\bar q_n(\theta) \|^2_{\Psi_{n, \rho}(\alpha)}, \quad \bar q_n(\theta) = \frac{1}{n} \bm Z_n^\top (\bm Y_n - \bm R_n^{\mathrm{obs}}\theta), \quad \Psi_{n, \rho}(\alpha) = \left(\Omega_n^{-1} + \alpha^2 \Xi_{n,m,\rho} \right)^{-1}, \\
    \theta_{n,\rho}^* & = \argmin_{\theta \in \Theta} Q_{n,\rho}^*(\theta) \\
        &\quad Q_{n,\rho}^*(\theta) = \|\varphi_n(\theta) \|^2_{\Psi_{n, \rho}^*(\alpha)}, \quad \varphi_n(\theta) = \frac{1}{n} \bm Z_n^\top (\bm M_n \bm X_n \beta_0 - \bbE[\bm R_n^{\mathrm{obs}}]\theta), \quad \Psi_{n, \rho}^*(\alpha) = \left(\Omega_n^{-1} + \alpha^2 \bbE[\Xi_{n,m,\rho}] \right)^{-1},
\end{align}
where $\bm M_n = (I_n - \alpha_0 \bm G_n)^{-1}$, $\bm R_n^{\mathrm{obs}} = (\bm G_n^{\mathrm{obs}}\bm Y_n,\bm X_n)$, and $\Xi_{n,m,\rho} = (m/\rho) \Phi_n$.
Recall also that $\Pi_n =  n^{-1}\bm Z_n^\top\bm R_n^{\mathrm{obs}} = (\pi_{n,1}, \Pi_{n, -1})$, and since $\Pi_{n, -1}$ is non-random by Assumption \ref{as:X}, 
\begin{align}
    \Pi_n - \bbE[\Pi_n] = (\pi_{n,1} - \bbE[\pi_{n,1}], \bm 0, \ldots, \bm 0)
\end{align}
holds.

The first and second derivatives of $\Psi_{n,\rho}(\alpha)$ with respect to $\alpha$ are
\begin{align}
    \dot \Psi_{n,\rho}(\alpha) 
    & \coloneqq \partial_\alpha \Psi_{n,\rho}(\alpha) \\
    & = -2\alpha \Psi_{n,\rho}(\alpha) \Xi_{n,m,\rho} \Psi_{n,\rho}(\alpha) \\
    \ddot \Psi_{n,\rho}(\alpha) 
    & \coloneqq \partial_{\alpha\alpha} \Psi_{n,\rho}(\alpha) \\
    &= -2 \Psi_{n,\rho}(\alpha) \Xi_{n,m,\rho} \Psi_{n,\rho}(\alpha) + 8\alpha^2 \Psi_{n,\rho}(\alpha) \Xi_{n,m,\rho} \Psi_{n,\rho}(\alpha) \Xi_{n,m,\rho} \Psi_{n,\rho}(\alpha).
\end{align}
Similarly, define
\begin{align}
    \dot \Psi_{n,\rho}^*(\alpha) 
    & \coloneqq -2\alpha \Psi_{n,\rho}^*(\alpha) \bbE[\Xi_{n,m,\rho}] \Psi_{n,\rho}^*(\alpha) \\
    \ddot \Psi_{n,\rho}^*(\alpha) 
    & \coloneqq -2 \Psi_{n,\rho}^*(\alpha) \bbE[\Xi_{n,m,\rho}] \Psi_{n,\rho}^*(\alpha) + 8 \alpha^2 \Psi_{n,\rho}^*(\alpha) \bbE[\Xi_{n,m,\rho}] \Psi_{n,\rho}^*(\alpha) \bbE[\Xi_{n,m,\rho}] \Psi_{n,\rho}^*(\alpha).
\end{align}

Letting $e_1=(1,0,\ldots,0)^\top\in\bbR^{d_\theta}$, the score function is
\begin{align}
    S_{n,\rho}(\theta)
    & \coloneqq \partial_\theta Q_{n,\rho}(\theta)\\
    & = -2 \Pi_n^\top \Psi_{n,\rho}(\alpha) \bar q_n(\theta) + e_1 \bar q_n(\theta)^\top \dot \Psi_{n,\rho}(\alpha) \bar q_n(\theta).
\end{align}
Moreover, the Hessian matrix is
\begin{align}
    H_{n,\rho}(\theta)
    & \coloneqq \partial_{\theta^\top} S_{n,\rho}(\theta) \\
    & = 2 \Pi_n^\top \Psi_{n,\rho}(\alpha) \Pi_n - 2 \Pi_n^\top \dot \Psi_{n,\rho}(\alpha) \bar q_n(\theta) e_1^\top - 2 e_1 \bar q_n(\theta)^\top \dot \Psi_{n,\rho}(\alpha) \Pi_n + e_1 e_1^\top \bar q_n(\theta)^\top \ddot \Psi_{n,\rho}(\alpha) \bar q_n(\theta).
\end{align}
The population counterparts of $S_{n,\rho}(\theta)$ and $H_{n,\rho}(\theta)$ are given by
\begin{align}
    S_{n,\rho}^*(\theta)
    & \coloneqq -2 \bbE[ \Pi_n ]^\top \Psi_{n,\rho}^*(\alpha) \varphi_n(\theta) + e_1 \varphi_n(\theta)^\top \dot \Psi_{n,\rho}^*(\alpha) \varphi_n(\theta)
\end{align}
and
\begin{align}\label{eq:hessian}
    H_{n,\rho}^*(\theta)
    & \coloneqq 2 \bbE[ \Pi_n ]^\top \Psi_{n,\rho}^*(\alpha) \bbE[ \Pi_n ] - 2 \bbE[ \Pi_n ]^\top \dot \Psi_{n,\rho}^*(\alpha) \varphi_n(\theta) e_1^\top - 2 e_1 \varphi_n(\theta)^\top \dot \Psi_{n,\rho}^*(\alpha) \bbE[\Pi_n] \\
    & \quad + e_1 e_1^\top \varphi_n(\theta)^\top \ddot \Psi_{n,\rho}^*(\alpha) \varphi_n(\theta),
\end{align}
respectively.

%%%%%%%%%%%%%%%%%%%%%%%%%%%%%%%%%%%%%%%%%%%%%

\begin{flushleft}
\textbf{Proof of Lemma \ref{lem:verify}}
\end{flushleft} 

\paragraph{Verification of Assumption \ref{as:ident}}

Observe that
\begin{align}
    Q_{n, \rho}^*(\theta) 
    & = \|\varphi_n(\theta) \|^2_{\Psi_{n, \rho}^*(\alpha)} \\
    & = \left(\frac{1}{n} \bm Z_n^\top (\bm M_n \bm X_n \beta_0 - \bbE[\bm R_n^{\mathrm{obs}}]\theta)\right)^\top \Psi_{n, \rho}^*(\alpha)\left(\frac{1}{n} \bm Z_n^\top (\bm M_n \bm X_n \beta_0 - \bbE[\bm R_n^{\mathrm{obs}}]\theta)\right) \\
    & = \left(\bm L_n - \bbE[\Pi_n]\theta \right)^\top \Psi_{n, \rho}^*(\alpha) \left(\bm L_n - \bbE[\Pi_n]\theta \right) \\
    & = \left(\bm L_n - \alpha \bbE[\pi_{n,1}] - \Pi_{n,-1}\beta \right)^\top \Psi_{n,\rho}^*(\alpha) \left(\bm L_n - \alpha \bbE[\pi_{n,1}] - \Pi_{n,-1}\beta \right),
\end{align}
where $\bm L_n \coloneqq n^{-1} \bm Z_n^\top \bm M_n \bm X_n \beta_0$.
By Assumption \ref{as:identsar}(i), the minimizer with respect to $\beta$ exists uniquely for each $\alpha \in \Theta_\alpha$ and is given by
\begin{align}
    \beta_{n,\rho}^{\circ}(\alpha)
    \coloneqq
    \left(\Pi_{n,-1}^\top \Psi_{n,\rho}^*(\alpha)\Pi_{n,-1}\right)^{-1}
    \Pi_{n,-1}^\top \Psi_{n,\rho}^*(\alpha) \left(\bm L_n-\alpha \bbE[\pi_{n,1}]\right).
\end{align}
Moreover, note that
\begin{align}
    Q_{n,\rho}^*(\alpha,\beta)
    & = Q_{n,\rho}^*(\alpha,\beta_{n,\rho}^{\circ}(\alpha)) + \left(\beta - \beta_{n,\rho}^{\circ}(\alpha)
    \right)^\top \Pi_{n,-1}^\top \Psi_{n,\rho}^*(\alpha) \Pi_{n,-1} \left( \beta - \beta_{n,\rho}^{\circ}(\alpha) \right) \\
    & \quad - 2 \underbracket{\left(\beta - \beta_{n,\rho}^{\circ}(\alpha) \right)^\top \Pi_{n,-1}^\top \Psi_{n,\rho}^*(\alpha) \left(\bm L_n - \alpha \bbE[\pi_{n,1}] - \Pi_{n,-1}\beta_{n,\rho}^{\circ}(\alpha) \right)}_{= 0}.
\end{align}
Hence, by Assumption \ref{as:identsar}(i),
\begin{align}\label{eq:beta_ident}
    Q_{n,\rho}^*(\alpha,\beta) - Q_{n,\rho}^*(\alpha,\beta_{n,\rho}^{\circ}(\alpha)) \ge c_\beta \left\| \beta - \beta_{n,\rho}^{\circ}(\alpha) \right\|^2 
\end{align}
for all sufficiently large $n$.

Now, suppose that $\|\theta-\theta_{n,\rho}^*\|\ge\epsilon$.
If $|\alpha-\alpha_{n,\rho}^*|\ge\epsilon_1$, Assumption \ref{as:identsar}(ii) gives
\begin{align}
    Q_{n,\rho}^*(\alpha, \beta) -  Q_{n,\rho}^*(\alpha_{n,\rho}^*,\beta_{n,\rho}^*) \ge Q_{n,\rho}^*(\alpha,\beta_{n,\rho}^{\circ}(\alpha)) -  Q_{n,\rho}^*(\alpha_{n,\rho}^*,\beta_{n,\rho}^*) \ge c_{\epsilon_1}.
\end{align}
On the other hand, suppose that $|\alpha-\alpha_{n,\rho}^*|<\epsilon_1$.
By the continuity of $\beta_{n,\rho}^{\circ}(\alpha)$, choose $\epsilon_1>0$ sufficiently small so that
\begin{align}
    \left\|
        \beta_{n,\rho}^{\circ}(\alpha)
        -
        \beta_{n,\rho}^{\circ}(\alpha_{n,\rho}^*)
    \right\|
    < \epsilon_2 
\end{align}
for some small $\epsilon_2 > 0$.
Since $\beta_{n,\rho}^{\circ}(\alpha_{n,\rho}^*)=\beta_{n,\rho}^*$ and
$\|\theta-\theta_{n,\rho}^*\|\ge\epsilon$, it follows that $\| \beta - \beta_{n,\rho}^{\circ}(\alpha) \| \ge \epsilon_3$ for some $\epsilon_3 > 0$.
Therefore, by \eqref{eq:beta_ident} and the fact that $Q_{n,\rho}^*(\alpha,\beta_{n,\rho}^{\circ}(\alpha))
\ge Q_{n,\rho}^*(\theta_{n,\rho}^*)$,
\begin{align}
    Q_{n,\rho}^*(\alpha,\beta) - Q_{n,\rho}^*(\theta_{n,\rho}^*)
    \ge c_\beta\epsilon_3^2 
\end{align}
for any $\theta \in \Theta$ such that $\|\theta-\theta_{n,\rho}^*\| \ge \epsilon$.
Hence, Assumption \ref{as:ident} holds.

\paragraph{Verification of Assumption \ref{as:unifmom}}

Recall that
\begin{align}
    \varphi_n(\theta) 
    & = \frac{1}{n}\bm Z_n^\top\left(\bm M_n\bm X_n\beta_0 - \bbE[\bm R_n^{\mathrm{obs}}]\theta\right) \\
    & = \frac{1}{n}\bm Z_n^\top\left( \bm M_n\bm X_n\beta_0 - \left\{\bm G_n^{\mathrm{obs}} \bm M_n\bm X_n\beta_0, \bm X_n \right\}\theta\right) .
\end{align}
Then, by the boundedness of $\bm Z_n$, $\bm M_n$, and $\bm G_n^{\mathrm{obs}}$, together with the compactness of $\Theta$, Assumption \ref{as:unifmom}(i) is satisfied.

\medskip

Next, uniformly in $\theta \in\Theta$,
\begin{align}
    \bar q_n(\theta)-\varphi_n(\theta) 
    & = \frac{1}{n}\bm Z_n^\top\left( \bm Y_n - \bm M_n\bm X_n\beta_0 - \left\{ \bm R_n^{\mathrm{obs}} - \bbE[\bm R_n^{\mathrm{obs}}] \right\}\theta\right) \\
    & = \frac{1}{n}\bm Z_n^\top \bm M_n\bm\varepsilon_n - \frac{\alpha}{n}\bm Z_n^\top \bm G_n^{\mathrm{obs}}\bm M_n \bm\varepsilon_n.
\end{align}
Under Assumptions \ref{as:error} and \ref{as:network}, it is easy to see that $\bbE \|\bm Z_n^\top \bm M_n\bm\varepsilon_n/n\|^2 = O(n^{-1})$ and  $\bbE \|\bm Z_n^\top \bm G_n^{\mathrm{obs}} \bm M_n \bm\varepsilon_n/n\|^2 = O(n^{-1})$.
Thus, Assumption \ref{as:unifmom}(ii) follows from Markov's inequality.

\medskip

In the SAR model, $V_j(\bm W_n,\theta)=-\alpha Y_j$.
Hence, to verify Assumption \ref{as:unifmom}(iii), it suffices to show
\begin{align}
    \left|\frac{1}{n}\sum_{j\in[n]}\bar a_{n,j}(Y_j^2-\bbE[Y_j^2])\right|=o_P(m^{-1}).
\end{align}
By the definition of $\bar a_{n,j}$,
\begin{align}\label{eq:aorder}
    |\bar a_{n,j}|
    \le
    \frac{C}{n}\left|\{i \in [n]: \: j \in \calU(i)\} \right| = O(n^{-1}),
\end{align}
by Assumption \ref{as:ubound}.
Let $M_{j.}$ be the $j$-th row of $\bm M_n$, and write $Y_j = M_{j.} \bm X_n \beta_0 + M_{j.} \bm \varepsilon_n$.
Then,
\begin{align}\label{eq:Y2}
    \begin{split}
    Y_j^2
    & = (M_{j.} \bm X_n \beta_0)^2 + 2\beta_0^\top \bm X_n^\top M_{j.}^\top M_{j.} \bm\varepsilon_n +
    \bm\varepsilon_n^\top M_{j.}^\top M_{j.}\bm\varepsilon_n \\
    \bbE[ Y_j^2 ]
    & = (M_{j.} \bm X_n \beta_0)^2 + \bbE[\bm\varepsilon_n^\top M_{j.}^\top M_{j.}\bm\varepsilon_n].
    \end{split}
\end{align}
Then,
\begin{align}
    \frac{1}{n}\sum_{j \in [n]}\bar a_{n,j}(Y_j^2-\bbE[Y_j^2])
    & =
    \frac{2}{n}\sum_{j \in [n]} \bar a_{n,j} \beta_0^\top \bm X_n^\top M_{j.}^\top M_{j.} \bm\varepsilon_n \\
    & \quad + \frac{1}{n} \sum_{j \in [n]} \bar a_{n,j} \left\{ \bm\varepsilon_n^\top M_{j.}^\top M_{j.}\bm\varepsilon_n - \bbE[\bm\varepsilon_n^\top M_{j.}^\top M_{j.}\bm\varepsilon_n] \right\}.
\end{align}
For the first term, the coefficient of $\varepsilon_k$ is
\begin{align}
    \frac{2}{n}\sum_{j\in[n]}\bar a_{n,j}(M_{j.}\bm X_n\beta_0)M_{jk} = O(n^{-2})
\end{align}
by \eqref{eq:aorder} and Assumption \ref{as:network}(ii).
Hence, by Assumption \ref{as:error},
\begin{align}
    \frac{2}{n}\sum_{j\in[n]}\bar a_{n,j}\beta_0^\top\bm X_n^\top M_{j.}^\top M_{j.}\bm\varepsilon_n
    =
    O_P(n^{-3/2}).
\end{align}

For the second term, observe that the coefficient on $\varepsilon_k\varepsilon_\ell$ is
\begin{align}
    \frac{1}{n}\sum_{j \in [n]} \bar a_{n,j} M_{jk} M_{j\ell}.
\end{align}
Again, by \eqref{eq:aorder} and Assumption \ref{as:network}(ii), we have
\begin{align}
    \sum_{k=1}^n \sum_{\ell=1}^n \left( \frac{1}{n}\sum_{j \in [n]} \bar a_{n,j} M_{jk} M_{j\ell} \right)^2
    & \le \max_{k,\ell} \left| \frac{1}{n}\sum_{j\in[n]}\bar a_{n,j}M_{jk}M_{j\ell} \right| \cdot \sum_{k=1}^n \sum_{\ell=1}^n \left| \frac{1}{n}\sum_{j\in[n]}\bar a_{n,j}M_{jk}M_{j\ell} \right| \\
    & = O(n^{-2})\cdot O(n^{-1})=O(n^{-3}).
\end{align}
Under Assumption \ref{as:error}, the variance of the second term is
\begin{align}
    &\operatorname{Var}\left[
        \sum_{k=1}^n\sum_{\ell=1}^n \left( \frac{1}{n}\sum_{j\in[n]}\bar a_{n,j}M_{jk}M_{j\ell} \right) \left( \varepsilon_k\varepsilon_\ell-\bbE[\varepsilon_k\varepsilon_\ell] \right)
    \right] \\
    &\quad = \sum_{k=1}^n\sum_{\ell=1}^n \sum_{r=1}^n\sum_{s=1}^n \left( \frac{1}{n}\sum_{j\in[n]}\bar a_{n,j}M_{jk}M_{j\ell} \right) \left( \frac{1}{n}\sum_{j\in[n]}\bar a_{n,j}M_{jr}M_{js} \right) \operatorname{Cov} \left( \varepsilon_k\varepsilon_\ell, \varepsilon_r\varepsilon_s \right) \\
    &\quad \le C \sum_{k=1}^n\sum_{\ell=1}^n \left( \frac{1}{n}\sum_{j\in[n]}\bar a_{n,j}M_{jk}M_{j\ell} \right)^2 =  O(n^{-3}).
\end{align}
Thus,
\begin{align}
    \frac{1}{n}\sum_{j\in[n]}\bar a_{n,j} \left[ \bm\varepsilon_n^\top M_{j.}^\top M_{j.}\bm\varepsilon_n - \bbE[\bm\varepsilon_n^\top M_{j.}^\top M_{j.}\bm\varepsilon_n] \right] = O_P(n^{-3/2}).
\end{align}
Noting that Assumption \ref{as:ubound} implies $m/n=O(1)$, this completes the proof.

\paragraph{Verification of Assumption \ref{as:LLN}}

Since $\partial_{\theta^\top}\bar q_n(\theta)=-\Pi_n$,
\begin{align}
    \partial_{\theta^\top}\bar q_n(\theta)-\dot\varphi_n(\theta)
    = -(\Pi_n-\bbE[\Pi_n])
    = -\left(
        \frac{1}{n}\bm Z_n^\top\bm G_n^{\mathrm{obs}}\bm M_n\bm\varepsilon_n,\bm 0
    \right),
\end{align}
independent of $\theta$.
By Assumption \ref{as:error} and Markov's inequality,
\begin{align}
    \sup_{\theta\in\mathcal N_0}\left\|\partial_{\theta^\top}\bar q_n(\theta)-\dot\varphi_n(\theta)\right\|=O_P(n^{-1/2})=o_P(1),
\end{align}
which verifies the first part of Assumption \ref{as:LLN}(i).
The second part follows trivially.

\medskip

Next, since $\Psi_{n,\rho}(\theta)=\Psi_{n,\rho}(\alpha)$ and $\Psi^*_{n,\rho}(\theta)=\Psi^*_{n,\rho}(\alpha)$ are functions of $\alpha$ only, all derivatives with respect to $\beta$ are zero.
We first show that $\|\Xi_{n,m,\rho}-\bbE[\Xi_{n,m,\rho}]\|_{\mathrm{op}}=o_P(1)$.
By the definition of $\Gamma_n$ and $\Phi_n = \Gamma_n\Gamma_n^\top$,
\begin{align}
    \Phi_n = \frac{1}{n^2} \sum_{i \in [n]} \sum_{j \in \calU(i)} Z_i Z_i^\top Y_j^2 .
\end{align}
Thus, for each $(r,s)$,
\begin{align}
    \left( \Xi_{n,m,\rho}-\bbE[\Xi_{n,m,\rho}] \right)_{rs}
    =
    \frac{1}{n}\sum_{j\in[n]}
    \left(\frac{m}{\rho n}\sum_{i:j\in\calU(i)}Z_{ir}Z_{is}\right)
    (Y_j^2-\bbE[Y_j^2]).
\end{align}
By Assumption \ref{as:ubound}, 
\begin{align}\label{eq:boundZZ}
    \left|\frac{m}{\rho n}\sum_{i:j\in\calU(i)}Z_{ir}Z_{is}\right|=O(1)
\end{align}
uniformly in $j,r,s$.
Then, following the analogous arguments as in the verification of Assumption \ref{as:unifmom}(iii), we can find that
\begin{align}
    [\Xi_{n,m,\rho}-\bbE[\Xi_{n,m,\rho}]]_{rs} = O_P(n^{-1/2}),
\end{align}
and thus $\|\Xi_{n,m,\rho}-\bbE[\Xi_{n,m,\rho}]\|_{\mathrm{op}}=O_P(n^{-1/2})=o_P(1)$.

Since $\Xi_{n,m,\rho}$ and $\bbE[\Xi_{n,m,\rho}]$ are positive semidefinite, $\Psi_{n,\rho}(\alpha)$ and $\Psi^*_{n,\rho}(\alpha)$ are uniformly bounded in $\alpha$ by Assumption \ref{as:omega}.
Therefore,
\begin{align}
    \sup_{\theta\in\mathcal N_0}\|\Psi_{n,\rho}(\theta)-\Psi^*_{n,\rho}(\theta)\|_{\mathrm{op}} = \sup_{\theta\in\mathcal N_0}\| \Psi_{n,\rho}(\alpha)\left[\alpha^2(\bbE[\Xi_{n,m,\rho}]-\Xi_{n,m,\rho})\right]\Psi^*_{n,\rho}(\alpha) \|_{\mathrm{op}} =o_P(1),
\end{align}
which verifies the first part of Assumption \ref{as:LLN}(ii).
The preceding result and $\|\Xi_{n,m,\rho}-\bbE[\Xi_{n,m,\rho}]\|_{\mathrm{op}}=o_P(1)$ also imply uniform convergence of the first and second derivative differences.
This verifies Assumption \ref{as:LLN}(ii).

\paragraph{Verification of Assumption \ref{as:CLT}}

Assumption \ref{as:CLT} follows directly from Lemma \ref{lem:SAR_CLT}.
\qed

%%%%%%%%%%%%%%%%%%%%%%%%%%%%%%%%%%%%%%%%%%%%%

\begin{lemma}\label{lem:SAR_CLT}
    Suppose that Assumptions \ref{as:error} -- \ref{as:ubound} and \ref{as:variance} hold.
    Then, for each given $\rho > 0$,
    \begin{align}
        \sqrt{n}
        \begin{pmatrix}
            \bar q_n(\theta_0) - \varphi_n(\theta_0) \\
            \pi_{n,1} - \bbE[\pi_{n,1}] \\
            \operatorname{vec}(\Xi_{n,m,\rho} - \bbE[\Xi_{n,m,\rho}])
        \end{pmatrix}
        \overset{d}{\to} N(\bm 0, \Sigma_\rho^{\mathrm{SAR}}),
    \end{align}
    where $\Sigma_\rho^{\mathrm{SAR}}\in\bbR^{d_\mu(2+d_\mu)\times d_\mu(2+d_\mu)}$ is the symmetric matrix satisfying
    \begin{align}
        \omega_\rho^2(v)=v^\top\Sigma_\rho^{\mathrm{SAR}}v
    \end{align}
    for any non-zero $v\in\bbR^{d_\mu(2+d_\mu)}$.
\end{lemma}

\begin{proof}
First, noting that 
\begin{align}
    \bm Y_n-\bm R_n^{\mathrm{obs}}\theta_0 = \alpha_0(\bm G_n-\bm G_n^{\mathrm{obs}})\bm Y_n + \bm \varepsilon_n,
\end{align}
we have
\begin{align}
    \bar q_n(\theta_0)-\varphi_n(\theta_0)
    & = \frac{1}{n}\bm Z_n^\top \left[\alpha_0(\bm G_n-\bm G_n^{\mathrm{obs}}) \left\{\bm Y_n - \bbE[\bm Y_n]\right\} + \bm \varepsilon_n \right] \\
    & = \frac{1}{n}\bm Z_n^\top \left[\alpha_0(\bm G_n-\bm G_n^{\mathrm{obs}}) \bm M_n + I_n \right] \bm \varepsilon_n.
\end{align}

Second, since $\pi_{n,1} = n^{-1} \bm Z_n^\top \bm G_n^{\mathrm{obs}}\bm Y_n$,  we obtain
\begin{align}
    \pi_{n,1} - \bbE[\pi_{n,1}]
    & = \frac{1}{n}\bm Z_n^\top \bm G_n^{\mathrm{obs}} \left\{\bm Y_n-\bbE[\bm Y_n]\right\} \\
    & = \frac{1}{n}\bm Z_n^\top \bm G_n^{\mathrm{obs}} \bm M_n \bm \varepsilon_n.
\end{align}

Third, 
\begin{align}
    \operatorname{vec}\left( \Xi_{n,m,\rho} - \bbE[\Xi_{n,m,\rho}] \right)
    = \frac{1}{n} \sum_{j \in [n]}  \operatorname{vec}\left( \frac{m}{\rho n} \sum_{i: \: j\in\calU(i)}Z_iZ_i^\top \right) ( Y_j^2 - \bbE[Y_j^2] ).
\end{align}
In view of \eqref{eq:Y2},
\begin{align}
    \operatorname{vec}\left( \Xi_{n,m,\rho} - \bbE[\Xi_{n,m,\rho}] \right)_\ell
    & = \frac{2}{n} \sum_{j \in [n]}  \zeta_{n,j}^{(\ell)} \beta_0^\top \bm X_n^\top M_{j.}^\top M_{j.} \bm\varepsilon_n + \frac{1}{n} \sum_{j \in [n]}  \zeta_{n,j}^{(\ell)}\left\{ \bm\varepsilon_n^\top M_{j.}^\top M_{j.}\bm\varepsilon_n - \bbE[\bm\varepsilon_n^\top M_{j.}^\top M_{j.}\bm\varepsilon_n]\right\} \\
    & = 2 \beta_0^\top \bm X_n^\top \bm S_{n, \ell} \bm\varepsilon_n + \bm\varepsilon_n^\top \bm S_{n, \ell} \bm\varepsilon_n - \bbE[\bm\varepsilon_n^\top \bm S_{n, \ell} \bm\varepsilon_n],
\end{align}
where 
\begin{align}
    \zeta_{n,j}^{(\ell)}
    \coloneqq \operatorname{vec}\left( \frac{m}{\rho n} \sum_{i: \: j\in\calU(i)}Z_i Z_i^\top \right)_\ell,
    \quad \bm S_{n, \ell}
    \coloneqq  n^{-1} \sum_{j \in [n]} \zeta_{n,j}^{(\ell)} M_{j.}^\top M_{j.}.
\end{align}

To establish the joint asymptotic normality, we use the Cramer--Wold device.
Specifically, for an arbitrary conformable vector $v=(v_1^\top,v_2^\top,v_3^\top)^\top \in \bbR^{d_\mu + d_\mu + d_\mu^2}$ such that $\|v\| = 1$, define
\begin{align}
    T_n
    & \coloneqq v_1^\top \sqrt n \left(\bar q_n(\theta_0)-\varphi_n(\theta_0)\right) + v_2^\top \sqrt n
    \left(\pi_{n,1} - \bbE[\pi_{n,1}] \right) + v_3^\top \sqrt n \operatorname{vec} \left( \Xi_{n,m,\rho} - \bbE[\Xi_{n,m,\rho}] \right) \\
    & = \mathcal A_n \bm\varepsilon_n + \bm \varepsilon_n^\top \mathcal B_n \bm \varepsilon_n - \bbE[\bm\varepsilon_n^\top \mathcal B_n \bm\varepsilon_n],
\end{align}
where
\begin{align}\label{eq:varcomp}
    \begin{split}
    \mathcal A_n \equiv \mathcal A_n(v)
    & \coloneqq \frac{1}{\sqrt n} v_1^\top \bm Z_n^\top \left\{ \alpha_0(\bm G_n-\bm G_n^{\mathrm{obs}})\bm M_n + I_n \right\} + \frac{1}{\sqrt n} v_2^\top \bm Z_n^\top \bm G_n^{\mathrm{obs}}\bm M_n + 2\sqrt n \sum_{\ell \in [d_\mu^2]} v_{3,\ell}\beta_0^\top \bm X_n^\top \bm S_{n,\ell} \\
    \mathcal B_n \equiv \mathcal B_n(v)
    & \coloneqq \sqrt n \sum_{\ell \in [d_\mu^2]} v_{3,\ell} \bm S_{n,\ell}.
    \end{split}
\end{align}
Let
\begin{align}
    \mathcal B_n^s \coloneqq \frac{\mathcal B_n+\mathcal B_n^\top}{2} = (B_{n,ij}^s)_{i,j \in [n]},
    \qquad
    Q_n \coloneqq \sqrt{n} \mathcal A_n\bm\varepsilon_n  + \sqrt{n} \bm\varepsilon_n^\top \mathcal B_n^s \bm\varepsilon_n.
\end{align}
Note that $\bm\varepsilon_n^\top \mathcal B_n \bm\varepsilon_n=\bm\varepsilon_n^\top \mathcal B_n^s \bm\varepsilon_n$.
We now apply Theorem 1 of \cite{kelejian2001asymptotic} to $Q_n$.
Assumptions 1 and 3(b) of \cite{kelejian2001asymptotic} and the variance nondegeneracy condition $n^{-1} \operatorname{Var}(Q_n) \ge c > 0$ for all large $n$ are directly assumed in Assumptions \ref{as:error} and \ref{as:variance}, respectively.
Thus, it remains to verify their Assumption 2, which is sufficient to show that the elements of $\sqrt{n} \mathcal A_n$ are uniformly bounded and that $\sqrt{n} \| \mathcal B_n^s\|_\infty < \infty$.

First, we verify the condition for the linear part.
Let $e_i$ be the $i$-th unit vector in $\bbR^n$.
The $i$-th element $\sqrt n A_{n,i}$ of $\sqrt n\mathcal A_n$ is
\begin{align}
    \sqrt n A_{n,i}
    &=
    \alpha_0 v_1^\top \bm Z_n^\top \bm D_n\bm M_n e_i
    + v_1^\top Z_i
    + v_2^\top \bm Z_n^\top \bm G_n^{\mathrm{obs}}\bm M_n e_i
    + 2n\sum_{\ell\in[d_\mu^2]}v_{3,\ell}\beta_0^\top \bm X_n^\top \bm S_{n,\ell}e_i .
\end{align}
The first three terms are uniformly bounded by Assumptions \ref{as:X} and \ref{as:network}.
For the last term, by the definition of $\bm S_{n,\ell}$,
\begin{align}
    n\beta_0^\top \bm X_n^\top \bm S_{n,\ell} e_i = \sum_{j \in [n]} \zeta_{n,j}^{(\ell)} \left(M_{j.} \bm X_n \beta_0 \right) M_{ji}.
\end{align}
The term $M_{j.} \bm X_n \beta_0$ is uniformly bounded by Assumptions \ref{as:X} and \ref{as:network}.
Moreover $\zeta_{n,j}^{(\ell)}$ is uniformly bounded in $j$ as shown in \eqref{eq:boundZZ}.
As a result, we obtain $\sup_{i \in [n]: \: n \ge 1}|\sqrt n A_{n,i}|<\infty$.

Next, we verify the condition for the quadratic part.
Since $\sqrt n\mathcal B_n=n\sum_{\ell\in[d_\mu^2]}v_{3,\ell}\bm S_{n,\ell}$, it is enough to show $\| n\bm S_{n,\ell} \|_\infty < \infty$.
By the definition,
\begin{align}
    (n \bm S_{n,\ell})_{ik}
    = \sum_{j \in [n]} \zeta_{n,j}^{(\ell)}M_{ji}M_{jk},
    \quad \sum_{k \in [n]} \left|(n \bm S_{n,\ell})_{ik}\right|
    \le \sum_{j \in [n]} \left|\zeta_{n,j}^{(\ell)}\right| \cdot \left| M_{ji} \right| \sum_{k \in [n]} \left| M_{jk} \right| < \infty
\end{align}
by Assumption \ref{as:network} and the uniform boundedness of $\zeta_{n,j}^{(\ell)}$.
Thus, $\|\sqrt n\mathcal B_n\|_\infty<\infty$.
The same bound holds for $\sqrt n\mathcal B_n^\top$, and hence $\sqrt n\|\mathcal B_n^s\|_\infty<\infty$ uniformly in $n$.

Consequently, Theorem 1 of \cite{kelejian2001asymptotic} gives
\begin{align}
    \frac{Q_n-\bbE[Q_n]}{\{\operatorname{Var}(Q_n)\}^{1/2}}
    \overset{d}{\to}N(0,1).
\end{align}
Finally, since $T_n = n^{-1/2}\{Q_n-\bbE[Q_n]\}$ and $\operatorname{Var}(Q_n) = n\omega_{n,\rho}^2(v)$, it follows that
\begin{align}
    T_n \overset{d}{\to} N(0, \omega_\rho^2(v)),
\end{align}
where $\omega_\rho^2(v) = \lim_{n\to\infty} \omega_{n,\rho}^2(v)$.
Since $v = (v_1^\top, v_2^\top, v_3^\top)^\top$ is arbitrary as long as $\|v\| = 1$, the Cramer--Wold device implies the desired result.
\end{proof}

%%%%%%%%%%%%%%%%%%%%%%%%%%%%%%%%%%%%%%%%%%%%%%%%

\begin{flushleft}
\textbf{Proof of Theorem \ref{thm:SAR}}
\end{flushleft} 

By the consistency result, $\hat\theta_{n,\rho}-\theta_{n,\rho}^*=o_P(1)$.
The first-order conditions give $S_{n,\rho}(\hat\theta_{n,\rho}) = \bm 0$ and $S_{n,\rho}^*(\theta_{n,\rho}^*) = \bm 0$.
By a mean-value expansion around $\theta_{n,\rho}^*$,
\begin{align}
    \bm 0 =  S_{n,\rho}(\theta_{n,\rho}^*) - S_{n,\rho}^*(\theta_{n,\rho}^*) + H_{n,\rho}(\tilde\theta_n)(\hat\theta_{n,\rho} - \theta_{n,\rho}^*),
\end{align}
where $\tilde\theta_n \in [\hat\theta_{n,\rho}, \theta_{n,\rho}^*]$.
Therefore,
\begin{align}
    \sqrt n(\hat\theta_{n,\rho}-\theta_{n,\rho}^*) = - \left(H_{n,\rho}(\tilde\theta_n)\right)^{-1}
    \sqrt n \left[ S_{n,\rho}(\theta_{n,\rho}^*) - S_{n,\rho}^*(\theta_{n,\rho}^*)\right].
\end{align}
By Lemma \ref{lem:SAR_CLT},
\begin{align}
    \bar q_n(\theta_0)-\varphi_n(\theta_0)=O_P(n^{-1/2}),
    \quad
    \pi_{n,1}-\bbE[\pi_{n,1}]=O_P(n^{-1/2}),
    \quad
    \Xi_{n,m,\rho}-\bbE[\Xi_{n,m,\rho}]=O_P(n^{-1/2}).
\end{align}
Since $\Pi_{n,-1}$ is non-stochastic, this also implies $\Pi_n-\bbE[\Pi_n]=O_P(n^{-1/2})$.
Thus, by the continuous mapping theorem and Lemma \ref{lem:verify}, we obtain
\begin{align}
    H_{n,\rho}(\tilde\theta_n) - H_{n,\rho}^*(\theta_{n,\rho}^*) = o_P(1).
\end{align}

Next, observe that
\begin{align}
    S_{n,\rho}(\theta) - S_{n,\rho}^*(\theta)
    & = -2 \left( \Pi_n^\top \Psi_{n,\rho}(\alpha) \bar q_n(\theta) - \bbE[ \Pi_n ]^\top \Psi_{n,\rho}^*(\alpha) \varphi_n(\theta) \right) \\
    & \quad + e_1 \left(\bar q_n(\theta)^\top \dot \Psi_{n,\rho}(\alpha) \bar q_n(\theta) - \varphi_n(\theta)^\top \dot \Psi_{n,\rho}^*(\alpha) \varphi_n(\theta)\right)
\end{align}
We first decompose the first term.
At $\theta=\theta_{n,\rho}^*$,
\begin{align}
    \Pi_n^\top \Psi_{n,\rho}(\alpha_{n,\rho}^*) \bar q_n(\theta_{n,\rho}^*) 
    & = \left[ \bbE[\Pi_n]^\top + \left(\Pi_n-\bbE[\Pi_n]\right)^\top \right]
    \left[\Psi_{n,\rho}^*(\alpha_{n,\rho}^*) + \left(\Psi_{n,\rho}(\alpha_{n,\rho}^*) - \Psi_{n,\rho}^*(\alpha_{n,\rho}^*)\right)
    \right] \\
    &\quad \times \left[ \varphi_n(\theta_{n,\rho}^*) + \left( \bar q_n(\theta_{n,\rho}^*) - \varphi_n(\theta_{n,\rho}^*)\right) \right].
\end{align}
Subtracting $\bbE[\Pi_n]^\top\Psi_{n,\rho}^*(\alpha_{n,\rho}^*)\varphi_n(\theta_{n,\rho}^*)$ from this, we obtain
\begin{align}\label{eq:decomp1}
    \begin{split}
    &\Pi_n^\top \Psi_{n,\rho}(\alpha_{n,\rho}^*) \bar q_n(\theta_{n,\rho}^*)
    -\bbE[\Pi_n]^\top \Psi_{n,\rho}^*(\alpha_{n,\rho}^*) \varphi_n(\theta_{n,\rho}^*) \\
    &= \bbE[\Pi_n]^\top \Psi_{n,\rho}^*(\alpha_{n,\rho}^*) \left( \bar q_n(\theta_{n,\rho}^*)-\varphi_n(\theta_{n,\rho}^*)\right) \\
    &\quad + \left(\Pi_n - \bbE[\Pi_n]\right)^\top  \Psi_{n,\rho}^*(\alpha_{n,\rho}^*)\varphi_n(\theta_{n,\rho}^*) \\
    &\quad + \bbE[\Pi_n]^\top \left(\Psi_{n,\rho}(\alpha_{n,\rho}^*) - \Psi_{n,\rho}^*(\alpha_{n,\rho}^*)\right) \varphi_n(\theta_{n,\rho}^*) \\
    &\quad + \left( \Pi_n - \bbE[\Pi_n] \right)^\top \Psi_{n,\rho}^*(\alpha_{n,\rho}^*) \left(\bar q_n(\theta_{n,\rho}^*)-\varphi_n(\theta_{n,\rho}^*)\right) \\
    &\quad + \left( \Pi_n - \bbE[\Pi_n] \right)^\top \left(\Psi_{n,\rho}(\alpha_{n,\rho}^*) - \Psi_{n,\rho}^*(\alpha_{n,\rho}^*)\right) \varphi_n(\theta_{n,\rho}^*) \\
    &\quad + \bbE[\Pi_n]^\top \left(\Psi_{n,\rho}(\alpha_{n,\rho}^*) - \Psi_{n,\rho}^*(\alpha_{n,\rho}^*)\right) \left(\bar q_n(\theta_{n,\rho}^*) - \varphi_n(\theta_{n,\rho}^*) \right) \\
    &\quad + \left(\Pi_n - \bbE[\Pi_n] \right)^\top \left(\Psi_{n,\rho}(\alpha_{n,\rho}^*) - \Psi_{n,\rho}^*(\alpha_{n,\rho}^*)\right) \left( \bar q_n(\theta_{n,\rho}^*) - \varphi_n(\theta_{n,\rho}^*)\right).
    \end{split}
\end{align}
Here, note that
\begin{align}\label{eq:Psidecomp}
    \begin{split}
    \Psi_{n,\rho}(\alpha_{n,\rho}^*)-\Psi_{n,\rho}^*(\alpha_{n,\rho}^*)
    &= -(\alpha_{n,\rho}^*)^2 \Psi_{n,\rho}^*(\alpha_{n,\rho}^*) \left(\Xi_{n,m,\rho}-\bbE[\Xi_{n,m,\rho}]\right) \Psi_{n,\rho}(\alpha_{n,\rho}^*) \\
    &= -(\alpha_{n,\rho}^*)^2 \Psi_{n,\rho}^*(\alpha_{n,\rho}^*) \left(\Xi_{n,m,\rho}-\bbE[\Xi_{n,m,\rho}]\right) \Psi_{n,\rho}^*(\alpha_{n,\rho}^*) + O_P(n^{-1}) \\
    &= O_P(n^{-1/2}).
    \end{split}
\end{align}
by Lemma \ref{lem:SAR_CLT}.
Moreover, by Lemma \ref{lem:SAR_CLT}, we can see that the last four terms in \eqref{eq:decomp1} are $O_P(n^{-1})$.
Therefore, since
\begin{align}\label{eq:qstar}
    \begin{split}
    \bar q_n(\theta_{n,\rho}^*)-\varphi_n(\theta_{n,\rho}^*)
    &= \left(\bar q_n(\theta_0)-\varphi_n(\theta_0)\right) + \left(\Pi_n-\bbE[\Pi_n]\right) (\theta_0-\theta_{n,\rho}^*) \\
    &= \left(\bar q_n(\theta_0)-\varphi_n(\theta_0)\right) + (\alpha_0-\alpha_{n,\rho}^*) \left(\pi_{n,1}-\bbE[\pi_{n,1}]\right),
    \end{split}
\end{align}
we have
\begin{align}
    &\Pi_n^\top \Psi_{n,\rho}(\alpha_{n,\rho}^*) \bar q_n(\theta_{n,\rho}^*)
    -\bbE[\Pi_n]^\top \Psi_{n,\rho}^*(\alpha_{n,\rho}^*) \varphi_n(\theta_{n,\rho}^*) \\
    &= \bbE[\Pi_n]^\top \Psi_{n,\rho}^*(\alpha_{n,\rho}^*) \left( \bar q_n(\theta_{n,\rho}^*)-\varphi_n(\theta_{n,\rho}^*)\right) \\
    &\quad + \left(\Pi_n - \bbE[\Pi_n]\right)^\top  \Psi_{n,\rho}^*(\alpha_{n,\rho}^*)\varphi_n(\theta_{n,\rho}^*) \\
    &\quad + \bbE[\Pi_n]^\top \left(\Psi_{n,\rho}(\alpha_{n,\rho}^*) - \Psi_{n,\rho}^*(\alpha_{n,\rho}^*)\right) \varphi_n(\theta_{n,\rho}^*) + O_P(n^{-1}) \\
    &= \bbE[\Pi_n]^\top \Psi_{n,\rho}^*(\alpha_{n,\rho}^*) \left(\bar q_n(\theta_0) - \varphi_n(\theta_0)\right) \\
    &\quad + (\alpha_0-\alpha_{n,\rho}^*) \bbE[\Pi_n]^\top \Psi_{n,\rho}^*(\alpha_{n,\rho}^*) \left(\pi_{n,1}-\bbE[\pi_{n,1}]\right) + e_1 \varphi_n(\theta_{n,\rho}^*)^\top \Psi_{n,\rho}^*(\alpha_{n,\rho}^*) \left(\pi_{n,1}-\bbE[\pi_{n,1}]\right)  \\
    &\quad -(\alpha_{n,\rho}^*)^2\bbE[\Pi_n]^\top \Psi_{n,\rho}^*(\alpha_{n,\rho}^*) \left(\Xi_{n,m,\rho}-\bbE[\Xi_{n,m,\rho}]\right) \Psi_{n,\rho}^*(\alpha_{n,\rho}^*)\varphi_n(\theta_{n,\rho}^*) + O_P(n^{-1}).
\end{align}

\bigskip

Next, we expand the quadratic term.
Write
\begin{align}
    \bar q_n(\theta_{n,\rho}^*)^\top \dot\Psi_{n,\rho}(\alpha_{n,\rho}^*) \bar q_n(\theta_{n,\rho}^*) 
    &=
    \left[
        \varphi_n(\theta_{n,\rho}^*)
        +\left(\bar q_n(\theta_{n,\rho}^*)-\varphi_n(\theta_{n,\rho}^*)\right)
    \right]^\top
    \left[
        \dot\Psi_{n,\rho}^*(\alpha_{n,\rho}^*)
        +\left(\dot\Psi_{n,\rho}(\alpha_{n,\rho}^*)-\dot\Psi_{n,\rho}^*(\alpha_{n,\rho}^*)\right)
    \right] \\
    &\quad
    \times
    \left[
        \varphi_n(\theta_{n,\rho}^*)
        +\left(\bar q_n(\theta_{n,\rho}^*)-\varphi_n(\theta_{n,\rho}^*)\right)
    \right].
\end{align}
Subtracting $\varphi_n(\theta_{n,\rho}^*)^\top \dot\Psi_{n,\rho}^*(\alpha_{n,\rho}^*) \varphi_n(\theta_{n,\rho}^*)$ from this, we obtain
\begin{align}
    &\bar q_n(\theta_{n,\rho}^*)^\top \dot\Psi_{n,\rho}(\alpha_{n,\rho}^*) \bar q_n(\theta_{n,\rho}^*)
    -\varphi_n(\theta_{n,\rho}^*)^\top \dot\Psi_{n,\rho}^*(\alpha_{n,\rho}^*) \varphi_n(\theta_{n,\rho}^*) \\
    &=
    2\varphi_n(\theta_{n,\rho}^*)^\top
    \dot\Psi_{n,\rho}^*(\alpha_{n,\rho}^*)
    \left(\bar q_n(\theta_{n,\rho}^*)-\varphi_n(\theta_{n,\rho}^*)\right) \\
    &\quad
    +\varphi_n(\theta_{n,\rho}^*)^\top
    \left(\dot\Psi_{n,\rho}(\alpha_{n,\rho}^*)-\dot\Psi_{n,\rho}^*(\alpha_{n,\rho}^*)\right)
    \varphi_n(\theta_{n,\rho}^*) \\
    &\quad
    +\left(\bar q_n(\theta_{n,\rho}^*)-\varphi_n(\theta_{n,\rho}^*)\right)^\top
    \dot\Psi_{n,\rho}^*(\alpha_{n,\rho}^*)
    \left(\bar q_n(\theta_{n,\rho}^*)-\varphi_n(\theta_{n,\rho}^*)\right) \\
    &\quad
    +2\varphi_n(\theta_{n,\rho}^*)^\top
    \left(\dot\Psi_{n,\rho}(\alpha_{n,\rho}^*)-\dot\Psi_{n,\rho}^*(\alpha_{n,\rho}^*)\right)
    \left(\bar q_n(\theta_{n,\rho}^*)-\varphi_n(\theta_{n,\rho}^*)\right) \\
    &\quad
    +\left(\bar q_n(\theta_{n,\rho}^*)-\varphi_n(\theta_{n,\rho}^*)\right)^\top
    \left(\dot\Psi_{n,\rho}(\alpha_{n,\rho}^*)-\dot\Psi_{n,\rho}^*(\alpha_{n,\rho}^*)\right)
    \left(\bar q_n(\theta_{n,\rho}^*)-\varphi_n(\theta_{n,\rho}^*)\right).
\end{align}
Note that \eqref{eq:qstar} with Lemma \ref{lem:SAR_CLT} implies that $\bar q_n(\theta_{n,\rho}^*) - \varphi_n(\theta_{n,\rho}^*) = O_P(n^{-1/2})$.
Also, it is straightforward to see that $\dot\Psi_{n,\rho}(\alpha_{n,\rho}^*)-\dot\Psi_{n,\rho}^*(\alpha_{n,\rho}^*)=O_P(n^{-1/2})$.
Therefore,
\begin{align}
    &\bar q_n(\theta_{n,\rho}^*)^\top \dot\Psi_{n,\rho}(\alpha_{n,\rho}^*) \bar q_n(\theta_{n,\rho}^*)
    -\varphi_n(\theta_{n,\rho}^*)^\top \dot\Psi_{n,\rho}^*(\alpha_{n,\rho}^*) \varphi_n(\theta_{n,\rho}^*) \\
    &= 2\varphi_n(\theta_{n,\rho}^*)^\top \dot\Psi_{n,\rho}^*(\alpha_{n,\rho}^*) \left(\bar q_n(\theta_{n,\rho}^*)-\varphi_n(\theta_{n,\rho}^*)\right) \\
    &\quad 
    +\varphi_n(\theta_{n,\rho}^*)^\top
    \left(\dot\Psi_{n,\rho}(\alpha_{n,\rho}^*)-\dot\Psi_{n,\rho}^*(\alpha_{n,\rho}^*)\right)
    \varphi_n(\theta_{n,\rho}^*)+O_P(n^{-1}).
\end{align}
Moreover, since $\dot\Psi_{n,\rho}(\alpha)=-2\alpha \Psi_{n,\rho}(\alpha) \Xi_{n,m,\rho} \Psi_{n,\rho}(\alpha)$, we have
\begin{align}
    \dot\Psi_{n,\rho}(\alpha_{n,\rho}^*)-\dot\Psi_{n,\rho}^*(\alpha_{n,\rho}^*)
    &=
    -2\alpha_{n,\rho}^*
    \left(
        \Psi_{n,\rho}(\alpha_{n,\rho}^*) \Xi_{n,m,\rho} \Psi_{n,\rho}(\alpha_{n,\rho}^*)
        -
        \Psi_{n,\rho}^*(\alpha_{n,\rho}^*) \bbE[\Xi_{n,m,\rho}] \Psi_{n,\rho}^*(\alpha_{n,\rho}^*)
    \right).
\end{align}
A similar decomposition as above gives
\begin{align}
    &\Psi_{n,\rho}(\alpha_{n,\rho}^*) \Xi_{n,m,\rho} \Psi_{n,\rho}(\alpha_{n,\rho}^*)
    -
    \Psi_{n,\rho}^*(\alpha_{n,\rho}^*) \bbE[\Xi_{n,m,\rho}] \Psi_{n,\rho}^*(\alpha_{n,\rho}^*) \\
    &=
    \Psi_{n,\rho}^*(\alpha_{n,\rho}^*)
    \left(\Xi_{n,m,\rho}-\bbE[\Xi_{n,m,\rho}]\right)
    \Psi_{n,\rho}^*(\alpha_{n,\rho}^*) \\
    &\quad
    +
    \left(\Psi_{n,\rho}(\alpha_{n,\rho}^*)-\Psi_{n,\rho}^*(\alpha_{n,\rho}^*)\right)
    \bbE[\Xi_{n,m,\rho}]
    \Psi_{n,\rho}^*(\alpha_{n,\rho}^*) \\
    &\quad
    +
    \Psi_{n,\rho}^*(\alpha_{n,\rho}^*)
    \bbE[\Xi_{n,m,\rho}]
    \left(\Psi_{n,\rho}(\alpha_{n,\rho}^*)-\Psi_{n,\rho}^*(\alpha_{n,\rho}^*)\right) \\
    &\quad
    +
    \left(\Psi_{n,\rho}(\alpha_{n,\rho}^*)-\Psi_{n,\rho}^*(\alpha_{n,\rho}^*)\right)
    \left(\Xi_{n,m,\rho}-\bbE[\Xi_{n,m,\rho}]\right)
    \Psi_{n,\rho}^*(\alpha_{n,\rho}^*) \\
    &\quad
    +
    \Psi_{n,\rho}^*(\alpha_{n,\rho}^*)
    \left(\Xi_{n,m,\rho}-\bbE[\Xi_{n,m,\rho}]\right)
    \left(\Psi_{n,\rho}(\alpha_{n,\rho}^*)-\Psi_{n,\rho}^*(\alpha_{n,\rho}^*)\right) \\
    &\quad
    +
    \left(\Psi_{n,\rho}(\alpha_{n,\rho}^*)-\Psi_{n,\rho}^*(\alpha_{n,\rho}^*)\right)
    \bbE[\Xi_{n,m,\rho}]
    \left(\Psi_{n,\rho}(\alpha_{n,\rho}^*)-\Psi_{n,\rho}^*(\alpha_{n,\rho}^*)\right) \\
    &\quad
    +
    \left(\Psi_{n,\rho}(\alpha_{n,\rho}^*)-\Psi_{n,\rho}^*(\alpha_{n,\rho}^*)\right)
    \left(\Xi_{n,m,\rho}-\bbE[\Xi_{n,m,\rho}]\right)
    \left(\Psi_{n,\rho}(\alpha_{n,\rho}^*)-\Psi_{n,\rho}^*(\alpha_{n,\rho}^*)\right) \\
    &=
    \Psi_{n,\rho}^*(\alpha_{n,\rho}^*)
    \left(\Xi_{n,m,\rho}-\bbE[\Xi_{n,m,\rho}]\right)
    \Psi_{n,\rho}^*(\alpha_{n,\rho}^*) \\
    &\quad
    +
    \left(\Psi_{n,\rho}(\alpha_{n,\rho}^*)-\Psi_{n,\rho}^*(\alpha_{n,\rho}^*)\right)
    \bbE[\Xi_{n,m,\rho}]
    \Psi_{n,\rho}^*(\alpha_{n,\rho}^*) \\
    &\quad
    +
    \Psi_{n,\rho}^*(\alpha_{n,\rho}^*)
    \bbE[\Xi_{n,m,\rho}]
    \left(\Psi_{n,\rho}(\alpha_{n,\rho}^*)-\Psi_{n,\rho}^*(\alpha_{n,\rho}^*)\right)
    +O_P(n^{-1}).
\end{align}
Moreover, by \eqref{eq:Psidecomp},
\begin{align}
    &\Psi_{n,\rho}(\alpha_{n,\rho}^*) \Xi_{n,m,\rho} \Psi_{n,\rho}(\alpha_{n,\rho}^*)
    -
    \Psi_{n,\rho}^*(\alpha_{n,\rho}^*) \bbE[\Xi_{n,m,\rho}] \Psi_{n,\rho}^*(\alpha_{n,\rho}^*) \\
    &=
    \Psi_{n,\rho}^*(\alpha_{n,\rho}^*)
    \left(\Xi_{n,m,\rho}-\bbE[\Xi_{n,m,\rho}]\right)
    \Psi_{n,\rho}^*(\alpha_{n,\rho}^*) \\
    &\quad
    -(\alpha_{n,\rho}^*)^2
    \Psi_{n,\rho}^*(\alpha_{n,\rho}^*)
    \left(\Xi_{n,m,\rho}-\bbE[\Xi_{n,m,\rho}]\right)
    \Psi_{n,\rho}^*(\alpha_{n,\rho}^*)
    \bbE[\Xi_{n,m,\rho}]
    \Psi_{n,\rho}^*(\alpha_{n,\rho}^*) \\
    &\quad
    -(\alpha_{n,\rho}^*)^2
    \Psi_{n,\rho}^*(\alpha_{n,\rho}^*)
    \bbE[\Xi_{n,m,\rho}]
    \Psi_{n,\rho}^*(\alpha_{n,\rho}^*)
    \left(\Xi_{n,m,\rho}-\bbE[\Xi_{n,m,\rho}]\right)
    \Psi_{n,\rho}^*(\alpha_{n,\rho}^*)
    +O_P(n^{-1}).
\end{align}
Therefore,
\begin{align}
    \dot\Psi_{n,\rho}(\alpha_{n,\rho}^*)-\dot\Psi_{n,\rho}^*(\alpha_{n,\rho}^*)
    &=
    -2\alpha_{n,\rho}^*
    \Psi_{n,\rho}^*(\alpha_{n,\rho}^*)
    \left(\Xi_{n,m,\rho}-\bbE[\Xi_{n,m,\rho}]\right)
    \Psi_{n,\rho}^*(\alpha_{n,\rho}^*) \\
    &\quad
    +2(\alpha_{n,\rho}^*)^3
    \Psi_{n,\rho}^*(\alpha_{n,\rho}^*)
    \left(\Xi_{n,m,\rho}-\bbE[\Xi_{n,m,\rho}]\right)
    \Psi_{n,\rho}^*(\alpha_{n,\rho}^*)
    \bbE[\Xi_{n,m,\rho}]
    \Psi_{n,\rho}^*(\alpha_{n,\rho}^*) \\
    &\quad
    +2(\alpha_{n,\rho}^*)^3
    \Psi_{n,\rho}^*(\alpha_{n,\rho}^*)
    \bbE[\Xi_{n,m,\rho}]
    \Psi_{n,\rho}^*(\alpha_{n,\rho}^*)
    \left(\Xi_{n,m,\rho}-\bbE[\Xi_{n,m,\rho}]\right)
    \Psi_{n,\rho}^*(\alpha_{n,\rho}^*)
    +O_P(n^{-1}).
\end{align}
Combining these results, in view of \eqref{eq:qstar}, we obtain
\begin{align}
    &\bar q_n(\theta_{n,\rho}^*)^\top \dot\Psi_{n,\rho}(\alpha_{n,\rho}^*) \bar q_n(\theta_{n,\rho}^*) - \varphi_n(\theta_{n,\rho}^*)^\top \dot\Psi_{n,\rho}^*(\alpha_{n,\rho}^*) \varphi_n(\theta_{n,\rho}^*) \\
    &=
    2\varphi_n(\theta_{n,\rho}^*)^\top \dot\Psi_{n,\rho}^*(\alpha_{n,\rho}^*) \left(\bar q_n(\theta_0)-\varphi_n(\theta_0)\right) \\
    &\quad
    +2(\alpha_0-\alpha_{n,\rho}^*)\varphi_n(\theta_{n,\rho}^*)^\top \dot\Psi_{n,\rho}^*(\alpha_{n,\rho}^*) \left(\pi_{n,1}-\bbE[\pi_{n,1}]\right) \\
    &\quad
    - 2 \alpha_{n,\rho}^* \varphi_n(\theta_{n,\rho}^*)^\top \Psi_{n,\rho}^*(\alpha_{n,\rho}^*) \left(\Xi_{n,m,\rho}-\bbE[\Xi_{n,m,\rho}]\right) \Psi_{n,\rho}^*(\alpha_{n,\rho}^*) \varphi_n(\theta_{n,\rho}^*) \\
    &\quad
    + 2 (\alpha_{n,\rho}^*)^3 \varphi_n(\theta_{n,\rho}^*)^\top \Psi_{n,\rho}^*(\alpha_{n,\rho}^*) \left(\Xi_{n,m,\rho}-\bbE[\Xi_{n,m,\rho}]\right) \Psi_{n,\rho}^*(\alpha_{n,\rho}^*) \bbE[\Xi_{n,m,\rho}] \Psi_{n,\rho}^*(\alpha_{n,\rho}^*) \varphi_n(\theta_{n,\rho}^*) \\
    &\quad
    + 2(\alpha_{n,\rho}^*)^3 \varphi_n(\theta_{n,\rho}^*)^\top \Psi_{n,\rho}^*(\alpha_{n,\rho}^*) \bbE[\Xi_{n,m,\rho}] \Psi_{n,\rho}^*(\alpha_{n,\rho}^*) \left(\Xi_{n,m,\rho}-\bbE[\Xi_{n,m,\rho}]\right) \Psi_{n,\rho}^*(\alpha_{n,\rho}^*) \varphi_n(\theta_{n,\rho}^*)+O_P(n^{-1}).
\end{align}
Consequently, we have
\small\begin{align}\label{eq:scorefull}
    \begin{split}
    &S_{n,\rho}(\theta_{n,\rho}^*)-S_{n,\rho}^*(\theta_{n,\rho}^*) \\
    &= -2 \bbE[\Pi_n]^\top \Psi_{n,\rho}^*(\alpha_{n,\rho}^*) \left(\bar q_n(\theta_0)-\varphi_n(\theta_0)\right) \\
    &\quad + 2 e_1 \varphi_n(\theta_{n,\rho}^*)^\top \dot\Psi_{n,\rho}^*(\alpha_{n,\rho}^*)  \left(\bar q_n(\theta_0)-\varphi_n(\theta_0)\right)\\
    & \\
    &\quad -2 (\alpha_0-\alpha_{n,\rho}^*) \bbE[\Pi_n]^\top \Psi_{n,\rho}^*(\alpha_{n,\rho}^*) \left(\pi_{n,1}-\bbE[\pi_{n,1}]\right) \\
    &\quad -2 e_1 \varphi_n(\theta_{n,\rho}^*)^\top \Psi_{n,\rho}^*(\alpha_{n,\rho}^*) \left(\pi_{n,1}-\bbE[\pi_{n,1}]\right) \\
    &\quad + 2 e_1 (\alpha_0-\alpha_{n,\rho}^*) \varphi_n(\theta_{n,\rho}^*)^\top \dot\Psi_{n,\rho}^*(\alpha_{n,\rho}^*) \left(\pi_{n,1}-\bbE[\pi_{n,1}]\right) \\
    & \\
    &\quad + 2 (\alpha_{n,\rho}^*)^2 \bbE[\Pi_n]^\top \Psi_{n,\rho}^*(\alpha_{n,\rho}^*) \left(\Xi_{n,m,\rho}-\bbE[\Xi_{n,m,\rho}]\right) \Psi_{n,\rho}^*(\alpha_{n,\rho}^*) \varphi_n(\theta_{n,\rho}^*) \\
    &\quad - 2 e_1 \alpha_{n,\rho}^* \varphi_n(\theta_{n,\rho}^*)^\top \Psi_{n,\rho}^*(\alpha_{n,\rho}^*) \left(\Xi_{n,m,\rho}-\bbE[\Xi_{n,m,\rho}]\right) \Psi_{n,\rho}^*(\alpha_{n,\rho}^*) \varphi_n(\theta_{n,\rho}^*) \\
    &\quad + 2 e_1 (\alpha_{n,\rho}^*)^3 \varphi_n(\theta_{n,\rho}^*)^\top \Psi_{n,\rho}^*(\alpha_{n,\rho}^*) \left(\Xi_{n,m,\rho}-\bbE[\Xi_{n,m,\rho}]\right) \Psi_{n,\rho}^*(\alpha_{n,\rho}^*) \bbE[\Xi_{n,m,\rho}] \Psi_{n,\rho}^*(\alpha_{n,\rho}^*) \varphi_n(\theta_{n,\rho}^*) \\
    &\quad + 2 e_1 (\alpha_{n,\rho}^*)^3 \varphi_n(\theta_{n,\rho}^*)^\top \Psi_{n,\rho}^*(\alpha_{n,\rho}^*) \bbE[\Xi_{n,m,\rho}] \Psi_{n,\rho}^*(\alpha_{n,\rho}^*) \left(\Xi_{n,m,\rho}-\bbE[\Xi_{n,m,\rho}]\right) \Psi_{n,\rho}^*(\alpha_{n,\rho}^*) \varphi_n(\theta_{n,\rho}^*) + O_P(n^{-1}).
    \end{split}
\end{align}\normalsize
This implies that $S_{n,\rho}(\theta_{n,\rho}^*)-S_{n,\rho}^*(\theta_{n,\rho}^*)$ is a linear combination of $\bar q_n(\theta_0)-\varphi_n(\theta_0)$, $\pi_{n,1}-\bbE[\pi_{n,1}]$, and $\operatorname{vec}(\Xi_{n,m,\rho}-\bbE[\Xi_{n,m,\rho}])$, up to an $O_P(n^{-1})$ remainder.
Hence, there exists $\mathcal J_{n,\rho}$ such that
\begin{align}
    S_{n,\rho}(\theta_{n,\rho}^*)-S_{n,\rho}^*(\theta_{n,\rho}^*) = \mathcal J_{n,\rho}
    \begin{pmatrix}
        \bar q_n(\theta_0)-\varphi_n(\theta_0) \\
        \pi_{n,1}-\bbE[\pi_{n,1}] \\
        \operatorname{vec}(\Xi_{n,m,\rho}-\bbE[\Xi_{n,m,\rho}])
    \end{pmatrix}
    + O_P(n^{-1}),
\end{align}
where
\begin{align}
    \mathcal J_{n,\rho} = \left(\mathcal J_{q,n,\rho}, \mathcal J_{\pi,n,\rho}, \mathcal J_{\Xi,n,\rho}\right) \in \bbR^{d_\theta\times(d_\mu + d_\mu + d_\mu^2)}.
\end{align}
The exact form of the matrix $\mathcal J_{n,\rho}$ can be obtained by collecting the corresponding components in \eqref{eq:scorefull}.
For example, $\mathcal J_{q,n,\rho} = 2(e_1 \varphi_n(\theta_{n,\rho}^*)^\top \dot\Psi_{n,\rho}^*(\alpha_{n,\rho}^*)  - \bbE[\Pi_n]^\top \Psi_{n,\rho}^*(\alpha_{n,\rho}^*))$.

Finally, we obtain
\begin{align}
    \sqrt n(\hat\theta_{n,\rho}-\theta_{n,\rho}^*)
    & = - \left(H_{n, \rho}^*(\theta^*_{n, \rho}) + o_P(1)\right)^{-1}
    \mathcal J_{n,\rho} \sqrt n 
    \begin{pmatrix}
        \bar q_n(\theta_0)-\varphi_n(\theta_0) \\
        \pi_{n,1}-\bbE[\pi_{n,1}] \\
        \operatorname{vec}(\Xi_{n,m,\rho}-\bbE[\Xi_{n,m,\rho}])
    \end{pmatrix}
    + o_P(1) \\
    & \overset{d}{\to} N(\bm 0, (H_\rho^*)^{-1} \mathcal J_\rho \Sigma_\rho^\text{SAR} \mathcal J_\rho^\top (H_\rho^*)^{-1})
\end{align}
by Lemma \ref{lem:SAR_CLT}.

\qed

\section{Supplementary Tables}\label{app:table}

\begin{table}[ht]
    \footnotesize
  \centering
  \caption{Simulation results for estimating $\alpha_0$ ($\Omega_n = I_{d_\mu}$)}
  \label{tab:alpha_I}
\begin{tabular}{lllcccccccccccc}
\toprule
 &  &  &
 \multicolumn{2}{c}{Naive ($n=500$)} &
 \multicolumn{2}{c}{Naive ($n=1000$)} &
 \multicolumn{2}{c}{FW ($n=500$)} &
 \multicolumn{2}{c}{FW ($n=1000$)} &
 \multicolumn{2}{c}{NA ($n=500$)} &
 \multicolumn{2}{c}{NA ($n=1000$)} \\
$p_{\mathrm{drop}}$ & $p_{\mathrm{add}}$ & $\rho$
& Bias & RMSE & Bias & RMSE
& Bias & RMSE & Bias & RMSE
& Bias & RMSE & Bias & RMSE \\
\midrule
  0 & 0 & 0.125 & 0.000 & 0.024 & -0.001 & 0.018 & 0.000 & 0.023 & -0.001 & 0.016 & 0.002 & 0.023 & 0.000 & 0.016 \\
   &  & 256 & 0.000 & 0.024 & -0.001 & 0.018 & 0.000 & 0.023 & -0.001 & 0.016 & 0.002 & 0.023 & 0.000 & 0.017 \\
   &  & 4096 & 0.000 & 0.024 & -0.001 & 0.018 & 0.000 & 0.023 & -0.001 & 0.016 & 0.001 & 0.023 & 0.000 & 0.017 \\
   & 0.05 & 0.125 & -0.151 & 0.199 & -0.173 & 0.198 & -0.133 & 0.163 & -0.158 & 0.174 & -0.071 & 0.113 & -0.118 & 0.138 \\
   &  & 256 & -0.151 & 0.199 & -0.173 & 0.198 & -0.137 & 0.168 & -0.161 & 0.177 & -0.089 & 0.130 & -0.132 & 0.152 \\
   &  & 4096 & -0.151 & 0.199 & -0.173 & 0.198 & -0.147 & 0.186 & -0.169 & 0.189 & -0.141 & 0.186 & -0.167 & 0.190 \\
   & 0.1 & 0.125 & -0.142 & 0.223 & -0.142 & 0.180 & -0.129 & 0.183 & -0.126 & 0.156 & -0.011 & 0.141 & -0.064 & 0.110 \\
   &  & 256 & -0.142 & 0.223 & -0.142 & 0.180 & -0.132 & 0.190 & -0.132 & 0.162 & -0.062 & 0.158 & -0.100 & 0.138 \\
   &  & 4096 & -0.142 & 0.223 & -0.142 & 0.180 & -0.139 & 0.210 & -0.140 & 0.174 & -0.136 & 0.213 & -0.138 & 0.175 \\
  0.2 & 0 & 0.125 & -0.106 & 0.109 & -0.121 & 0.123 & -0.108 & 0.111 & -0.124 & 0.125 & -0.105 & 0.107 & -0.122 & 0.123 \\
   &  & 256 & -0.106 & 0.109 & -0.121 & 0.123 & -0.108 & 0.111 & -0.124 & 0.125 & -0.105 & 0.108 & -0.122 & 0.123 \\
   &  & 4096 & -0.106 & 0.109 & -0.121 & 0.123 & -0.107 & 0.110 & -0.123 & 0.124 & -0.106 & 0.108 & -0.122 & 0.123 \\
   & 0.05 & 0.125 & -0.233 & 0.265 & -0.265 & 0.280 & -0.198 & 0.219 & -0.261 & 0.269 & -0.119 & 0.151 & -0.215 & 0.226 \\
   &  & 256 & -0.233 & 0.265 & -0.265 & 0.280 & -0.205 & 0.226 & -0.261 & 0.270 & -0.147 & 0.177 & -0.231 & 0.242 \\
   &  & 4096 & -0.233 & 0.265 & -0.265 & 0.280 & -0.224 & 0.250 & -0.264 & 0.275 & -0.222 & 0.253 & -0.262 & 0.275 \\
   & 0.1 & 0.125 & -0.226 & 0.282 & -0.237 & 0.260 & -0.201 & 0.239 & -0.216 & 0.236 & -0.053 & 0.158 & -0.135 & 0.165 \\
   &  & 256 & -0.226 & 0.282 & -0.237 & 0.260 & -0.209 & 0.249 & -0.224 & 0.243 & -0.130 & 0.203 & -0.190 & 0.215 \\
   &  & 4096 & -0.226 & 0.282 & -0.237 & 0.260 & -0.222 & 0.270 & -0.234 & 0.255 & -0.221 & 0.275 & -0.234 & 0.257 \\
  0.3 & 0 & 0.125 & -0.159 & 0.161 & -0.178 & 0.179 & -0.159 & 0.161 & -0.180 & 0.181 & -0.155 & 0.157 & -0.178 & 0.179 \\
   &  & 256 & -0.159 & 0.161 & -0.178 & 0.179 & -0.159 & 0.161 & -0.180 & 0.181 & -0.155 & 0.157 & -0.178 & 0.179 \\
   &  & 4096 & -0.159 & 0.161 & -0.178 & 0.179 & -0.159 & 0.161 & -0.180 & 0.181 & -0.157 & 0.159 & -0.178 & 0.179 \\
   & 0.05 & 0.125 & -0.281 & 0.305 & -0.310 & 0.322 & -0.249 & 0.266 & -0.310 & 0.317 & -0.161 & 0.188 & -0.260 & 0.268 \\
   &  & 256 & -0.281 & 0.305 & -0.310 & 0.322 & -0.255 & 0.272 & -0.310 & 0.317 & -0.194 & 0.220 & -0.279 & 0.288 \\
   &  & 4096 & -0.281 & 0.305 & -0.310 & 0.322 & -0.272 & 0.292 & -0.310 & 0.319 & -0.271 & 0.296 & -0.307 & 0.319 \\
   & 0.1 & 0.125 & -0.275 & 0.321 & -0.284 & 0.303 & -0.244 & 0.277 & -0.262 & 0.279 & -0.069 & 0.177 & -0.164 & 0.194 \\
   &  & 256 & -0.275 & 0.321 & -0.284 & 0.303 & -0.254 & 0.288 & -0.270 & 0.286 & -0.166 & 0.235 & -0.235 & 0.257 \\
   &  & 4096 & -0.275 & 0.321 & -0.284 & 0.303 & -0.270 & 0.310 & -0.281 & 0.298 & -0.270 & 0.316 & -0.281 & 0.300 \\
\bottomrule
\end{tabular}
\begin{minipage}{0.98\textwidth}
\textit{Note.} Naive, FW, and NA denote naive GMM, fixed-weight NA-GMM, and original NA-GMM, respectively.
\end{minipage}
\end{table}

\begin{table}[ht]
    \footnotesize
  \centering
  \caption{Simulation results for estimating $\alpha_0$ ($\Omega_n =$ 2SLS weight)}
  \label{tab:alpha_S}
\begin{tabular}{lllcccccccccccc}
\toprule
 &  &  &
 \multicolumn{2}{c}{2SLS ($n=500$)} &
 \multicolumn{2}{c}{2SLS ($n=1000$)} &
 \multicolumn{2}{c}{FW ($n=500$)} &
 \multicolumn{2}{c}{FW ($n=1000$)} &
 \multicolumn{2}{c}{NA ($n=500$)} &
 \multicolumn{2}{c}{NA ($n=1000$)} \\
$p_{\mathrm{drop}}$ & $p_{\mathrm{add}}$ & $\rho$
& Bias & RMSE & Bias & RMSE
& Bias & RMSE & Bias & RMSE
& Bias & RMSE & Bias & RMSE \\
\midrule
  0 & 0 & 0.125 & 0.000 & 0.023 & -0.001 & 0.016 & 0.000 & 0.023 & -0.001 & 0.016 & 0.002 & 0.023 & 0.000 & 0.016 \\
   &  & 256 & 0.000 & 0.023 & -0.001 & 0.016 & 0.000 & 0.023 & -0.001 & 0.016 & 0.002 & 0.023 & 0.000 & 0.016 \\
   &  & 4096 & 0.000 & 0.023 & -0.001 & 0.016 & 0.000 & 0.023 & -0.001 & 0.016 & 0.002 & 0.023 & 0.000 & 0.016 \\
   & 0.05 & 0.125 & -0.135 & 0.165 & -0.159 & 0.175 & -0.133 & 0.163 & -0.158 & 0.174 & -0.071 & 0.113 & -0.118 & 0.138 \\
   &  & 256 & -0.135 & 0.165 & -0.159 & 0.175 & -0.133 & 0.163 & -0.158 & 0.174 & -0.072 & 0.114 & -0.118 & 0.139 \\
   &  & 4096 & -0.135 & 0.165 & -0.159 & 0.175 & -0.133 & 0.163 & -0.158 & 0.174 & -0.084 & 0.123 & -0.126 & 0.146 \\
   & 0.1 & 0.125 & -0.129 & 0.184 & -0.127 & 0.156 & -0.129 & 0.183 & -0.126 & 0.156 & -0.011 & 0.141 & -0.064 & 0.110 \\
   &  & 256 & -0.129 & 0.184 & -0.127 & 0.156 & -0.129 & 0.183 & -0.126 & 0.156 & -0.013 & 0.141 & -0.065 & 0.111 \\
   &  & 4096 & -0.129 & 0.184 & -0.127 & 0.156 & -0.129 & 0.183 & -0.126 & 0.156 & -0.034 & 0.144 & -0.076 & 0.120 \\
  0.2 & 0 & 0.125 & -0.108 & 0.111 & -0.124 & 0.125 & -0.108 & 0.111 & -0.124 & 0.125 & -0.105 & 0.107 & -0.122 & 0.123 \\
   &  & 256 & -0.108 & 0.111 & -0.124 & 0.125 & -0.108 & 0.111 & -0.124 & 0.125 & -0.105 & 0.107 & -0.122 & 0.123 \\
   &  & 4096 & -0.108 & 0.111 & -0.124 & 0.125 & -0.108 & 0.111 & -0.124 & 0.125 & -0.105 & 0.108 & -0.122 & 0.123 \\
   & 0.05 & 0.125 & -0.201 & 0.222 & -0.261 & 0.270 & -0.198 & 0.219 & -0.261 & 0.269 & -0.119 & 0.151 & -0.215 & 0.226 \\
   &  & 256 & -0.201 & 0.222 & -0.261 & 0.270 & -0.198 & 0.219 & -0.261 & 0.269 & -0.120 & 0.152 & -0.216 & 0.227 \\
   &  & 4096 & -0.201 & 0.222 & -0.261 & 0.270 & -0.199 & 0.220 & -0.261 & 0.269 & -0.139 & 0.168 & -0.229 & 0.239 \\
   & 0.1 & 0.125 & -0.202 & 0.241 & -0.217 & 0.236 & -0.201 & 0.239 & -0.216 & 0.236 & -0.053 & 0.158 & -0.135 & 0.164 \\
   &  & 256 & -0.202 & 0.241 & -0.217 & 0.236 & -0.201 & 0.239 & -0.216 & 0.236 & -0.056 & 0.159 & -0.137 & 0.166 \\
   &  & 4096 & -0.202 & 0.241 & -0.217 & 0.236 & -0.201 & 0.239 & -0.216 & 0.236 & -0.092 & 0.171 & -0.157 & 0.184 \\
  0.3 & 0 & 0.125 & -0.159 & 0.161 & -0.180 & 0.181 & -0.159 & 0.161 & -0.180 & 0.181 & -0.155 & 0.157 & -0.178 & 0.179 \\
   &  & 256 & -0.159 & 0.161 & -0.180 & 0.181 & -0.159 & 0.161 & -0.180 & 0.181 & -0.155 & 0.157 & -0.178 & 0.179 \\
   &  & 4096 & -0.159 & 0.161 & -0.180 & 0.181 & -0.159 & 0.161 & -0.180 & 0.181 & -0.155 & 0.158 & -0.178 & 0.179 \\
   & 0.05 & 0.125 & -0.252 & 0.269 & -0.310 & 0.317 & -0.249 & 0.266 & -0.310 & 0.317 & -0.161 & 0.188 & -0.260 & 0.268 \\
   &  & 256 & -0.252 & 0.269 & -0.310 & 0.317 & -0.249 & 0.266 & -0.310 & 0.317 & -0.163 & 0.190 & -0.262 & 0.270 \\
   &  & 4096 & -0.252 & 0.269 & -0.310 & 0.317 & -0.250 & 0.266 & -0.310 & 0.317 & -0.190 & 0.213 & -0.279 & 0.287 \\
   & 0.1 & 0.125 & -0.246 & 0.279 & -0.263 & 0.279 & -0.244 & 0.277 & -0.262 & 0.279 & -0.068 & 0.177 & -0.164 & 0.194 \\
   &  & 256 & -0.246 & 0.279 & -0.263 & 0.279 & -0.244 & 0.277 & -0.262 & 0.279 & -0.072 & 0.178 & -0.167 & 0.197 \\
   &  & 4096 & -0.246 & 0.279 & -0.263 & 0.279 & -0.244 & 0.278 & -0.262 & 0.279 & -0.120 & 0.197 & -0.196 & 0.221 \\
\bottomrule
\end{tabular}
\begin{minipage}{0.98\textwidth}
\textit{Note.} FW and NA denote fixed-weight NA-GMM and original NA-GMM, respectively.
\end{minipage}
\end{table}

\begin{table}[ht]
    \footnotesize
  \centering
  \caption{Simulation results for estimating $\beta_{10}$ ($\Omega_n = I_{d_\mu}$)}
  \label{tab:beta_I}
\begin{tabular}{lllcccccccccccc}
\toprule
 &  &  &
 \multicolumn{2}{c}{Naive ($n=500$)} &
 \multicolumn{2}{c}{Naive ($n=1000$)} &
 \multicolumn{2}{c}{FW ($n=500$)} &
 \multicolumn{2}{c}{FW ($n=1000$)} &
 \multicolumn{2}{c}{NA ($n=500$)} &
 \multicolumn{2}{c}{NA ($n=1000$)} \\
$p_{\mathrm{drop}}$ & $p_{\mathrm{add}}$ & $\rho$
& Bias & RMSE & Bias & RMSE
& Bias & RMSE & Bias & RMSE
& Bias & RMSE & Bias & RMSE \\
\midrule
  0 & 0 & 0.125 & -0.001 & 0.047 & -0.002 & 0.033 & -0.001 & 0.047 & -0.002 & 0.033 & -0.001 & 0.047 & -0.002 & 0.033 \\
    &   & 256 & -0.001 & 0.047 & -0.002 & 0.033 & -0.001 & 0.047 & -0.002 & 0.033 & -0.001 & 0.047 & -0.002 & 0.033 \\
    &   & 4096 & -0.001 & 0.047 & -0.002 & 0.033 & -0.001 & 0.047 & -0.002 & 0.033 & -0.001 & 0.047 & -0.002 & 0.033 \\
    & 0.05 & 0.125 & -0.007 & 0.071 & -0.011 & 0.054 & -0.006 & 0.071 & -0.010 & 0.054 & -0.005 & 0.071 & -0.010 & 0.054 \\
    &   & 256 & -0.007 & 0.071 & -0.011 & 0.054 & -0.006 & 0.071 & -0.010 & 0.054 & -0.005 & 0.071 & -0.010 & 0.054 \\
    &   & 4096 & -0.007 & 0.071 & -0.011 & 0.054 & -0.006 & 0.071 & -0.011 & 0.054 & -0.006 & 0.071 & -0.011 & 0.054 \\
    & 0.1 & 0.125 & -0.006 & 0.073 & -0.011 & 0.055 & -0.006 & 0.073 & -0.011 & 0.055 & -0.004 & 0.073 & -0.010 & 0.055 \\
    &   & 256 & -0.006 & 0.073 & -0.011 & 0.055 & -0.006 & 0.073 & -0.011 & 0.055 & -0.005 & 0.073 & -0.011 & 0.055 \\
    &   & 4096 & -0.006 & 0.073 & -0.011 & 0.055 & -0.006 & 0.073 & -0.011 & 0.055 & -0.006 & 0.073 & -0.011 & 0.055 \\
  0.2 & 0 & 0.125 & -0.005 & 0.054 & -0.006 & 0.041 & -0.005 & 0.054 & -0.006 & 0.041 & -0.005 & 0.054 & -0.006 & 0.041 \\
    &   & 256 & -0.005 & 0.054 & -0.006 & 0.041 & -0.005 & 0.054 & -0.006 & 0.041 & -0.005 & 0.054 & -0.006 & 0.041 \\
    &   & 4096 & -0.005 & 0.054 & -0.006 & 0.041 & -0.005 & 0.054 & -0.006 & 0.041 & -0.005 & 0.054 & -0.006 & 0.041 \\
    & 0.05 & 0.125 & -0.009 & 0.072 & -0.013 & 0.056 & -0.008 & 0.072 & -0.013 & 0.056 & -0.006 & 0.072 & -0.012 & 0.056 \\
    &   & 256 & -0.009 & 0.072 & -0.013 & 0.056 & -0.008 & 0.072 & -0.013 & 0.056 & -0.007 & 0.072 & -0.013 & 0.056 \\
    &   & 4096 & -0.009 & 0.072 & -0.013 & 0.056 & -0.008 & 0.072 & -0.013 & 0.056 & -0.008 & 0.072 & -0.013 & 0.056 \\
    & 0.1 & 0.125 & -0.008 & 0.074 & -0.013 & 0.057 & -0.008 & 0.074 & -0.013 & 0.056 & -0.006 & 0.074 & -0.012 & 0.056 \\
    &   & 256 & -0.008 & 0.074 & -0.013 & 0.057 & -0.008 & 0.074 & -0.013 & 0.056 & -0.007 & 0.074 & -0.013 & 0.056 \\
    &   & 4096 & -0.008 & 0.074 & -0.013 & 0.057 & -0.008 & 0.074 & -0.013 & 0.057 & -0.008 & 0.074 & -0.013 & 0.057 \\
  0.3 & 0 & 0.125 & -0.007 & 0.060 & -0.009 & 0.043 & -0.007 & 0.060 & -0.009 & 0.043 & -0.007 & 0.060 & -0.009 & 0.043 \\
    &   & 256 & -0.007 & 0.060 & -0.009 & 0.043 & -0.007 & 0.060 & -0.009 & 0.043 & -0.007 & 0.060 & -0.009 & 0.043 \\
    &   & 4096 & -0.007 & 0.060 & -0.009 & 0.043 & -0.007 & 0.060 & -0.009 & 0.043 & -0.007 & 0.060 & -0.009 & 0.043 \\
    & 0.05 & 0.125 & -0.009 & 0.073 & -0.015 & 0.057 & -0.009 & 0.073 & -0.015 & 0.057 & -0.007 & 0.073 & -0.014 & 0.057 \\
    &   & 256 & -0.009 & 0.073 & -0.015 & 0.057 & -0.009 & 0.073 & -0.015 & 0.057 & -0.008 & 0.073 & -0.014 & 0.057 \\
    &   & 4096 & -0.009 & 0.073 & -0.015 & 0.057 & -0.009 & 0.073 & -0.015 & 0.057 & -0.009 & 0.073 & -0.015 & 0.057 \\
    & 0.1 & 0.125 & -0.009 & 0.075 & -0.014 & 0.057 & -0.009 & 0.075 & -0.014 & 0.057 & -0.006 & 0.075 & -0.013 & 0.057 \\
    &   & 256 & -0.009 & 0.075 & -0.014 & 0.057 & -0.009 & 0.075 & -0.014 & 0.057 & -0.008 & 0.075 & -0.013 & 0.057 \\
    &   & 4096 & -0.009 & 0.075 & -0.014 & 0.057 & -0.009 & 0.075 & -0.014 & 0.057 & -0.009 & 0.075 & -0.014 & 0.057 \\
\bottomrule
\end{tabular}
\begin{minipage}{0.98\textwidth}
\textit{Note.} Naive, FW, and NA denote naive GMM, fixed-weight NA-GMM, and original NA-GMM, respectively.
\end{minipage}
\end{table}

\begin{table}[ht]
    \footnotesize
  \centering
  \caption{Simulation results for estimating $\beta_{10}$ ($\Omega_n =$ 2SLS weight)}
  \label{tab:beta_S}
\begin{tabular}{lllcccccccccccc}
\toprule
 &  &  &
 \multicolumn{2}{c}{2SLS ($n=500$)} &
 \multicolumn{2}{c}{2SLS ($n=1000$)} &
 \multicolumn{2}{c}{FW ($n=500$)} &
 \multicolumn{2}{c}{FW ($n=1000$)} &
 \multicolumn{2}{c}{NA ($n=500$)} &
 \multicolumn{2}{c}{NA ($n=1000$)} \\
$p_{\mathrm{drop}}$ & $p_{\mathrm{add}}$ & $\rho$
& Bias & RMSE & Bias & RMSE
& Bias & RMSE & Bias & RMSE
& Bias & RMSE & Bias & RMSE \\
\midrule
  0 & 0 & 0.125 & -0.001 & 0.047 & -0.002 & 0.033 & -0.001 & 0.047 & -0.002 & 0.033 & -0.001 & 0.047 & -0.002 & 0.033 \\
   &  & 256 & -0.001 & 0.047 & -0.002 & 0.033 & -0.001 & 0.047 & -0.002 & 0.033 & -0.001 & 0.047 & -0.002 & 0.033 \\
   &  & 4096 & -0.001 & 0.047 & -0.002 & 0.033 & -0.001 & 0.047 & -0.002 & 0.033 & -0.001 & 0.047 & -0.002 & 0.033 \\
   & 0.05 & 0.125 & -0.006 & 0.070 & -0.010 & 0.054 & -0.006 & 0.071 & -0.010 & 0.054 & -0.005 & 0.071 & -0.010 & 0.054 \\
   &  & 256 & -0.006 & 0.070 & -0.010 & 0.054 & -0.006 & 0.071 & -0.010 & 0.054 & -0.005 & 0.071 & -0.010 & 0.054 \\
   &  & 4096 & -0.006 & 0.070 & -0.010 & 0.054 & -0.006 & 0.071 & -0.010 & 0.054 & -0.005 & 0.070 & -0.010 & 0.054 \\
   & 0.1 & 0.125 & -0.006 & 0.073 & -0.011 & 0.055 & -0.006 & 0.073 & -0.011 & 0.055 & -0.004 & 0.073 & -0.010 & 0.055 \\
   &  & 256 & -0.006 & 0.073 & -0.011 & 0.055 & -0.006 & 0.073 & -0.011 & 0.055 & -0.004 & 0.073 & -0.010 & 0.055 \\
   &  & 4096 & -0.006 & 0.073 & -0.011 & 0.055 & -0.006 & 0.073 & -0.011 & 0.055 & -0.005 & 0.073 & -0.010 & 0.055 \\
  0.2 & 0 & 0.125 & -0.005 & 0.054 & -0.006 & 0.041 & -0.005 & 0.054 & -0.006 & 0.041 & -0.005 & 0.054 & -0.006 & 0.041 \\
   &  & 256 & -0.005 & 0.054 & -0.006 & 0.041 & -0.005 & 0.054 & -0.006 & 0.041 & -0.005 & 0.054 & -0.006 & 0.041 \\
   &  & 4096 & -0.005 & 0.054 & -0.006 & 0.041 & -0.005 & 0.054 & -0.006 & 0.041 & -0.005 & 0.054 & -0.006 & 0.041 \\
   & 0.05 & 0.125 & -0.008 & 0.072 & -0.013 & 0.056 & -0.008 & 0.072 & -0.013 & 0.056 & -0.006 & 0.072 & -0.012 & 0.056 \\
   &  & 256 & -0.008 & 0.072 & -0.013 & 0.056 & -0.008 & 0.072 & -0.013 & 0.056 & -0.006 & 0.072 & -0.012 & 0.056 \\
   &  & 4096 & -0.008 & 0.072 & -0.013 & 0.056 & -0.008 & 0.072 & -0.013 & 0.056 & -0.006 & 0.072 & -0.013 & 0.056 \\
   & 0.1 & 0.125 & -0.008 & 0.074 & -0.013 & 0.056 & -0.008 & 0.074 & -0.013 & 0.056 & -0.006 & 0.074 & -0.012 & 0.056 \\
   &  & 256 & -0.008 & 0.074 & -0.013 & 0.056 & -0.008 & 0.074 & -0.013 & 0.056 & -0.006 & 0.074 & -0.012 & 0.056 \\
   &  & 4096 & -0.008 & 0.074 & -0.013 & 0.056 & -0.008 & 0.074 & -0.013 & 0.056 & -0.006 & 0.074 & -0.012 & 0.056 \\
  0.3 & 0 & 0.125 & -0.007 & 0.060 & -0.009 & 0.043 & -0.007 & 0.060 & -0.009 & 0.043 & -0.007 & 0.060 & -0.009 & 0.043 \\
   &  & 256 & -0.007 & 0.060 & -0.009 & 0.043 & -0.007 & 0.060 & -0.009 & 0.043 & -0.007 & 0.060 & -0.009 & 0.043 \\
   &  & 4096 & -0.007 & 0.060 & -0.009 & 0.043 & -0.007 & 0.060 & -0.009 & 0.043 & -0.007 & 0.060 & -0.009 & 0.043 \\
   & 0.05 & 0.125 & -0.009 & 0.073 & -0.015 & 0.057 & -0.009 & 0.073 & -0.015 & 0.057 & -0.007 & 0.073 & -0.014 & 0.057 \\
   &  & 256 & -0.009 & 0.073 & -0.015 & 0.057 & -0.009 & 0.073 & -0.015 & 0.057 & -0.007 & 0.073 & -0.014 & 0.057 \\
   &  & 4096 & -0.009 & 0.073 & -0.015 & 0.057 & -0.009 & 0.073 & -0.015 & 0.057 & -0.008 & 0.073 & -0.014 & 0.057 \\
   & 0.1 & 0.125 & -0.009 & 0.074 & -0.014 & 0.057 & -0.009 & 0.075 & -0.014 & 0.057 & -0.006 & 0.075 & -0.013 & 0.057 \\
   &  & 256 & -0.009 & 0.074 & -0.014 & 0.057 & -0.009 & 0.075 & -0.014 & 0.057 & -0.006 & 0.075 & -0.013 & 0.057 \\
   &  & 4096 & -0.009 & 0.074 & -0.014 & 0.057 & -0.009 & 0.075 & -0.014 & 0.057 & -0.007 & 0.075 & -0.013 & 0.057 \\
\bottomrule
\end{tabular}
\begin{minipage}{0.98\textwidth}
\textit{Note.} FW and NA denote fixed-weight NA-GMM and original NA-GMM, respectively.
\end{minipage}
\end{table}

\clearpage
\bibliography{references}
\end{document}